%% file: paper_arxiv.tex
\documentclass[12pt]{article}
\usepackage[usenames]{color}
\usepackage{array}
\usepackage{graphicx}
\usepackage{amssymb}
\usepackage[ruled]{algorithm2e}
\usepackage{algpseudocode}
\usepackage{graphicx}
\usepackage{pdfpages}
\usepackage{soul}
\usepackage{amsmath}
\usepackage{amssymb}
\usepackage{amsthm}
\usepackage{mathrsfs}
\usepackage{cases}
\usepackage{setspace}
\usepackage{enumerate}
\usepackage{subfigure}
\usepackage{makecell}
\usepackage{float}
\usepackage{listings}
\usepackage{xcolor}
\usepackage{appendix}
\usepackage{pdfpages}
\usepackage{abstract}
\usepackage{longtable,booktabs}
\usepackage{rotating}
\usepackage{lscape}
\usepackage{multicol}
\usepackage{makecell}
\usepackage{natbib}
\usepackage{mathtools}
\usepackage{multirow}
\usepackage{threeparttable}
\usepackage{array}
\usepackage{fancyhdr}
\usepackage{pgfplots}
\usepackage[bookmarks=false,hypertexnames=false]{hyperref}\hypersetup{colorlinks=true, linkcolor=blue, filecolor=gray, urlcolor=blue, citecolor=blue}
\usepackage{caption}
\usepackage{authblk}

\newtheorem{Proposition}{Proposition}
\newtheorem{Assumption}{Assumption}
\newtheorem{Lemma}{Lemma}

\newtheorem{Theorem}{Theorem}

\theoremstyle{remark}
\newtheorem{Remark}{Remark}

\newcommand{\convd}{\stackrel{d}{\to}}
\newcommand{\convp}{\stackrel{p}{\to}}

\newcommand{\ic}{{i\cdot}}

\newcommand{\R}{\mathbb{R}}
\newcommand{\E}{{\mathbb{E}}}

\newcommand{\rank}{{\rm rank}} 
\newcommand{\RP}{{\rm P}}
\newcommand{\RM}{{\rm M}}
\newcommand{\RI}{{\rm I}}

\renewcommand{\hat}{\widehat}
\renewcommand{\tilde}{\widetilde}

\begin{document}

\title{Identification and Robust Inference for Multiple Treatments with Possibly Invalid Instruments\thanks{Ziwei Mei: \href{mailto:ziweimei@um.edu.mo}{ziweimei@um.edu.mo}, Faculty of Business Administration and Asia-Pacific Academy of Economics and Management, University of Macau; Qingliang Fan: \href{mailto:michaelqfan@cuhk.edu.hk}{michaelqfan@cuhk.edu.hk}, Department of Economics, The Chinese University of Hong Kong; Zijian Guo: \href{mailto:zijguo@zju.edu.cn}{zijguo@zju.edu.cn}, Center for Data Science, Zhejiang University. Mei acknowledges the partial financial support from
the Start-up Research Grant  (SRG, Project No.~SRG2025-00062-FBA) and the APAEM Seed Grant (Project No.~APAEM/SG/00004/2026)  from University of Macau.  }}

\author{Ziwei Mei$^a$, Qingliang Fan$^b$, and Zijian Guo$^c$ \\
$^a$Faculty of Business Administration, University of Macau\\ 
$^b$Department of Economics, The Chinese University of Hong Kong \\
$^c$Center for Data Science, Zhejiang University} 
\date{\today}
\maketitle
\vspace{-2em}
\begin{abstract}\footnotesize The instrumental variable (IV) method is widely used to infer causal effects in observational studies with unmeasured confounding, but invalid instruments can compromise both population identification and finite-sample inference. This paper studies linear IV models with multiple endogenous treatments and possibly invalid instruments. Identification of multiple effects is more delicate than in the single-treatment setting because a single instrument no longer identifies a single candidate effect; instead, each relevant instrument defines a hyperplane in the multidimensional effect space. For identification of multiple treatment effects, we introduce generalized plurality and majority rules which require a sufficiently large number of IVs to be valid. For inference, data-dependent instrument selection may fail to separate certain invalid IVs from valid ones, leading to undercoverage of confidence intervals when these invalid instruments are mistakenly selected as valid. We propose a sampling confidence interval for each treatment effect, which is robust to IV selection errors. We establish asymptotic coverage and parametric-rate length of our sampling confidence interval under regularity conditions and illustrate this method in Monte Carlo simulations and a Mendelian randomization application. 
\end{abstract}
 \clearpage
 \onehalfspacing

\section{Introduction}\label{intro}
\par Unmeasured confounders are prevalent in the study of treatment effects using observational data. The instrumental variable (IV) method is one of the most popular tools to tackle the endogeneity issue in causal inference. Ideally, conditioning on the observed covariates, IVs are required to satisfy the following conditions:  
\begin{itemize}
    \item[(A1)] the IVs are relevant to the treatments; 
    \item[(A2)] the IVs do not directly affect the outcome;
    \item[(A3)] the IVs are uncorrelated with the unmeasured confounders. 
\end{itemize}
\par Condition (A1) is the IV relevance condition. Conditions (A2) and (A3), often imposed conditional on the observed covariates, ensure that treatment is the only pathway from the IVs to the outcome. IVs that satisfy both (A2) and (A3) are called ``valid IVs''. These conditions are essential for identifying causal effects. However, either condition can be violated in observational studies. When instruments violate (A2), they have direct effects on the outcome in addition to their effects through the treatments. When instruments violate (A3), they are correlated with unobserved confounders omitted from the outcome model. We define an IV as ``invalid'' if it violates either Condition (A2) or (A3). If an invalid IV is mistakenly treated as valid, causal effects are generally not correctly identified. 

Causal studies with multiple treatments are common. Examples include studies of multiple risk factors, including smoking, high cholesterol, diabetes, kidney disease, inactivity, and obesity, for cardiovascular disease \citep{burgess2015multi}, as well as education value-added models with multiple policy interventions \citep{goldsmith24contamination}. Another example is the multivariable Mendelian randomization, a widely used approach in epidemiology \citep{zhou2020mendelian} that involves multiple exposures and requires methods that can accommodate invalid genetic markers. Despite a growing literature on IV methods for a single treatment, identification and robust inference for multiple treatments with possibly invalid instruments remain less developed.  This paper addresses two gaps in this setting: identification conditions for multidimensional causal effects (Section \ref{sec: identification}) and an inference procedure that is robust to IV selection errors (Section \ref{sec: uniform inference}).   

\subsection{Causal Identification}\label{subsec: intro identification}

Causal identification is the first-order issue in IV analysis with potentially invalid instruments. A rapidly growing literature, primarily focused on the single-treatment model, has developed identification conditions for treatment effects in the presence of invalid instruments. Representative conditions include the \emph{majority rule} \citep{han2008,kang2016instrumental,windmeijer2018use} requiring that more than half of the IVs are valid, and the \emph{plurality rule} \citep{guo2018confidence,windmeijer2021confidence} requiring valid IVs to form a sufficiently large group. For settings with multiple treatments, however, these scalar rules do not directly identify the true multidimensional effect when invalid IVs are present.  

We now explain why the usual majority rule is not enough in the multiple-treatment regime. Consider a linear IV model with $p_z$ instruments and $p_d$ treatments, and assume all instruments are relevant. When $p_d=1$, each instrument identifies a scalar candidate effect: valid instruments yield the true effect, whereas invalid instruments yield biased values. If more than half of the instruments are valid, the true effect can be recovered as the majority of the instrument-specific candidate effects. When $p_d>1$, a single instrument cannot identify the full vector of treatment effects; instead, at least $p_d$ instruments are required. Geometrically, each instrument defines a $(p_d-1)$-dimensional hyperplane in the $p_d$-dimensional effect space. Thus, for $p_d>1$, an instrument-specific identified set is no longer a singleton but a continuum of candidate values. Consequently, even if a strict majority of instruments are valid, there is no immediate ``majority vote'' over scalar candidates, and the single-treatment majority or plurality logic does not deliver identification without further model restrictions.

This paper develops identification results tailored to multiple treatments. Our main population condition is the \textit{generalized plurality rule}. As illustrated above, each relevant IV defines a $(p_d-1)$-dimensional hyperplane in the $p_d$-dimensional effect space. The generalized plurality rule requires that, among all possible effect values, the true value lies on strictly more of these hyperplanes than any false value. This geometric formulation makes the identification problem transparent and complements earlier algebraic formulations of related conditions. We also establish a sufficient condition, the \textit{generalized majority rule}. This rule shows that identification is guaranteed when the number of valid IVs exceeds the lower bound in (\ref{eq: general major}). When $p_d = 1$, it reduces to the familiar majority rule requiring more than half of the relevant IVs to be valid. 

\subsection{Robust Inference} \label{intro: inference} 
In observational studies, valid inference for endogenous effects hinges on  correct selection of valid IVs. In the single-treatment regime, related literature \citep{guo2018confidence,kang2016instrumental,windmeijer2018use,windmeijer2021confidence} proposes selecting valid instruments through hard thresholding or penalized regressions. It is a natural idea to generalize these methods to accommodate multiple treatments. 

However, inference for causal effects following IV selection can be biased if invalid IVs are mistakenly classified as valid. Specifically, selection of valid IVs involves solving a system of linear equations and retaining solutions that satisfy the identification condition (e.g.~the generalized majority rule). In observational data, the coefficients in this system are estimated with finite-sample error. 
Consequently, if the violation of (A2) or (A3) by certain instruments is local to zero, standard selection methods can struggle to distinguish them from valid IVs. We refer to these invalid IVs as \textit{locally invalid IVs}, with a formal description in (\ref{eq: locally invalid}) of Section \ref{subsec: local invalid}. The inclusion of these locally invalid IVs in the selected set can lead to significant size distortion in subsequent causal inference.

To address the aforementioned size distortion, we propose a novel inference method that is robust to errors in IV selection. Our approach accounts for finite-sample estimation error by injecting randomness into solving the aforementioned linear equations and determining which IVs are valid. In particular, we perturb the solutions to the linear equations and select a valid-IV set for each perturbation. We show that, with sufficiently many perturbations, at least one selected set contains all truly valid IVs with high probability. We use each selected IV set to construct a confidence interval for the causal effect using the two-stage least squares (TSLS) estimator. We then aggregate these intervals and show that the resulting confidence set attains nominal coverage. We further prove that its length converges at the parametric rate $n^{-1/2}$. 
The resulting procedure provides inference for multiple treatment effects that is robust to IV selection error without requiring prior knowledge of which IVs are valid. We exhibit the usefulness of the method through Monte Carlo simulations (Section \ref{sec: simul}) and a Mendelian randomization study (Section \ref{sec: real data}).

\subsection{Additional Related Literature}
In the multiple-treatment setting, \citet{liang2022instrumental} proved that treatment effects are identified when more than $(p_z+p_d-1)/2$ instruments are valid, under an additional restriction on IV relevance. This result is a special case of our generalized majority rule in Section \ref{subsec: how many}. \citet{apfel2021agglomerative} also discussed identification under multiple treatments, but their identification condition requires the additional restriction that every subset of $p_d$ IVs identifies a unique candidate effect. In contrast, our generalized plurality rule does not impose this restriction. Under a single treatment, our generalized plurality rule reduces to the \textit{sparsest rule} developed in \citet{lin2024instrumental}. 

In terms of inference robust to selection error, a recent work by \citet{guo2023causal} proposed an inference procedure based on searching and sampling focusing on a single treatment. Specifically, \citet{guo2023causal} searches over candidate effects on the one-dimensional real line and collects values of $\beta$ that suggest a large proportion of valid IVs. A direct grid-search generalization to $p_d$ dimensions is computationally costly: for example, 500 grid points in each coordinate would require $500^{p_d}$ evaluations. Rather than searching over the entire effect space, our proposal searches over TSLS estimators generated by subsets of $p_d$ relevant IVs. When a valid just-identifying IV subset exists, at least one of these estimators is centered at the truth. Thus, unlike a direct extension of \citet{guo2023causal}, our method avoids a multidimensional grid search.

\par We further review other literature on invalid IVs. \citet{tchetgen2021genius} proposed the GENIUS method for robust inference in Mendelian randomization, whose identification draws from heteroskedasticity \citep{lewbel2012using}. \citet{hu2022mendelian} proposed MR-APSS to deal with invalid IVs at the summary level by leveraging genome-wide information. \citet{lin2024instrumental} and \citet{ye2024genius} examine robust inference for a single endogenous treatment with many weak and invalid instruments. Other studies on inferring a single treatment effect with possibly invalid IVs include  \citet{sun2022selective,sun2023semiparametric}, \citet{fan2024endogenous}, and \citet{dukes2025using}. Unlike these studies focusing on a single treatment, our paper considers unknown IV validity in a multiple-treatment framework. As for estimation with potentially invalid moment conditions in generalized method of moments (GMM),  \citet{andrews1999consistent}  developed consistent moment-selection procedures using information criteria. \citet{liao2013adaptive} and \citet{cheng2015select} used adaptive LASSO to select invalid moment conditions in GMM with a known set of valid moments, different from our regime that does not require the known subset of valid IVs. These studies on GMM neither derive an explicit lower bound on the number of valid IVs sufficient for identifying multiple treatment effects as provided by our generalized majority rule, nor consider robust inference under locally invalid IVs. Therefore, their primary focus differs from ours. See \citet{kang2024identification} for a more detailed review of related literature. 
\medskip 
\par \noindent \textbf{Structure}. The rest of the paper is organized as follows. Section \ref{sec: identification} studies causal identification under the regime of multiple treatments. Section \ref{sec: uniform inference} illustrates IV selection error with locally invalid IVs and proposes an inference procedure robust to such error. Section \ref{subsec: theory} establishes asymptotic theory to justify the validity of our inference procedure. Section \ref{sec: simul} displays numerical results and Section \ref{sec: real data} presents a real data application of Mendelian randomization. 
Section \ref{sec: conclusion} concludes the paper.  

\medskip
\par \noindent \textbf{Notations}. We use ``$\convp$" and ``$\convd$" to denote convergence in probability and distribution, respectively. We use $[n]$ for some natural number $n$ to denote the integer set $\{1,2,\cdots,n\}$, and use $n!$ to denote the factorial. We use $\RI_n$ to denote the $n$-dimensional identity matrix, and ${\rm \bf 1}\{\cdot\}$ to denote the indicator function.  For any  positive integers $m>n$, let $\left(\begin{array}{c}
     m  \\
     n 
\end{array}\right)$ be the binomial coefficient $\frac{m!}{(m-n)!n!}$. For a $p$-dimensional vector $x=(x_{1},x_{2},\cdots,x_{p})^{\top}$, the $L_{2}$ norm is $\left\Vert x\right\Vert _{2}=\sqrt{\sum_{j=1}^{p}x_{j}^{2}}$, the $L_{1}$ norm is $\left\Vert x\right\Vert_{1}=\sum_{j=1}^{p}\left|x_{j}\right|$,  and the maximum norm is $\|x\|_{\infty}=\max_{j\in[p]}|x_{j}|$. We use $\lambda_{\min}(A)$ and $\lambda_{\max}(A)$ to denote the minimum and maximum eigenvalues of any symmetric matrix $A$. For any  $m\times n$ matrix $A=(A_{jk})_{j\in[m],k\in[n]}$, the spectral norm is $\|A\|_2 = \sqrt{\lambda_{\max}(A^\top A)}$ and the maximum norm is $\|A\|_\infty = \max_{j\in[m],k\in[n]}|A_{jk}|$.  Matrix vectorization is denoted by ${\rm vec}(A)$. For any sets $\mathcal{A}$ and $\mathcal{B}$, $|\mathcal{A}|$ denotes the cardinality of $\mathcal{A}$, and $\mathcal{A}\backslash\mathcal{B} = \{x\in \mathcal{A}:x\notin\mathcal{B}\}$. 

\section{The Model and Effect Identification}\label{sec: identification}
For each individual $i\in \{1,2,\dots,n\}$, we observe an outcome $Y_i\in\R$, a multivariate treatment vector $D_\ic = (D_{i,1},D_{i,2},\dots,D_{i,p_d})^\top\in\R^{p_d}$, the instrumental variables $Z_\ic= (Z_{i,1},Z_{i,2},\dots,Z_{i,p_z})^\top\in\R^{p_z}$, and the covariates $X_\ic = (X_{i,1},X_{i,2},\dots,X_{i,p_x})^\top\in\R^{p_x}$. We assume that $\{Y_i,D_\ic^\top,Z_\ic^\top,X_\ic^\top\}_{i=1}^n$ are generated in an independently and identically distributed (i.i.d.) fashion. Let $Y_i^{(d,z,x)}$ be the potential outcome if the individual $i$ had treatment $d \in \R^{p_d}$, instruments $z\in \R^{p_z}$, and covariates $x\in\R^{p_x}$. In practice, we observe only one outcome, $Y_i=Y_i^{(D_\ic,Z_\ic,X_\ic)}$, for each observation. 
Throughout the paper, we assume that $p_z \geq p_d \geq 1$ and all variable dimensions $p_d$, $p_z$, and  $p_x$ are fixed. 
\par We consider the additive linear potential outcome model  \citep{small2007sensitivity,kang2016instrumental,guo2018confidence,windmeijer2018use}, and allow for multiple treatments and possibly invalid instruments. For two possible sets of values of the exposures $d$, $d^\prime$,  instruments $z$, $z^\prime$, and covariates $x$, $x^\prime$, we assume the potential outcome model
\begin{equation}\label{eq: model}
\begin{aligned}
    Y_i^{(d,z,x)} - Y_i^{(d^\prime,z^\prime,x^\prime)} &= (d-d^\prime)^\top\beta^* + (x-x^\prime)^\top \xi^*+(z-z^\prime)^\top \kappa^*,\\
    \E(Y_i^{(0,0,0)}|Z_\ic,X_\ic)&= Z_\ic^\top \eta^* + X_\ic^\top \zeta^*,
\end{aligned}
\end{equation}
where $\beta^*\in\R^{p_d}$ denotes the true causal effect, while  $\kappa^*,\eta^*\in\R^{p_z}$ and $\xi^*,\zeta^*\in\R^{p_x}$ are other unknown parameters. In particular, $\beta_j^*$, the $j$-th entry of the vector $\beta^*$, measures the population causal effect of the $j$-th treatment $D_{\cdot j}$ on the potential outcome.

The parameter $\kappa^*$ measures the direct effect of instrumental variables on the outcome, and $\kappa^*=0$ means the exclusion condition (A2) is satisfied. The parameter $\eta^*$ represents the presence of unmeasured confounding between the instruments and the outcome, and $\eta^*=0$ indicates that the condition of no unmeasured confounder (A3) is satisfied. Lastly, the parameters $\xi^*$ and $\zeta^*$ measure the effects of the covariates $X_\ic$ on the potential outcome. 
\par Define $u_i = Y_i^{(0,0,0)}-\E(Y_i^{(0,0,0)}|Z_\ic,X_\ic)$. Applying the model (\ref{eq: model}) together with the expression $Y_i^{(0,0,0)}=u_i +\E(Y_i^{(0,0,0)}|Z_\ic,X_\ic)$, we have the following main model of multiple treatments with observational data:
\begin{equation}\label{eq: DGP y}
\begin{aligned}
    Y_i = D_\ic^\top\beta^* + Z_\ic^\top\pi^* + X_\ic^\top \varphi^* + u_i, \quad \text{with}\quad 
    \E(X_\ic u_i) = 0,\quad\E(Z_\ic u_i) = 0,
\end{aligned}
\end{equation}
where $\pi^*=\kappa^* + \eta^*$, and $\varphi^*=\xi^* + \zeta^*$. We allow $D_\ic\in\mathbb{R}^{p_d}$ to be endogenous due to its correlation with the unmeasured error term $u_i\in\mathbb{R}$. The $k$-th instrumental variable $Z_{i,k}$ is valid if $\pi_k^*=0$, for $k\in[p_z]$. Since $\pi^*=\kappa^*+\eta^*$, the IV validity measured by $\pi^*$ consists of two violations:  the direct effect of IVs on the outcome measured by $\kappa^*$, and the confounding between IVs and the outcome represented by $\eta^*$.  
\par To measure the relevance between the IVs and the treatments stated by Condition (A1) in Section \ref{intro}, we further assume the following linear reduced-form model for the treatment
\begin{equation}\label{eq: DGP d}
\begin{aligned}
    D_\ic^\top =  Z_\ic^\top\Upsilon^* + X_\ic^\top \Psi^* + \varepsilon_\ic^\top\quad \text{with}\quad 
    \E(X_\ic \varepsilon_\ic^\top) = 0,\quad\E(Z_\ic \varepsilon_\ic^\top) = 0.
\end{aligned}
\end{equation} 
Model (\ref{eq: DGP d}) can be viewed as the best linear approximation of $D_\ic$ by the instruments $Z_\ic$ and covariates $X_\ic$. The \emph{relevance matrix} $\Upsilon^*=(\Upsilon^*_{k,j})_{k\in[p_z],j\in[p_d]}$ measures the relevance between the treatments and IVs. For any subset $\mathcal{A}\subset [p_z]$, we use $\Upsilon^*_{\mathcal{A}\cdot}$ to denote the submatrix of $\Upsilon^*$ consisting of the rows that correspond to the subset $\mathcal{A}$. Define the set of IVs relevant to at least one treatment  
as
\begin{equation}\label{eq: strong}
    \mathcal{S}^* = \{k\in[p_z]:\|\Upsilon^*_{ k\cdot}\|_2 > 0\},
\end{equation}
and refer to them as \textit{relevant IVs} for simplicity. Accordingly, $\mathcal{S}^{*c}$ contains the irrelevant IVs that have a true coefficient of zero in the reduced-form model (\ref{eq: DGP d}) for all the treatments. Furthermore, define the set of valid IVs as $$\mathcal{V}^* = \{k\in \mathcal{S}^*:\pi^*_k = 0\}.$$ For identification, we further impose the following \textit{full rank condition} on the valid IVs.
\begin{Assumption}\label{cond: full rank}Suppose that $\rank(\Upsilon^*_{\mathcal{V}^*\cdot}) = p_d$. 
\end{Assumption}
\par Assumption \ref{cond: full rank} imposes a full column rank on the submatrix of $\Upsilon^*$ composed of the rows indexed by the valid IV set $\mathcal{V}^*$, which means that using the true set of valid instruments enables the identification of the true effect $\beta^*$.  Without this condition, identification of the true effect may fail.  For example, identification fails when $p_d>1$ if all valid IVs are relevant to the first treatment but irrelevant to all other treatments. In this example, $\rank(\Upsilon^*_{\mathcal{V}^*\cdot}) = 1$ and Assumption \ref{cond: full rank} is violated. 
\par In the remaining parts of Section \ref{sec: identification}, we provide detailed interpretations and discussions of the generalized plurality rule for a general $p_d \geq 1$. Section \ref{subsec: geometric} formulates the generalized plurality rule for identification of multiple treatment effects with a geometric interpretation under the two-treatment case. Section \ref{subsec: how many} discusses a sufficient condition for identification called the generalized majority rule, stating that identification success is always guaranteed when the number of valid IVs exceeds a lower bound specified in (\ref{eq: general major}).

\subsection{Generalized Plurality Rule}\label{subsec: geometric}
\par We formulate the generalized plurality rule in the following and elaborate why it enables causal identification when there are possibly invalid IVs. Substituting model (\ref{eq: DGP d}) into (\ref{eq: DGP y}) yields the reduced-form model 
\begin{equation}
    \begin{aligned}\label{eq: reduced form}
    Y_i &= Z_\ic^\top\Gamma^* + X_\ic^\top\Phi^* + e_i^*,\ \ 
    D_\ic^\top = Z_\ic^\top\Upsilon^* + X_\ic^\top\Psi^* + \varepsilon_\ic^\top, 
\end{aligned}
\end{equation} 
where $e_i^* = u_i + \varepsilon_\ic^\top \beta^* \in \mathbb{R}$, $\Phi^* = \Psi^*\beta^* + \varphi^* \in \mathbb{R}^{p_x}$ and 
\begin{equation}\label{eq: Gamma}
    \Gamma^* = \Upsilon^*\beta^* + \pi^* \in \mathbb{R}^{p_z}.
\end{equation}
Note that the regressors on the right-hand side of both equations in (\ref{eq: reduced form}) only include the instruments and covariates, without the endogenous treatments. Consequently, we can consistently estimate the reduced form coefficients $\Gamma^*$ and $\Upsilon^*$ using ordinary least squares (OLS).

In the remainder of this section, we assume $\Gamma^*$ and $\Upsilon^*$ are known to simplify the discussion and emphasize the population identification conditions for $\beta^*$. We will provide a fully data-dependent procedure in Section \ref{sec: uniform inference} where we shall estimate all model parameters, namely $\Gamma^*$, $\Upsilon^*$, $\beta^*$ and $\pi^*$. 

We emphasize that the true $\beta^*$ is not identified by the system of equations in \eqref{eq: Gamma} alone, as it involves $p_d+p_z$ unknown parameters in $\beta^*$ and $\pi^*$, but provides only $p_z$ equations. The identification needs some extra conditions on the IV validity. In conventional IV model studies, all instruments are assumed to be valid with $\pi^* = 0$, and thus $\Gamma^* - \Upsilon^* \beta^* = 0 $ according to (\ref{eq: Gamma}). Identification of $\beta^*$ is therefore viewed as searching for a candidate $\beta$ satisfying
\begin{equation} \label{eq: pi beta eq system}
\pi(\beta) = 0 \quad \text{with}\quad \pi(\beta) = \Gamma^* - \Upsilon^*\beta.
\end{equation}
The $k$-th entry of $\pi(\beta)$ is given by 
\begin{equation}\label{eq: pi beta k}
    \pi_k(\beta) = \Gamma^*_k - {(\Upsilon^*_{k\cdot})^{\top}}\beta =  - \sum_{j\in[p_d]} \Upsilon^*_{k,j}(\beta_j-\beta^*_j) + \pi^*_k.
\end{equation}
Componentwise, the $k$-th IV yields the equation  $\pi_k(\beta)=0$ and identifies all solutions to this equation. For illustration, we say that the $k$-th IV \textit{votes for} the candidate $\beta$ if $\pi_k(\beta) = 0$. 

All valid IVs unanimously vote for the true $\beta^*$. By contrast, any invalid IV with $\pi^*_k\neq 0$ does not vote for the true $\beta^*$.
Practically, researchers tend to select candidate IVs using domain knowledge. We hence assume that valid IVs are sufficiently prevalent for the true $\beta^*$ to receive more votes than any other candidate value. We formulate this assumption as the following \textit{generalized plurality rule}.
\begin{Assumption}[Generalized Plurality Rule]\label{cond: plurality new}
\begin{equation}\label{eq: plurality}
    |\mathcal{V}^*| > \max_{\beta\neq\beta^*}|\mathcal{K}(\beta)|,\ \text{where } \mathcal{K}(\beta)= \left\{k\in\mathcal{S}^*:\pi_k(\beta) = 0\right\}.
\end{equation}
\end{Assumption}
By definition, $\mathcal{K}(\beta)$ collects the relevant IVs that vote for the candidate $\beta$, and thus $|\mathcal{K}(\beta)|$ is the number of relevant IVs voting for $\beta$. In \eqref{eq: plurality}, the left-hand side  of the inequality $|\mathcal{V}^*|$ is the number of truly valid IVs that vote for the true $\beta^*$. The right-hand side of the inequality in \eqref{eq: plurality} is the maximum number of relevant IVs voting for the same false candidate $\beta\neq\beta^*$. Therefore, (\ref{eq: plurality}) means that the true $\beta^*$ is the unique winner of the voting, and we can distinguish the true effect $\beta^*$ from all other $\beta\in\mathbb{R}^{p_d}$ by searching for the unique candidate that wins the largest number of votes.  

\begin{Remark} The voting does not involve irrelevant IVs since they make no difference to the voting result. Note that when an IV is irrelevant with $\Upsilon_{k\cdot}^*=0$, we have $\pi_k(\beta
) = \pi_k^*$ according to (\ref{eq: pi beta k}). Therefore, it either votes for all candidates when $\pi_k^*=0$, or votes for no candidate when   $\pi_k^*\neq 0$. For inference in practice, we will remove the irrelevant IVs in the first stage before selecting valid IVs. See Section \ref{subsec: first stage} for details.  
\end{Remark}

We now elaborate the difference between single-treatment and multiple-treatment settings. When $p_d = 1$, the relevance matrix $\Upsilon^*=(\Upsilon^*_1,\Upsilon^*_2,\cdots,\Upsilon^*_{p_z})^\top$ becomes a vector. The valid IVs with $\pi_k^*=0$ only vote for the true $\beta^*$, and the invalid IVs with the same value of $\pi_k^*/\Upsilon_k^*$ only vote for the same false candidate $\beta^*+\pi_k^*/\Upsilon_k^*$. Hence, we can easily classify the IVs into the same group if they vote for the same candidate $\beta$, and the generalized plurality rule (\ref{eq: plurality}) reduces to the plurality rule in \citet{guo2018confidence} stating that the valid IVs constitute the largest group. 

\par In sharp contrast to single-treatment models where each relevant IV votes for a singleton, under the multiple-treatment framework each relevant IV votes for all candidates on a $(p_d-1)$-dimensional hyperplane in the $p_d$-dimensional Euclidean space. According to (\ref{eq: pi beta k}), the $k$-th IV votes for all candidates in the following hyperplane in $\mathbb{R}^{p_d}$
\begin{equation}\label{eq: hyperplane}
    \mathcal{L}_k := \left\{\beta = (\beta_j)_{j\in[p_d]}:  \sum_{j\in[p_d]} \Upsilon^*_{k,j}(\beta_j -\beta^*_j )  =  \pi^*_k\right\}.
\end{equation}
Under a single treatment, $\mathcal{L}_k$ degenerates to a singleton on the real line. For bivariate treatments, $\mathcal{L}_k$ represents a straight line. For trivariate treatments, $\mathcal{L}_k$ becomes a two-dimensional plane, and so forth. 

For a clearer picture of Assumption \ref{cond: plurality new}, in Figure \ref{fig: thm} we provide two graphical examples with $p_d = 2$ and $p_z = 5$, where all IVs are relevant and thus the straight lines $\{\mathcal{L}_k\}_{k=1}^5$ specified in (\ref{eq: hyperplane}) are well defined. If  a candidate effect $\beta$ falls on a straight line, it means that the corresponding IV votes for $\beta$. The straight lines for valid instruments are sketched by blue lines, and those for invalid instruments are represented by red lines. In both panels of Figure \ref{fig: thm}, three of the five instruments are valid. 
All blue lines in Figure \ref{fig: thm} representing valid IVs go across the true effect $\beta^*$, meaning that all valid IVs vote for the true $\beta^*$. In contrast, none of the red lines representing invalid IVs goes across $\beta^*$, meaning that no invalid IVs vote for the true effect. 
\input{fig/fig_thm} 

We now discuss the validity of the generalized plurality rule in the two cases of Figure \ref{fig: thm}. Panel (a) in Figure \ref{fig: thm} gives an example where Assumption \ref{cond: plurality new} is satisfied. Among the intersection points of the straight lines, only the true effect $\beta^*$ receives three votes. All other points receive at most two votes. For example, the false candidates $\beta^{(1)}$ and $\beta^{(2)}$ receive two votes. Therefore, the true effect $\beta^*$ with three votes is the unique winner in the voting. Even without prior knowledge of instrument validity in practice (i.e., the colors of the straight lines are unknown), we can identify $\beta^*$ as the unique winner.
\par Panel (b) in Figure \ref{fig: thm} gives a counterexample where Assumption \ref{cond: plurality new} is violated. Again, we highlight three candidates here. The false candidate $\beta^{(1)}$ receives two votes, fewer than the three votes received by the true effect $\beta^*$. However, two red lines and one blue line intersect at another false candidate effect $\beta^{(2)}$, resulting in 3 votes for this candidate. The generalized plurality rule is violated so that we cannot determine whether $\beta^*$ or $\beta^{(2)}$ is the true effect, leading to identification failure.

Although Panel (b) illustrates a failure in effect identification, we emphasize that such identification failure is rare. If we fix the blue lines representing the hyperplanes $\mathcal{L}_k$ for $k \in \mathcal{V}^*$, identification failure occurs only if the intersection of two red lines (for example, the $\beta^{(2)}$ in Panel (b)) coincidentally falls on one of the blue lines. Section \ref{subsec: How often identification failure} formalizes this rarity under a random-coefficient model.

\subsubsection{How Often Does Identification Fail?}\label{subsec: How often identification failure}
\par To formulate the identification failure problem, in this Section \ref{subsec: How often identification failure} only, we assume that all nonzero coefficients in the relevance matrix $\Upsilon^*$ and the vector $\pi^*$ measuring IV validity are randomly generated from continuous distributions. We use $\mathcal{I}^* := \mathcal{S}^* \backslash \mathcal{V}^*$ to denote the index set of invalid instruments. 
\begin{Proposition}\label{prop: rare}Suppose that all entries in the vector of coefficients $({\rm vec}(\Upsilon^*_{\mathcal{S}^*\cdot})^\top,\pi_{\mathcal{I}^*}^{*\top})^\top$ are drawn from independent continuous distributions with supports on the whole $\mathbb{R}^{|\mathcal{S}^*|\times p_d+|\mathcal{I}^*|}$ space, while the subvector $\pi^*_{\mathcal{V}^*} = 0_{|\mathcal{V}^*|}$ and the true effect $\beta^*$ are deterministic. Then
\begin{equation}\label{eq: max vote pd}
    \Pr\left\{  \max_{\beta\neq\beta^*}\left|\left\{k\in\mathcal{S}^*:\pi_k(\beta) = 0\right\}\right| \leq p_d \right\} = 1. 
\end{equation}
Thus, if $|\mathcal{V}^*| > p_d$, we have 
$\Pr\left\{ \text{Assumption \ref{cond: plurality new} holds}\right\} = 1$. 
\end{Proposition}
\par Equation (\ref{eq: max vote pd}) means that with probability one, all false candidate effects receive no more than $p_d$ votes. Therefore, when the true model is overidentified with $|\mathcal{V}^*| > p_d$, the true $\beta^*$ is the unique winner in the voting with probability one. Hence, Assumption \ref{cond: plurality new} holds with probability one and identification failure in population is a zero-probability event under the random coefficient framework. 

When coefficients are nonrandom, how many valid IVs are sufficient for identifying multiple treatment effects? The following Section \ref{subsec: how many} gives a confirmative answer to this question with a generalized majority rule.

\subsection{Generalized Majority Rule}\label{subsec: how many} 
\par Under multiple treatments, the \textit{majority rule} \citep{kang2016instrumental}
\begin{equation} \label{eq: literal major}
    |\mathcal{V}^*| > |\mathcal{S}^*| / 2
\end{equation} is not sufficient for identification. Figure \ref{fig: thm}(b) provides a counterexample: three out of five instruments are valid and thus the majority rule (\ref{eq: literal major}) holds, but the identification of $\beta^*$ fails. 
\par We now discuss the sufficient condition for effect identification under multiple treatments. A natural sufficient condition is that $|\mathcal{V}^*|$ is larger than the maximum number of votes a false candidate $\beta$ can receive. The complexity rises from the fact that both valid IVs and invalid ones can vote for false candidates. In the worst case, all $(|\mathcal{S}^*| - |\mathcal{V}^*|)$ relevant but invalid instruments vote for the same false candidate $\beta \neq \beta^*$. To quantify valid IVs that vote for the same false candidate, let $\mathscr{D} = \{\mathcal{H}\subset \mathcal{V}^*: {\rm rank}(\Upsilon^*_{\mathcal{H}\cdot}) < p_d\}$ collect the valid IV subsets with a rank-deficient IV relevance matrix. Furthermore, define the integer $h_0$ as
\begin{equation}\label{def: h}
    h_0 := 1 + \max_{ \mathcal{H} \in \mathscr{D} } |\mathcal{H}|.
\end{equation} 
By definition, $h_0-1$ is the maximum cardinality of a subset of valid IVs whose relevance submatrix is rank deficient. Thus, the relevance submatrix associated with any $h_0$ valid IVs has full column rank. Geometrically, the corresponding hyperplanes in \eqref{eq: hyperplane} have a unique common intersection. Equivalently, any $h_0$ valid IVs can only jointly vote for one candidate. This unique  candidate must be the true $\beta^*$ because all valid IVs vote for the true value. Consequently, no false candidate can receive votes from $h_0$ or more valid IVs; that is, it can receive at most $h_0-1$ votes from valid IVs.

Therefore, a sufficient condition for the true $\beta^*$ to receive strictly more votes than any false candidate is
\begin{equation}\label{eq: max i + max v}
    \underbrace{|\mathcal{V}^*|}_{\text{votes for the true $\beta^*$}}
    >
    \underbrace{|\mathcal{S}^*|-|\mathcal{V}^*|}_{
        \substack{\text{maximum possible votes for $\beta\neq\beta^*$}\\
        \text{from invalid IVs}}
    }
    +
    \underbrace{h_0-1}_{
        \substack{\text{maximum possible votes for $\beta\neq\beta^*$}\\
        \text{from valid IVs}}
    }.
\end{equation}
Rearranging \eqref{eq: max i + max v} yields the following generalized majority rule.

\begin{Theorem}[Generalized Majority Rule]\label{thm: plur}
Suppose that Assumption \ref{cond: full rank} holds. Then Assumption \ref{cond: plurality new} holds if
\begin{equation}\label{eq: general major}
    |\mathcal{V}^*|
    >
    \frac{|\mathcal{S}^*|+h_0-1}{2}.
\end{equation}
\end{Theorem}

The integer $h_0$ captures the rank structure of the relevance matrix and plays a crucial role in the generalized majority rule. A larger $h_0$ reflects a weaker rank structure, but correspondingly requires more valid IVs in \eqref{eq: general major}. By construction, we have 
\[
    p_d\leq h_0\leq|\mathcal{V}^*|,
\]
because at least $p_d$ rows are required for a $p_d$-column matrix to be full-column-rank, whereas the relevance matrix associated with the full set of valid IVs has rank $p_d$ by Assumption \ref{cond: full rank}. The least restrictive case of \eqref{eq: general major} occurs when $h_0$ attains its smallest possible value  $p_d$. In this case, the generalized majority rule reduces to
\begin{equation}\label{eq: general major pd}
    |\mathcal{V}^*|
    >
    \frac{|\mathcal{S}^*|+p_d-1}{2}.
\end{equation}
When all IVs are relevant, so that $|\mathcal{S}^*|=p_z$, \eqref{eq: general major pd} becomes $
    |\mathcal{V}^*|
    >
    \frac{p_z+p_d-1}{2},$ 
which coincides with the condition developed by \citet{liang2022instrumental}. In particular, when $p_d=1$, every relevant IV individually yields a rank-one relevance matrix, so that $h_0=1$. The generalized majority rule then reduces to the conventional majority rule $
    |\mathcal{V}^*|> |\mathcal{S}^*|/2.$


\section{Robust Inference 
Accounting for IV Selection Error}\label{sec: uniform inference}

\par Section \ref{sec: identification} has established the identification conditions assuming the parameters in the reduced-form model (\ref{eq: reduced form}) are given. We now switch our focus to statistical  inference for treatment effects $\beta^*$ with observational data, where the unknown coefficients $\Gamma^*$, $\Upsilon^*$, $\Phi^*$, and $\Psi^*$ in (\ref{eq: reduced form}) are estimated.  In Section \ref{subsec: first stage}, we devise a data-dependent way to select the relevant IVs, and the subsets comprising $p_d$ relevant IVs that can identify a unique value of $\beta$. Following the tradition of IV models, we refer to these subsets as \textit{just-identifying IV subsets}. Section \ref{subsec: TSHT} states a valid IV selection procedure with hard thresholding. Section \ref{subsec: local invalid} illustrates that in practice, certain invalid IVs cannot be easily separated from valid ones since their violations of validity, measured by $\pi^*$, are not sufficiently different from zero in finite samples. We refer to these IVs as \textit{locally invalid IVs} with a formal discussion in Section \ref{subsec: local invalid}.  Therefore, IV selection methods may mistakenly include these invalid IVs as valid ones, causing a size distortion in the inference for the treatment effects using the selected IVs. This issue motivates us to provide a robust inference procedure of the multiple treatment effects allowing for locally invalid instruments, with the full details provided in Section \ref{subsec: uniform} that focuses on the use of generalized majority rule. Section \ref{subsec: not use plurality} compares the use of generalized plurality rule and generalized majority rule.

We fix some notations before diving into the procedures. We collect the i.i.d.\ observational data introduced at the beginning of Section \ref{sec: identification} into matrices by $Y=(Y_1,\cdots,Y_n)^\top \in \mathbb{R}^{n}$, $D = (D_{1\cdot},\cdots,D_{n\cdot})^\top \in \mathbb{R}^{n\times p_d}$, $X = (X_{1\cdot},\cdots,X_{n\cdot})^\top \in \mathbb{R}^{n\times p_x}$ and $Z = (Z_{1\cdot},\cdots,Z_{n\cdot})^\top \in \mathbb{R}^{n\times p_z}$. Define the collection of IVs and  covariates as $W_\ic=(
Z_\ic^\top,X_\ic^\top)^\top$ and $W = (W_{1\cdot},\cdots,W_{n\cdot})^\top \in \mathbb{R}^{n\times p}$, where $p = p_x + p_z$. Furthermore, we define the projection matrix for any full-column-rank  matrix $A$ with $n$ rows as $\RP_A = A(A^\top A)^{-1} A^{\top}$, and $\RM_A = \RI_n - \RP_A$.

\par 
To prepare for the inference of treatment effect, we define the TSLS estimator using an arbitrary subset of IVs indexed by $\mathcal{H}$ as
\begin{equation}\label{eq: def generic TSLS}
\hat\beta^{(\mathcal{H})} = \left[\hat\theta^{(\mathcal{H})}\right]_{1:p_d}\text{ where }\hat\theta^{(\mathcal{H})} = (D^{(\mathcal{H})\top} \RP_W D^{(\mathcal{H})})^{-1}D^{(\mathcal{H})\top} \RP_W Y. 
\end{equation}
In (\ref{eq: def generic TSLS}), $D^{(\mathcal{H})} = (D,Z_{\cdot [p_z]\backslash\mathcal{H}},X)$ collects the treatments, the IVs not in $\mathcal{H}$, and the covariates. Equivalently, $\hat\beta^{(\mathcal{H})}$ is the estimated coefficients of $D$ in the second-stage  regression of $Y$ on $(\hat D,Z_{\cdot[p_z]\backslash\mathcal{H}},X)$, where $\hat D = \RP_W D$ is the fitted value in the first-stage regression.  

For simplicity, we will focus on the inference for the first treatment effect $\beta_1^*$. The same procedure applies to the other coordinates. We relegate the analytical formula of standard error of $\hat\beta^{(\mathcal{H})}_1$, denoted as $\hat{\rm SE}(\hat\beta_1^{(\mathcal{H})})$, to (A.1) in the supplementary Section A. If the instruments in $\mathcal{H}$ are valid, the $100(1-\alpha)\%$ TSLS confidence interval for $\beta^*_1$ is   
\begin{equation}\label{eq: def generic CI}
    {\rm CI}(\mathcal{H},\alpha) = \left[\hat\beta_1^{(\mathcal{H})}-z_{1-\alpha/2}\cdot \hat{\rm SE}(\hat\beta_1^{(\mathcal{H})}), \hat\beta_1^{(\mathcal{H})} + z_{1-\alpha/2}\cdot \hat{\rm SE}(\hat\beta_1^{(\mathcal{H})})\right]
\end{equation}
where $z_{1-\alpha/2}$ is the $(1-\alpha/2)$-th quantile of the standard normal distribution. We then move forward to the procedure.
 
\subsection{First-Stage Hard Thresholding}\label{subsec: first stage}
\par We first estimate the IV relevance matrix and collect the IVs that are significantly relevant to at least one of the treatments. We then select the just-identifying IV subsets based on the selected relevant IVs.
\par To select relevant IVs, we estimate the reduced form coefficients $\Gamma^*$ and $\Upsilon^*$ in (\ref{eq: reduced form}) by ordinary least squares (OLS), given as 
\begin{equation}\label{eq: OLS reduce}
    \hat\Gamma = (Z^\top \RM_X  Z)^{-1}Z^\top \RM_X Y,\text{ and }\hat\Upsilon = (Z^\top \RM_X Z)^{-1}Z^\top\RM_X D. 
\end{equation}
We use $\hat\Upsilon$ to select relevant IVs, and  $\hat\Gamma$ to select valid IVs later on. We use $\hat{\rm SE}(\hat \Upsilon_{k,j}) $ to denote the standard error of $\hat \Upsilon_{k,j}$ and relegate its analytical formula to (A.2) in the supplementary Section A. With the estimator of relevance matrix $\hat\Upsilon$, we define the estimated relevant IV set as
\begin{equation}\label{eq: est Strong IV}
   \hat{\mathcal{S}} = \left\{k\in[p_z]: \max_{j\in[p_d]} \left|\hat \Upsilon_{k,j}\right|/ \hat{\rm SE}(\hat \Upsilon_{k,j}) \geq \sqrt{ \log n  } \right\}, 
\end{equation}
which collects the IVs relevant to at least one of the treatments. The diverging threshold $\sqrt{\log n}$ has two purposes. First, it adjusts for multiplicity of evaluating the relevance of multiple IVs to multiple treatments. Second, it excludes weakly relevant IVs. 

Here we discuss the impact of weak IVs. We focus on the setting when we have a sufficient number of relevant IVs that strongly identify $\beta^*$, differentiating from the classical weak IV setting where all IVs as a whole are weak \citep{staiger1997instrumental}. Under our setting, the existence of weak IVs is not an essential concern.  According to the large-sample theory in \citet{staiger1997instrumental}, TSLS is asymptotically biased if the used IVs have local-to-zero relevance to the treatments  with $\|\Upsilon^*_{k\cdot}\|_2 = O(1/\sqrt{n})$. Under regularity conditions, the OLS estimator $\hat\Upsilon$ is $\sqrt{n}$-consistent, and therefore (\ref{eq: est Strong IV}) excludes these weak IVs with $O(1/\sqrt{n})$ coefficients. Consequently, in the presence of sufficient relevant IVs that ensure strong identification, weak IVs are excluded by our first-stage selection and therefore do not affect the subsequent inference. Consequently, in our setting, the weak IV issue is substantially less concerning than the locally invalid IV issue discussed in Section \ref{subsec: local invalid}. We return to this point in Remark \ref{rem: weak IV}.

Using the selected relevant IVs, we construct  $\binom{|\hat{\mathcal{S}}|}{p_d}$ subsets comprising $p_d$ relevant IVs. Note that not every subset comprising $p_d$ relevant IVs identifies a unique treatment-effect vector.  
Let $\mathcal{H}$ denote a subset comprising $p_d$ relevant IVs. Even if all IVs in $\mathcal{H}$ are relevant to at least one treatment, this subset fails to identify a unique treatment-effect vector when $\Upsilon^*_{\mathcal{H}\cdot}$ is rank deficient, that is, when $\rank(\Upsilon^*_{\mathcal{H}\cdot}) < p_d$. One example is when  $p_d = 2$ and all IVs in $\mathcal{H}$ are relevant to the first treatment but irrelevant to the second, so the second column of $\Upsilon^*_{\mathcal{H}\cdot}$ is zero. Another example of rank deficiency is when two columns of $\Upsilon^*_{\mathcal{H}\cdot}$ are proportional. This is essentially different from the single-treatment case with $p_d = 1$, where a nonzero vector $\Upsilon^*_{\mathcal{H}}$ always has full column rank.

In the following, we rule out the subsets comprising $p_d$ relevant IVs suffering from rank deficiency. A classical test statistic for this rank deficiency in the multivariate linear instrumental variable model dates back to \citet{cragg1993testing}.
For any subset $\mathcal{H}$ comprising $p_d$ estimated relevant IVs in $\hat{\mathcal{S}}$, we use ${\rm CD}^{(\mathcal{H})}$ to denote the Cragg and Donald (CD) test statistic, given as 
\begin{equation}\label{eq: CD test}
\begin{aligned}
{\rm CD}^{(\mathcal{H})} =  \lambda_{\min}\left[\hat\Sigma_\varepsilon^{-1/2}\hat\Upsilon_{\mathcal{H}\cdot}^\top \left( Z_{\cdot \mathcal{H}}^\top \RM_{W^{(\mathcal{H})}} Z_{\cdot \mathcal{H}} \right)\hat\Upsilon_{\mathcal{H}\cdot}  \hat\Sigma_\varepsilon^{-1/2}\right],\text{ where }\hat\Sigma_\varepsilon = n^{-1}\hat\varepsilon^\top\hat\varepsilon, 
\end{aligned}
\end{equation} $\hat\varepsilon = \RM_W D$ denotes the OLS residual for (\ref{eq: DGP d}), and $W^{(\mathcal{H})}=(X,Z_{\cdot [p_z]\backslash\mathcal{H}})$ collects the covariates and the instruments not in $\mathcal{H}$. Intuitively, the CD statistic measures the minimum squared singular value of $\Upsilon_{\mathcal{H}\cdot}^*$, where the matrices  $\left( Z_{\cdot \mathcal{H}}^\top \RM_{W^{(\mathcal{H})}} Z_{\cdot \mathcal{H}} \right)$ and $\hat\Sigma_\varepsilon^{-1/2}$ rescale the statistic accounting for the variances of IVs and the error term. 
 
 The CD statistic has an $O_p(1)$ order when $\Upsilon^*_{\mathcal{H}\cdot}$ is rank-deficient; otherwise it diverges to infinity. A small CD statistic indicates the existence of zero singular values in $\Upsilon_{\mathcal{H}\cdot}^*$, thereby suggesting rank deficiency. The CD statistic is analogous to the concentration parameter under a single treatment that measures the IV strength. We store the subsets comprising $p_d$ relevant IVs resulting in a CD test statistic no smaller than the threshold $\log n$:   
\begin{equation}\label{eq: hat H set}
    \hat{\mathscr{H}} =  \left\{ \mathcal{H} \subset \hat{\mathcal{S}} : |\mathcal{H}| = p_d, {\rm CD}^{(\mathcal{H})} \geq \log n\right\},
\end{equation}  
where the slowly diverging threshold $\log n$ accounts for the multiplicity in IV selections. We can show that $\hat{\mathscr{H}}$ recovers the collection of just-identifying IV subsets with high probability, where the latter is defined as
\begin{equation}\label{eq: strong finite}
    \mathscr{H}^* = \left\{ \mathcal{H}\subset \mathcal{S}^*: |\mathcal{H}|=p_d,\  \rank(\Upsilon^*_{\mathcal{H}\cdot}) = p_d \right\}.
\end{equation}
The set $\mathscr{H}^*$ collects the just-identifying IV subsets comprising $p_d$ IVs that identify a unique value of $\beta$. We write the collection of just-identifying IV subsets as  $\mathscr{H}^* = \{\mathcal{H}_1,\cdots,\mathcal{H}_{|\mathscr{H}^*|}\}$. Correspondingly, its finite-sample analogue is $\hat{\mathscr{H}} = \{\hat{\mathcal{H}}_1,\cdots,\hat{\mathcal{H}}_{|\hat{\mathscr{H}}|}\}$.

Thus far, we have completed the first-stage hard thresholding to select the just-identifying IV subsets. Each of these just-identifying IV subsets corresponds to a TSLS estimator. We move forward to the selection of valid IVs using these TSLS estimators by hard thresholding in Section \ref{subsec: TSHT}, and explain the IV selection error and its impact on statistical inference on the treatment effects in Section \ref{subsec: local invalid}.

\subsection{Two-Stage Hard-Thresholding for Multiple Treatments} \label{subsec: TSHT}
 In this section, we illustrate the extension of the two-stage hard-thresholding (TSHT, \citet{guo2018confidence}) to multiple treatment effects, and discuss the variable selection error by hard thresholding. We start with the procedures of TSHT, and then provide an illustrative graph to clarify the procedure in Figure \ref{fig:TSHT} with a specific example.
 
 Recall that for each possible value of effects $\beta$, we can use the vector $\pi(\beta)$ in (\ref{eq: pi beta eq system}) to evaluate the IV validity. When $\beta$ equals the true effect $\beta^*$, the vector $\pi(\beta^*) = \pi^*$ discloses the truth of IV validity. With observational data, for each just-identifying IV subset in $\hat{\mathscr{H}} = \{\hat{\mathcal{H}}_\ell\}_{1\leq\ell\leq|\hat{\mathscr{H}}|}$ defined in (\ref{eq: hat H set}), we use the TSLS estimator $\hat\beta^{(\hat{\mathcal{H}}_\ell)}$, $\ell=1,2,\cdots,|\hat{\mathscr{H}}|$, to back out an estimate of IV validity through 
\begin{equation}\label{eq: hat pi H}
    \hat\pi^{(\ell)} = \hat\Gamma - \hat\Upsilon\hat\beta^{(\hat{\mathcal{H}}_\ell)},
\end{equation}
where $\hat\Gamma$ and $\hat\Upsilon$ are the reduced-form model estimates defined in (\ref{eq: OLS reduce}).

\par Selection of valid IVs is usually based on hard thresholding in the literature \citep{guo2018confidence,windmeijer2021confidence}. In our case of multiple treatments, one strategy is to select $Z_{i,k}$ as a valid IV if the absolute value of  $\hat\pi_k^{(\ell)}$ is smaller than some threshold. Let $\hat{\rm SE}(\hat\pi_k^{(\ell)})$ denote the standard error of $\hat\pi_k^{(\ell)}$ with the definition relegated to (A.3) in the supplementary Section A. According to the common practice for single-treatment models \citep{guo2018confidence,windmeijer2021confidence}, we select the following as the valid IVs: 
\begin{equation}\label{eq: V hat hard}
    \hat{\mathcal{V}}^{(\ell)} = \{k\in\hat{\mathcal{S}}: |\hat\pi_k^{(\ell)}|\leq C_*\cdot \hat{\rm SE}(\hat\pi_k^{(\ell)})  \sqrt{\log n}\}.
\end{equation}
The factor $\sqrt{\log n}$ adjusts for the multiplicity of evaluating the validity of multiple IVs, and $C_*>0$ is an absolute constant that we specify in the  illustrative simulation below.  The threshold in (\ref{eq: V hat hard}) is of order $\sqrt{\log n/n}$ if $\hat\pi_k^{(\ell)}$
is $\sqrt{n}$-consistent.

Under Assumption \ref{cond: full rank}, with high probability there exists at least one just-identifying IV subset $\hat{\mathcal{H}}_\ell$ where all IVs are valid. Under regularity conditions, we can show that $(\hat\pi_k^{(\ell)} - \pi_k^*)=O_p\left(\hat{\rm SE}(\hat\pi_k^{(\ell)})\right)$.  Therefore, when the $k$-th IV is valid such that $\pi_k^*=0$,  (\ref{eq: V hat hard}) selects the $k$-th IV with high probability. 

In practice, it is unknown which just-identifying IV subset $\hat{\mathcal{H}}_\ell$ includes valid IVs, and therefore also unknown which $\hat{\mathcal{V}}^{(\ell)}$ correctly estimates the valid IV set $\mathcal{V}^*$. We thus need to leverage the identification condition developed in Section \ref{sec: identification}. For example, we can use the generalized plurality rule (Assumption \ref{cond: plurality new}) to select valid IV sets $\hat{\mathcal{V}}^{(\ell)}$ with the largest number of IVs that are the most likely to correctly select all valid IVs. We therefore store $\hat{\mathcal{V}}^{(\ell)}$ as a candidate valid IV set if $\ell\in\hat{\mathcal{L}}$, where 
\begin{equation}\label{eq: L hat TSHT}
 \hat{\mathcal{L}}  = \arg\max_{\ell=1,2,...,|\hat{\mathscr{H}}|} |\hat{\mathcal{V}}^{(\ell)}|.
\end{equation}
There may be multiple selected valid IV sets $\hat{\mathcal{V}}^{(\ell)}$ attaining the maximum cardinality in practice, and we cannot identify which set correctly selects all valid IVs. For robust inference, we define 
the TSHT confidence interval for $\beta_1^*$ under the significance level $\alpha$ as the union of the TSLS confidence intervals (\ref{eq: def generic CI}) using all selections that are possibly correct, given as
\begin{equation}\label{eq: TSHT CI}
    {\rm CI}^{\rm TSHT}(\alpha) = \bigcup_{\ell\in\hat{\mathcal{L}} } {\rm CI}(\hat{\mathcal{V}}^{(\ell)},\alpha).
\end{equation} 
The point estimate of $\beta^*$ is given as $\hat\beta^{\rm TSHT} = \frac{1}{|\hat{\mathcal{L}|}}\sum_{\ell\in\hat{\mathcal{L}}}\hat\beta^{(\hat{\mathcal{V}}^{(\ell)})}$.

\begin{figure}[t]
\centering 
\includegraphics[scale = 0.88]{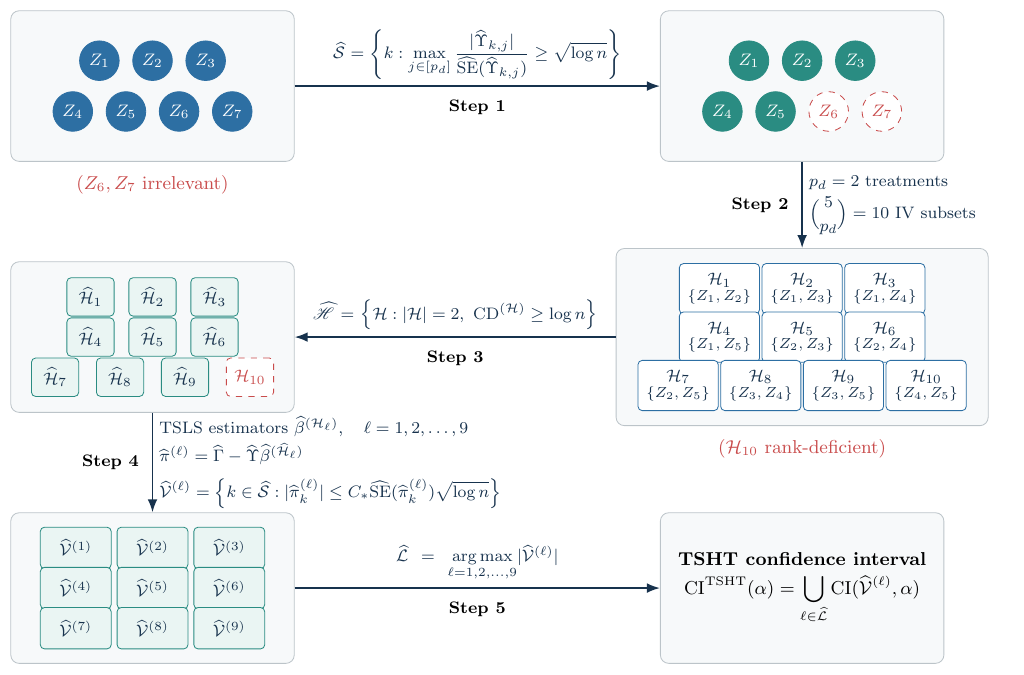}
\caption{\label{fig:TSHT} Illustration of TSHT for Multiple Treatments}
{\begin{flushleft}
    Note: This graph illustrates the procedures of multiple-treatment TSHT using an example with $p_d = 2$ treatments, $p_z = 7$ IVs, $|\mathcal{S}^*|=5$ relevant IVs, and $|\mathscr{H}^*| = 9$ just-identifying IV subsets. Specifically, $Z_6$ and $Z_7$ are irrelevant IVs, and the IV subset $\mathcal{H}_{10}$ suffers from rank deficiency. 
\end{flushleft}}
\end{figure}

Figure \ref{fig:TSHT} illustrates the TSHT procedure for  multiple treatments using an example with $p_d = 2$ treatments, $p_z = 7$ IVs, $|\mathcal{S}^*| = 5$ relevant IVs, and $|\mathscr{H}^*| = 9$ just-identifying subsets. Specifically, the IVs $Z_6$ and $Z_7$ are irrelevant and ruled out in Step 1 with $\hat{\mathcal{S}}$ in (\ref{eq: est Strong IV}). The 5 relevant IVs constitute $\left(\begin{array}{c}
             5 \\ 
             p_d   
          \end{array}\right) = 10$ 
IV subsets in Step 2, while the subset $\mathcal{H}_{10}$ suffers from rank-deficiency. Step 3 rules out $\mathcal{H}_{10}$  by (\ref{eq: hat H set}) using the CD statistic in (\ref{eq: CD test}). Step 4 adopts the retained 9 just-identifying IV subsets to generate 9 selected valid IV sets using the TSLS estimators, the IV validity estimators $\hat\pi^{(\ell)}$ in (\ref{eq: hat pi H}), and the hard thresholding (\ref{eq: V hat hard}). It is unknown which of the 9 selected valid IV subsets  correctly recovers all valid IVs. Therefore, Step 5 follows (\ref{eq: L hat TSHT}) to retain the selected valid IV sets leveraging the generalized plurality rule, and constructs  the TSHT confidence interval as the union of all TSLS confidence intervals using the retained valid IV sets following (\ref{eq: TSHT CI}). 

Like Step 4 in Figure \ref{fig:TSHT}, our procedure only goes through a finite number of TSLS estimators of $\beta^*$. This is the essential difference from the searching method in \citet{guo2023causal}. In multiple treatments, the searching method requires grid search over the ${p_d}$-dimensional space, which is computationally costly. 

We next discuss the IV selection error of TSHT.

\subsection{IV Selection Error of TSHT}\label{subsec: local invalid}
Valid IV selection  (\ref{eq: V hat hard}) is prone to errors if  certain invalid IVs are hard to distinguish from valid ones. We refer to these IVs as \textit{locally invalid IVs}, satisfying
\begin{equation}\label{eq: locally invalid}
    0 < |\pi_k^*| \leq C_*\cdot \max_{\ell\in\hat{\mathcal{L}}}\hat{\rm SE}(\hat\pi_k^{(\ell)}) \sqrt{\log n}.
\end{equation}
These locally invalid IVs can be selected as valid ones, thereby distorting the size of inference for $\beta_1^*$ based on the selected IVs in (\ref{eq: V hat hard}). These arguments echo the discussions of  \citet[Section 3]{guo2023causal} for single-treatment models about the IV selection error by hard thresholding, and the undercoverage of confidence intervals caused by erroneous IV selection.  

\begin{figure}[t]
\centering 
\includegraphics[scale = 0.8]{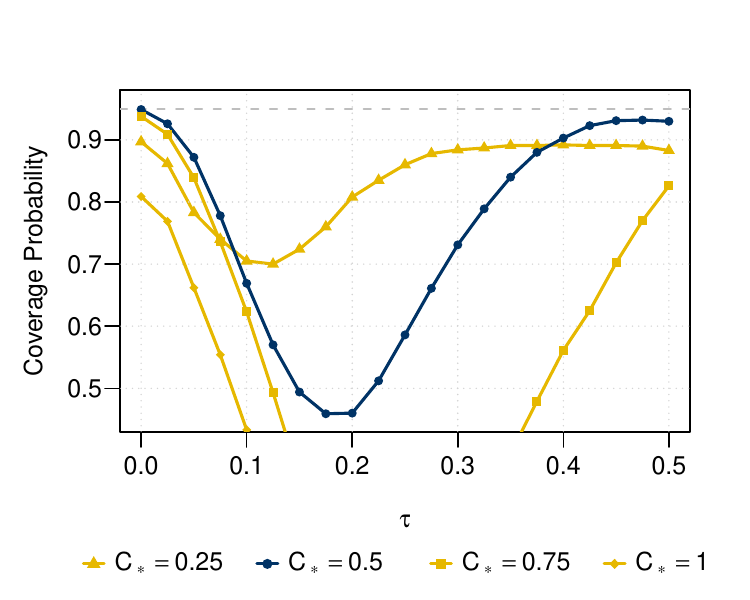}
\caption{\label{fig:illu} Coverage probability of the 95\% TSHT confidence interval with different violations of IV validity $\tau$}
\end{figure}
We provide an illustrative simulation to further clarify the impact of locally invalid IVs. We consider an example with $p_d=2$ treatments, $p_z=7$ IVs, and $p_x=5$ covariates. We set the sample size $n=1000$, and other settings follow the setting (S1) in Section \ref{sec: simul}. In particular, the measurement of IV validity $\pi^*=(0,0,0,0,0,\tau,0.5)^\top$. When $\tau \ne0$, the last two IVs are invalid. We vary $\tau$ from 0 to 0.5, where a relatively small positive $\tau$ characterizes a locally invalid IV. The 7th IV with $\pi_7^*=0.5$ is well separated from the valid IVs. 

Figure \ref{fig:illu} shows the coverage probabilities of the confidence interval for $\beta_1^*$ by TSHT under different $\tau$. The tuning parameter $C_*$ in (\ref{eq: V hat hard}) takes values in $\{0.25,0.5,0.75,1\}$, and we focus on $C_*=0.5$ under which the coverage probability is close to the nominal level when $\tau$ is around 0 or 0.5. When $\tau=0$, the 6th IV is valid and the coverage of the TSHT confidence interval is close to the nominal level. When $\tau$ ranges from 0.05 to 0.3, inference with TSHT shows severe size distortions, because IV selection based on hard thresholding cannot always distinguish a valid IV from a locally invalid one with a relatively small value of $\tau$. When $\tau$ exceeds 0.4, the coverage probability recovers to the nominal level, since the 6th IV becomes easily separated from the valid IVs. These results indicate that IV selection by hard thresholding suffers from errors in finite samples, causing size distortions in inference for $\beta_1^*$. 

Why does variable selection using hard thresholding cause bias in the second stage that selects valid IVs, but not in the first stage that selects relevant IVs? We answer this question in the following remark. 
\begin{Remark}[Locally Invalid IVs and Weak IVs]\label{rem: weak IV} Both the locally invalid IVs and the weak IVs are characterized by small coefficients. The key difference is that the validity screening selects small coefficients, whereas relevance screening selects large coefficients. When selecting valid IVs with a zero coefficient $\pi^*_k$, the hard thresholding (\ref{eq: V hat hard}) retains the IVs with small absolute estimated coefficients. This procedure will include the IVs with small but nonzero $\pi^*_k$ and thus cannot rule out locally invalid IVs. In sharp contrast, when selecting relevant IVs, the first-stage hard thresholding (\ref{eq: est Strong IV}) retains IVs with large absolute OLS coefficients, thereby excluding weak IVs. Consequently, our procedure is robust to the existence of weak IVs under the assumption that we have a sufficient number of valid and relevant IVs that strongly identify $\beta^*$.
\end{Remark}

To address the challenge in inference caused by locally invalid IVs, we shall devise a novel sampling procedure for uncertainty quantification in the following.

\subsection{Inference with Sampling Confidence Interval}\label{subsec: uniform}
We first provide intuition for how sampling can correct IV selection errors before presenting the formal details. Conditional on the observed data, for each just-identifying subset $\hat{\mathcal H}_\ell$ in (\ref{eq: hat H set}), we generate $M$ sampled IV validity measures $\{\tilde\pi^{(\ell,m)}\}_{m=1}^M$ from a normal distribution centered at $\hat\pi^{(\ell)}$ defined in \eqref{eq: hat pi H}. Because we have at least one just-identifying subset $\hat{\mathcal{H}}_{\ell^*}$ in which all IVs are valid, at least one estimator $\hat\pi^{(\ell^*)}$ is  distributed around the true IV validity $\pi^*$. Therefore, the sampling distribution of $\{\tilde\pi^{(\ell^*,m)}\}_{m=1}^M$ falls within a   neighborhood of $\pi^*$. As $M$ grows, we can show that at least one of the $M$ draws $\tilde\pi^{(\ell^*,m^*)}$ is almost the same as $\pi^*$, thereby recovering the correct set of valid IVs.

We now elaborate our inferential procedure for $\beta_1^*$, which  consists of two steps.  
\par {\bf Step 1: Sampling and Screening}. Recall that for each of the $|\hat{\mathscr{H}}|$ just-identifying IV subsets in (\ref{eq: hat H set}), we have an estimate $\hat\pi^{(\ell)}$ of IV validity in (\ref{eq: hat pi H}). To sample these estimates for relevant IVs selected by $\hat{\mathcal{S}}$, we draw $\xi^{[m]} \sim N(0,\RI_{|\hat{\mathcal{S}}|})$ for $m=1,2,\dots,M$, where $M$ is the number of resamples. We filter out the $\xi^{[m]}$'s excessively deviating from the center by  
\begin{equation}\label{eq: screen sample}
    \mathcal{M} = \left \{m\in[M]: \max_{k\in\hat{\mathcal{S}}}|\xi^{[m]}_k|\leq 1.1z_{1-\alpha_0/(2|\hat{\mathcal{S}}|)} \right \}
\end{equation} 
with $\alpha_0 \in (0,\alpha)$ being a small tuning parameter discussed in Remark \ref{rem: tuning}. The truncation by $z_{1-\alpha_0/(2|\hat{\mathcal{S}}|)}$ rules out extremely large sampled values. The truncation is necessary to avoid an overly conservative confidence interval and  guarantee the parametric rate of the interval length.

We perturb the estimated IV validity measures associated with each just-identifying IV subset in \eqref{eq: hat H set} using the sampled normal vectors $\xi^{[m]}\in\mathbb{R}^{|\hat{\mathcal S}|}$. For the $\ell$th just-identifying IV subset $\hat{\mathcal H}_\ell$ and the $m$th draw, define
\begin{equation}\label{eq: sampling pi}
    \tilde{\pi}_k^{(\ell,m)}
    =
    {\rm \bf 1}\{k\notin\hat{\mathcal H}_\ell\}
    \left\{
        \hat{\pi}_k^{(\ell)}
        +
        \widehat{\mathrm{SE}}\bigl(\hat{\pi}_k^{(\ell)}\bigr)
        \xi_k^{[m]}
    \right\},
    \qquad k\in\hat{\mathcal S}.
\end{equation}
Recall that $\hat{\pi}_k^{(\ell)}$, defined in \eqref{eq: hat pi H}, estimates the validity measure of IV $k$ using $\hat{\mathcal H}_\ell$, and its standard error $\widehat{\mathrm{SE}}\bigl(\hat{\pi}_k^{(\ell)}\bigr)$ is given in (A.3) of the supplementary Section A. We can derive that   $\hat{\pi}_k^{(\ell)}=0$ is always true for every $k\in\hat{\mathcal H}_\ell$. Intuitively, using the IVs in $\hat{\mathcal H}_\ell$ identifies a value of $\beta$ by solving $\pi_k(\beta)=0$ for any $k\in\hat{\mathcal H}_\ell$, and therefore these IVs always evaluate themselves as valid IVs. We therefore set $\tilde{\pi}_k^{(\ell,m)}=0$ for $k\in\hat{\mathcal H}_\ell$, as imposed by the indicator function in \eqref{eq: sampling pi}. We also define  $\widehat{\mathrm{SE}}\bigl(\hat{\pi}_k^{(\ell)}\bigr) = 0$ when $k\in\hat{\mathcal H}_\ell$.

Under the $m$-th sampling and the $\ell$-th just-identifying IV subset, we select valid IVs according to the sampled coefficients in (\ref{eq: sampling pi}), given as 
\begin{equation}\label{eq: sampling V hat}
     \tilde{\mathcal{V}}^{(\ell,m)} = \left\{k\in \hat{\mathcal{S}}: |\tilde\pi_k^{(\ell,m)}| \leq \hat{\rm SE}(\hat\pi_k^{(\ell)})\cdot \rho_n(M)  \right\},
\end{equation}  
where $\rho_n(M) = C_0(\log n/M)^{1/(2|\hat{\mathcal{S}}|)}$ with the tuning parameter $C_0$ discussed in Remark \ref{rem: tuning}.
\par We now demonstrate the effect of sampling. Recall that Assumption \ref{cond: full rank} guarantees the existence of  $p_d$ valid IVs with a full-rank relevance matrix. This implies that  with high probability the collection in \eqref{eq: hat H set} contains at least one just-identifying IV subset where all IVs are valid. Let $\ell^*$ denote the index of this all-valid just-identifying IV subset. Then $\hat\pi^{(\ell^*)}$ consistently estimates the true IV-validity vector $\pi^*$. The sampling in (\ref{eq: sampling pi}) generates a large number of estimators  $\{\tilde\pi^{(\ell^*,m)}\}_{m\in\mathcal{M}}$ in a neighborhood of the consistent estimate $\hat\pi^{(\ell^*)}$. We will show in Proposition \ref{prop: m star} that there exists an $m^*\in\mathcal{M}$ such that
\begin{equation}\label{eq: m star} 
     \left|  { \tilde\pi_k^{(\ell^*,m^*)} - \pi_k^* }\right| \leq {\hat{\rm SE}(\hat\pi_k^{(\ell^*)})}\rho_n(M),\,\ \ \forall k \in \hat{\mathcal{S}},
\end{equation}
where $\rho_n(M)$ is used in (\ref{eq: sampling V hat}). With $M\to\infty$, we have $\rho_n(M)\to0$ and thus $\tilde\pi^{(\ell^*,m^*)}$ is nearly the same as the true $\pi^*$. Therefore, the set $\tilde{\mathcal{V}}^{(\ell^*,m^*)}$ defined in (\ref{eq: sampling V hat}) correctly selects all valid IVs, even if the original estimate $\hat\pi^{(\ell^*)}$ causes IV selection error.  

Although there exists a $\tilde{\mathcal{V}}^{(\ell^*,m^*)}$ that correctly selects all valid IVs with high probability, in practice we need to determine which $\tilde{\mathcal{V}}^{(\ell,m)}$ to use. In the following discussions, we focus on leveraging (\ref{eq: general major pd}), the least restrictive scenario of the generalized majority rule. We will discuss the use of generalized plurality rule in Section \ref{subsec: not use plurality}. Under the $m$-th sampling, if 
$|\tilde{\mathcal{V}}^{(\ell,m)}|> (|\hat{\mathcal{S}}|+p_d - 1)/2$, the estimated valid IV set $\tilde{\mathcal{V}}^{(\ell,m)}$ is consistent with (\ref{eq: general major pd}), and we store it as a candidate valid IV set. We collect these candidate sets as 
\begin{equation}\label{eq: T t}
     \hat{\mathcal{T}}_m^{\rm GenMaj}  = \left\{ 1\leq\ell  \leq |\hat{\mathscr{H}}| : \left|\tilde{\mathcal{V}}^{(\ell,m)}\right| > \frac{|\hat{\mathcal{S}}|+p_d-1}{2}\right\}. 
\end{equation}
Under (\ref{eq: general major pd}), in the $m^*$-th sampling $\hat{\mathcal{T}}_{m^*}^{\rm GenMaj}$ includes $\ell^*$ that indexes the aforementioned valid IV set $\tilde{\mathcal{V}}^{(\ell^*,m^*)}$ with high probability.

{\bf Step 2: Aggregation.} In this step we use aggregation to utilize the reserved index set  $\hat{\mathcal{T}}^{\rm GenMaj}_m $ from Step 1. Recall that the $100(1-\alpha)\%$ TSLS confidence interval ${\rm CI}\left( \mathcal{H} ,\alpha \right)$ using the IVs in $\mathcal{H}$ is defined in (\ref{eq: def generic CI}). Since one of the selected valid IV sets $\tilde{\mathcal{V}}^{(\ell,m)}$ in (\ref{eq: sampling V hat}) correctly selects all valid IVs, at least one TSLS confidence interval ${\rm CI}\left( \tilde{\mathcal{V}}^{(\ell,m)} ,\alpha-\alpha_0\right)$ has sufficient coverage, where $\alpha_0$ compensates for the multiplicity caused by the screening in (\ref{eq: screen sample}). For robust inference, we aggregate these TSLS confidence intervals into our \emph{Sampling Confidence Interval} (SCI) for $\beta_1^*$, given as 
\begin{equation}\label{eq: def SCI CI}
    {\rm SCI}(\alpha) =
\bigcup_{m\in{\mathcal{M}}}\bigcup_{ \ell\in\hat{\mathcal{T}}_m^{\rm GenMaj}} {\rm CI}\left( \tilde{\mathcal{V}}^{(\ell,m)} ,\alpha-\alpha_0\right).
\end{equation}
 The set ${\rm SCI}(\alpha)$ is not necessarily a continuous interval, while we still call it a confidence interval for terminological consistency. 

Recall that there exist $\ell^*$ and $m^*$ such that (\ref{eq: m star}) holds, which yields correct IV selection with high probability. Under the leveraged identification condition, the $\ell^*$-th just-identifying IV subset is included in $\hat{\mathcal{T}}_{m^*}^{\rm GenMaj}$ with high probability, and thus at least one of the TSLS confidence intervals on the right-hand side of (\ref{eq: def SCI CI}) is valid. Therefore, the union confidence interval ${\rm SCI}(\alpha)$ covers the true value with a probability no less than the nominal coverage level $1-\alpha$. We summarize the procedure of SCI in Algorithm \ref{algo: SCI}.

\begin{algorithm}[t]
\footnotesize
\caption{\label{algo: SCI} Sampling Confidence Interval (SCI) for $\beta^*_1$}
\hspace*{0.01in}
\hspace*{\algorithmicindent} \textbf{Input:} Data $\{ Y_i,D_\ic,Z_\ic,X_\ic\}_{i=1}^n$; Sampling number $M$; Significance level $\alpha$; screening parameter $\alpha_0$; $\rho_n(M)=C_0(\log n/M)^{1/(2|\hat{\mathcal{S}}|)}$.\\
\hspace*{\algorithmicindent} \textbf{Output:} A $100(1-\alpha)\%$ confidence interval ${\rm SCI}(\alpha)$ for $\beta_1^*$.\\
\begin{algorithmic}[1]
\State Save $\hat{\mathcal{S}}$ by (\ref{eq: est Strong IV}) and $\hat{\mathscr{H}}$ by (\ref{eq: hat H set}). 
\State For each $m\in[M]$, sample $\xi^{[m]}\sim N(0,\RI_{|\hat{\mathcal{S}}|})$ and filter these random samples by $\mathcal{M}$ in (\ref{eq: screen sample}). \\
\For{$m\in\mathcal{M}$}{
\For{$\ell=1,2,\dots,|\hat{\mathscr{H}}|$}{
Calculate the sampled coefficient $\tilde\pi^{(\ell,m)}$ by (\ref{eq: sampling pi}). \newline
\hspace*{\algorithmicindent} \  Save estimated valid IV set $\tilde{\mathcal{V}}^{(\ell,m)}$ by (\ref{eq: sampling V hat}).
}

} 
\State Use $\tilde{\mathcal{V}}^{(\ell,m)}$ in (\ref{eq: sampling V hat}) and   $\hat{\mathcal{T}}^{\rm GenMaj}_m$ in (\ref{eq: T t}) to produce ${\rm SCI}(\alpha)$ in (\ref{eq: def SCI CI}).
\end{algorithmic}
\end{algorithm}

\subsection{Sampling Confidence Interval Using Generalized Plurality Rule}\label{subsec: not use plurality}
We now demonstrate the inference based on the generalized plurality rule. For the $m$-th resample, we retain a
candidate valid-IV set $\widetilde{\mathcal V}^{(\ell,m)}$ if its
cardinality is maximal among all candidate sets generated from that
resample. This criterion is the sample analogue of the generalized
plurality rule in Assumption \ref{cond: plurality new}. Specifically,
in place of \eqref{eq: T t}, define
\begin{equation}\label{eq: T t plur}
     \hat{\mathcal{T}}_m^{\rm GenPlur}  = \left\{1\leq \ell \leq |\hat{\mathscr{H}}| : |\tilde{\mathcal{V}}^{(\ell,m)}| = \max_{1\leq \ell^\prime\leq |\hat{\mathscr{H}}|}|\tilde{\mathcal{V}}^{(\ell^\prime,m)}| \right\}.
\end{equation}   
Other procedures still follow Algorithm \ref{algo: SCI}.

We use an illustrative simulation to compare the generalized plurality rule to the generalized majority rule in Section \ref{subsec: uniform}. Figure \ref{fig:illu2} shows the simulation results of the sampling method using $\hat{\mathcal{T}}^{\rm GenPlur}_m$ or $\hat{\mathcal{T}}^{\rm GenMaj}_m$ under the simulation settings in Section \ref{subsec: local invalid}, where both the generalized plurality rule and the generalized majority rule are satisfied. We observe that compared to $\hat{\mathcal{T}}^{\rm GenPlur}_m$, the use of $\hat{\mathcal{T}}^{\rm GenMaj}_m$ substantially shortens the confidence intervals without distorting the coverage probabilities. The generalized plurality rule is weaker than the generalized majority rule for causal effect identification in population, while in practice the filter $\hat{\mathcal{T}}_m^{\rm GenPlur}$ includes too many candidate IV sets so that the aggregate SCI is too wide.  
Therefore, we recommend using the generalized majority rule for practitioners, and in the following we focus on  the generalized majority rule for both theoretical justifications and numerical experiments.  
\begin{figure}[t]
\centering 
\includegraphics[scale = 0.72]{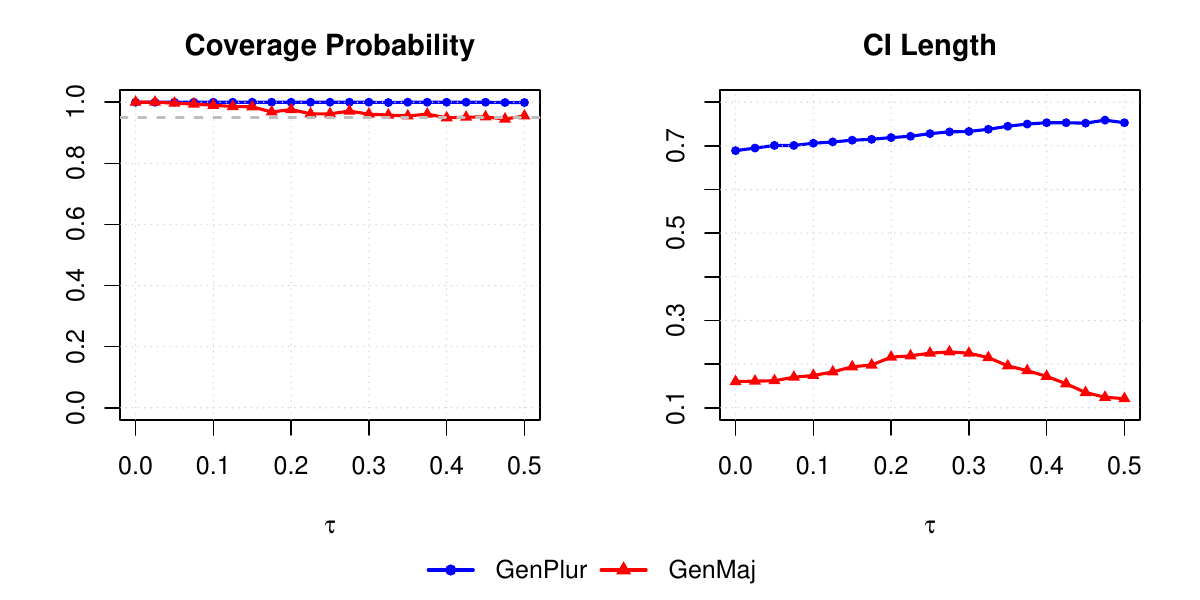}
\caption{\label{fig:illu2} Coverage probability and length of the 95\% sampling confidence interval  with $\hat{\mathcal{T}}^{\rm GenPlur}_m$ and $\hat{\mathcal{T}}^{\rm GenMaj}_m$}
{\begin{flushleft} \footnotesize{Note: The blue curves show the results with $\hat{\mathcal{T}}^{\rm GenPlur}_m$, while the red ones represent the outputs using $\hat{\mathcal{T}}^{\rm GenMaj}_m$}. The left panel shows the coverage probability of the 95\% intervals, with the gray dashed line highlighting the nominal level. The right panel exhibits the length of confidence intervals. 
\end{flushleft}}
\end{figure}

\begin{Remark}[Tuning Parameter Selection]\label{rem: tuning}
Tuning parameters for our robust inference include the number of resamples $M$, the constant $C_0$ in the sequence $\rho_n(M)$ of (\ref{eq: m star}), and the screening parameter $\alpha_0$. We suggest $M=1000$, under which the SCI performs well in all our simulation studies. For $C_0$ in  $\rho_n(M)$, we follow the idea of \citet{guo2023causal} and start with a small value $0.05$. We multiply $C_0$ by 1.25 until  $\hat{\mathcal{T}}^{\rm GenMaj}_m$ is nonempty for more than $5\%$ of the resamples $m$. The algorithm also stops when the number of iterations reaches the default maximum of 15. If the iteration number reaches the maximum, it indicates that the generalized majority rule is possibly violated, which can be used as a partial check of the generalized majority rule. Finally, we use $\alpha_0 = \alpha/20$. In the supplementary Section B, we provide additional simulations to exhibit the robust performance of SCI under various choices of tuning parameters.
\end{Remark}

\section{Theoretical Justifications for Sampling Confidence Interval}\label{subsec: theory}
We present the theoretical properties of the proposed  inference procedure based on the generalized majority rule. The following regularity conditions are used in the theoretical derivations. Throughout this theoretical section, an \textit{absolute constant} is positive and independent of the sample size $n$.
\begin{Assumption}\label{assu: distribution}Suppose the following conditions hold, for all $1\leq i\leq n$: 
\begin{enumerate}[(a)]
    \item The instruments and covariates $W_\ic = (Z_\ic^\top,X_\ic^\top)^\top$ are i.i.d.\  random vectors with bounded fourth moments. Suppose that $\Sigma = \E(W_\ic W_\ic^\top)$ satisfies $c_\Sigma \leq \lambda_{\min}(\Sigma) \leq \lambda_{\max}(\Sigma) \leq C_\Sigma$ for some absolute constants $C_\Sigma \geq c_\Sigma > 0$. 
    \item The error terms $v_\ic = (u_i,\varepsilon_i^\top)^\top$ are i.i.d. random vectors with uniformly bounded fourth moments. Assume that $\E(v_\ic | W_\ic) = 0$ and    $\Lambda = \E(v_\ic v_\ic^\top|W_\ic)$  satisfies  $c_\Lambda \leq \lambda_{\min}(\Lambda) \leq \lambda_{\max}(\Lambda) \leq C_\Lambda$ for some absolute constants $C_\Lambda \geq c_\Lambda > 0$.
\end{enumerate}
\end{Assumption}
Assumption \ref{assu: distribution} specifies the moments of the observed exogenous variables in $W_\ic$ and the unobserved error terms, ruling out perfect correlations through the bounded-eigenvalue condition.  We impose conditional homoskedasticity  via the deterministic conditional covariance matrix, following \citet{guo2018confidence}, \citet{windmeijer2021confidence}, and \citet{lin2024instrumental}. 
\par The next assumption is imposed on the IV relevance matrix $\Upsilon^*$. Recall that $\mathcal{S}^*$ collects the relevant IVs and $\mathscr{H}^*$ collects the just-identifying IV subsets.

\begin{Assumption}[IV Relevance]\label{assu: finite sample rank}
There exists some absolute constant $c_\Upsilon>0$, such that  (a) $\max_{j\in[p_d]}|\Upsilon^*_{k,j}| > c_\Upsilon$ for all $k\in\mathcal{S}^*$, and (b)  $\min_{\mathcal{H}\in\mathscr{H}^*} \lambda_{\min}(\Upsilon_{\mathcal{H} \cdot}^{*\top} \Upsilon_{\mathcal{H} \cdot}^*) > c_\Upsilon $. 
\end{Assumption}
Assumption \ref{assu: finite sample rank}(a) ensures that the relevant IVs are well separated from the irrelevant ones, therefore the screening by (\ref{eq: est Strong IV}) consistently selects the relevant IVs. Assumption \ref{assu: finite sample rank}(b) ensures the just-identifying IV subsets are well separated from the subsets not identifying a unique value of effect. These two conditions ensure that the IV subset selection by (\ref{eq: hat H set}) is consistent. For simplicity, we rule out weak IVs in the theoretical assumption to avoid distractions from the central issue of locally invalid IVs. As  explained in Remark \ref{rem: weak IV}, weak IVs are excluded  in the first-stage selection and thus not a concern for inference if there exist sufficiently many relevant IVs that ensure strong identification.  

\par We next derive the results for inference. The following proposition states the existence of $(\ell^*,m^*)$,  such that (\ref{eq: m star}) holds with high probability. That means there exists a sampled estimator $\tilde\pi^{(\ell^*,m^*)}$ sufficiently close to the true measurement of IV validity $\pi^*$. As explained around (\ref{eq: m star}), this is the essential result that corrects the IV selection errors under locally invalid IVs.  
\begin{Proposition}\label{prop: m star}Suppose that Assumptions \ref{cond: full rank}, \ref{assu: distribution} and \ref{assu: finite sample rank} hold. Then 
\begin{equation}\label{eq: m star exist} 
\mathop{\lim\inf}_{n\to\infty}\lim_{M\to\infty}\Pr\left\{\min_{\ell\in[|\hat{\mathscr{H}}|],m\in\mathcal{M}} \,\max_{k\in \hat{\mathcal{S}} } \left| \dfrac{ \tilde\pi_k^{(\ell,m)} - \pi_k^* }{\hat{\rm SE}(\hat\pi_k^{(\ell)})}\right| \leq \rho_n(M)\right\} \geq 1 - \alpha_0,
\end{equation}
where we define ratios of the form $0/0$ as zero to accommodate zero standard errors.
\end{Proposition}

Given that the sampled estimator $\tilde\pi^{(\ell^*,m^*)}$ is sufficiently close to the true IV validity, the estimated valid IV set $\tilde{\mathcal{V}}^{(\ell^*,m^*)}$ defined in (\ref{eq: sampling V hat}) correctly selects all valid IVs with a high probability. Therefore, at least one of the TSLS confidence intervals on the right-hand side of (\ref{eq: def SCI CI}) is valid. The final confidence interval ${\rm SCI}(\alpha)$ thus has sufficient coverage with high probability. The following theorem formally establishes the coverage probability of our SCI under $|\mathcal{V}^*| > \frac{|\mathcal{S}^*|+p_d-1}{2}$. 
\begin{Theorem}\label{thm: cover}
    Suppose that $|\mathcal{V}^*| > (|\mathcal{S}^*|+p_d-1)/2$. Under the conditions in Proposition \ref{prop: m star}, we have
    \begin{equation}
\mathop{\lim\inf}_{n\to\infty}\lim_{M\to\infty} \Pr\left\{ \beta^*_1 \in {\rm SCI}(\alpha) \right\} \geq 1 - \alpha, 
    \end{equation}
    where ${\rm SCI}(\alpha)$ is defined in (\ref{eq: def SCI CI}). 
\end{Theorem}
\par Theorem \ref{thm: cover} establishes that the SCI attains at least the nominal coverage probability asymptotically. We next study its length. A sufficient condition for the SCI to have the parametric length $O(n^{-1/2})$ is that every TSLS confidence interval contributing to the SCI has length $O(n^{-1/2})$ and is located within an $O(n^{-1/2})$ neighborhood of the true effect $\beta_1^*$. This condition requires $\beta^*$ to be uniquely identified. Otherwise, the SCI may cover not only the true $\beta_1^*$ but also other candidate values, in which case its length need not shrink at the parametric rate.

Section \ref{sec: identification} provides conditions for unique identification in the population.  To account for estimation errors in observational studies, we introduce a rate-dependent analogue of the population identification condition that accounts for the finite-sample estimation error. Define the collection of strongly invalid just-identifying subsets as
\begin{equation}\label{eq: Strongly invalid}
    \mathcal{SI}
    =
    \left\{
    1\leq \ell\leq |\mathscr{H}^*|:
    \|\pi^*_{\mathcal{H}_\ell}\|_\infty
    >
    \frac{C_1\log(C_2/\alpha_0)}{\sqrt{n}}
    \right\},
\end{equation}
where $C_1$ and $C_2$ are absolute constants. Thus, a subset belongs to $\mathcal{SI}$ if it contains at least one IV whose violation is larger than the local $n^{-1/2}$ scale.
\begin{Assumption}  \label{assume: finite multiple plurality} For each $1 \leq \ell \leq |{\mathscr{H}}^*|$, the initial estimator (\ref{eq: hat pi H}) is centered at some nonrandom vector $\pi^{(\ell)} = (\pi^{(\ell)}_1,\pi^{(\ell)}_2,\cdots,\pi^{(\ell)}_{p_z})^\top$ that may depend on $n$, such that 
\begin{equation}\label{eq: CLT pi ell}
    \sqrt{n}(\hat\pi^{(\ell)} - \pi^{(\ell)}) \convd \mathcal{N}(0,\Omega^{(\ell)}), 
\end{equation}
where $\Omega^{(\ell)}$ is a $p_z \times p_z $ positive semidefinite matrix. Assume that 
\begin{equation}\label{eq: finite plurality}
       \max_{\ell \in \mathcal{SI}}   \sum_{k\in\mathcal{S}^*} {\rm \bf 1}\left\{|\pi_k^{(\ell)}| \leq \dfrac{C_1\log(C_2/\alpha_0)}{\sqrt{n}} \right\}  \leq  \frac{|\mathcal{S}^*|+p_d-1}{2},
\end{equation}
where  $\alpha_0$ is used in (\ref{eq: screen sample}). 
\end{Assumption}

Assumption \ref{assume: finite multiple plurality} is imposed to control the CI length and not required for the coverage result in Theorem \ref{thm: cover}. For simplicity, (\ref{eq: CLT pi ell}) imposes the high-level asymptotic normality that can be verified using standard central limit theorems. It characterizes the error between the IV validity estimator $\hat\pi^{(\ell)}$ and the pseudo-true value $\pi^{(\ell)}$, where the latter equals the true IV validity $\pi^*$ only if all IVs in the just-identifying subset $\mathcal{H}_\ell$ are valid.   

We now explain the intuition behind \eqref{eq: finite plurality}. Its   purpose is to prevent a strongly invalid IV subset from generating a false candidate effect that is distant from the true $\beta^*$ but passes the majority screening in (\ref{eq: T t}). Specifically, a  strongly invalid just-identifying subset may cause some IVs with \(\pi_k^{(\ell)}\) close to zero to be mistakenly classified as valid. Condition \eqref{eq: finite plurality} ensures that there are no more than $(|\mathcal{S}^*| + p_d - 1)/2$ such IVs, and therefore these  IVs cannot pass the screening based on the generalized majority rule in (\ref{eq: T t}). Thus, candidate valid IV sets generated by strongly invalid subsets are excluded from SCI with high probability. The subsets retained in the SCI therefore contain only valid or locally invalid IVs, and thus their corresponding TSLS confidence intervals in (\ref{eq: def SCI CI}) remain close to the true value \(\beta_1^*\). This ensures that the SCI has length of order \(n^{-1/2}\).

Combining  $|\mathcal{V}^*|>\frac{|\mathcal{S}^*|+p_d-1}{2}$ with (\ref{eq: finite plurality}) in Assumption \ref{assume: finite multiple plurality} gives 
\begin{equation}\label{eq: finite plur analog}
|\mathcal{V}^*| > \max_{\ell\in\mathcal{SI}} \left|\left\{k\in \mathcal{S}^*:|\pi_k^{(\ell)}| \leq \frac{C_1\log(C_2/\alpha_0)}{\sqrt{n}}\right\}\right|.
\end{equation}
Thus, the number of valid instruments exceeds the largest number of locally invalid instruments under the pseudo-true value \(\pi^{(\ell)}\) associated with any strongly invalid just-identifying subset. This is the rate-dependent version of the population plurality rule, with exact equality to zero replaced by an \(O(n^{-1/2})\) neighborhood of zero to accommodate finite-sample estimation errors.

\par Define the length of the SCI as 
\[{\rm Len}(\alpha) = \max({\rm SCI}(\alpha)) -  \min({\rm SCI}(\alpha)),\]
where $\max(\mathcal{A})$ and $\min(\mathcal{A})$ denote the maximum and minimum elements in $\mathcal{A}\subset\mathbb{R}$. We conclude the theoretical section with Theorem \ref{thm: length}, showing that the length of the SCI is controlled by the parametric rate $1/\sqrt{n}$ with high probability. 
\begin{Theorem}\label{thm: length}
    Under Assumption \ref{assume: finite multiple plurality} and the conditions in Theorem \ref{thm: cover}, 
  there exist absolute constants $C_3>0$ and $C_4>0$ such that 
    \begin{equation}\label{eq: len}
\mathop{\lim\inf}_{n\to\infty}\lim_{M\to\infty}\Pr\left\{ {\rm Len}(\alpha)
 \leq \dfrac{C_3\log(C_4/\alpha_0)}{\sqrt{n}} \right\} \geq 1 - \alpha_0.
    \end{equation} 
\end{Theorem}
\medskip

\section{Monte Carlo Simulations}\label{sec: simul}
 \par Consider the linear instrumental variable regression model  (\ref{eq: DGP y}) under a setting with $p_d = 2$ treatments, and  $p_x = 5$ covariates, and $p_z = 7$ instrumental variables.  The i.i.d.\ data $\{Y_i,D_\ic,Z_\ic,X_\ic\}_{1\leq i\leq n}$ are generated via models (\ref{eq: DGP y}) and (\ref{eq: DGP d}), where $W_\ic = (Z_\ic^\top,X_\ic^\top)^{\top}$ follows i.i.d.\ joint normal distributions with mean zero and covariance matrix $\Sigma = (0.5^{|j-k|})_{j,k\in[p]}$. The error terms $v_i = (u_i,\varepsilon_\ic^\top)^{\top}$ follow i.i.d.\ joint normal distributions with mean zero, individual variance 1, and pairwise covariance 0.5. We vary the sample size $n$ across $\{500,1000,2000,5000\}$. 
 
 \par We consider $\beta^*=(\beta_1^*,\beta_2^*)^\top=(1,-1)^\top$ and set the  coefficients of covariates as $\varphi^*=0.5(0.3,0.6,0.9,1.2,1.5)^\top$, 
$\Psi^* = 0.5(\psi_1^*,\psi_2^*)$ where $\psi_1^* = (1,0.8,0.6,0.4,0.2)^\top$ and $\psi_2^* = (0.2,0.4,0.6,0.8,1)^\top$. We set the IV relevance matrix as $\Upsilon^* = 0.5(\gamma_1^*,\gamma_2^*)$ where $\gamma_1^* = (1,1,1,1,1,1,1)^\top$ and $\gamma_2^* = (-1.5,-1,-0.5,0,1.5,1,0.5)^\top$. We consider two settings of IV validity:
\begin{enumerate}[(S1)]
    \item  $\pi^*=(\textbf{0}_{5}^\top,\tau,0.5)^\top$,
  \item $\pi^*=(\textbf{0}_{5}^\top,\tau,\tau)^\top$.  
\end{enumerate}
Both settings include five valid instruments, thereby satisfying the condition (\ref{eq: general major pd}) and guaranteeing the validity of our SCI procedure. When $\tau$ is close to zero, (S1) contains a locally invalid IV and a strongly invalid IV, while both invalid IVs in (S2) are locally invalid. We vary $\tau$ from 0.05 to 0.5. A relatively small $\tau$ characterizes the locally invalid IVs discussed in Section \ref{subsec: local invalid}, under which the CI based on IV selection using hard thresholds, like TSHT, will suffer from size distortions. Under these settings, we compare the following methodologies to construct confidence intervals for $\beta_1^*$:
\begin{enumerate}
    \item Our proposed SCI method introduced in Section \ref{subsec: uniform};
    \item The TSHT confidence interval demonstrated in Section \ref{subsec: TSHT}, with $C_*=0.5$ and the selection criterion (\ref{eq: L hat TSHT}) replaced by the analogue for generalized majority rule
    \[\hat{\mathcal{L}} = \left\{1\leq \ell \leq |\hat{\mathscr{H}}|: |\hat{\mathcal{V}}^{(\ell)}| > \frac{|\hat{\mathcal{S}}|+p_d - 1}{2}\right\};\]
    \item The oracle bias-aware (OracleBA) confidence interval. For the TSHT estimator $\hat\beta_1^{\rm TSHT}$ in Section \ref{subsec: TSHT} and  its standard error ${\rm SE}(\hat\beta_1^{\rm TSHT})$, we assume that $(\hat\beta_1^{\rm TSHT} - \beta_1^*)/{\rm SE}(\hat\beta_1^{\rm TSHT}) \convd N(b_1,1)$ with $b_1$ denoting the asymptotic bias due to the IV selection error. Following Eq.~(7) in \citet{armstrong2020bias}, we leverage the oracle
      knowledge of $\E(\hat\beta_1^{\rm TSHT} - \beta_1^*)$ and construct the OracleBA confidence interval  
      \[(\hat\beta_1^{\rm TSHT}-\Delta_1,\hat\beta_1^{\rm TSHT}+\Delta_1),\ \Delta_1={\rm SE}(\hat\beta_1^{\rm TSHT})\sqrt{{\rm cv}_\alpha(|\E(\hat\beta_1^{\rm TSHT} - \beta_1^*)|^2/[{\rm SE}(\hat\beta_1^{\rm TSHT})]^2)},\]
where ${\rm cv}_\alpha(B^2)$ is the ($1-\alpha$) quantile of the $\chi^2$ distribution with 1 degree of freedom and non-centrality parameter $B^2$. The OracleBA CI serves as a fairer benchmark than the oracle TSLS confidence interval assuming prior knowledge of $\mathcal{V}^*$, since OracleBA accounts for the IV selection error.  
\end{enumerate}
When $\tau$ is relatively small, the sixth IV in (S1), as well as the sixth and seventh IVs in (S2), is locally invalid, under which the TSHT confidence interval is expected to suffer from undercoverage, as illustrated in Section \ref{subsec: local invalid}.

\begin{figure}[t]
\centering 
\includegraphics[scale = 0.8]{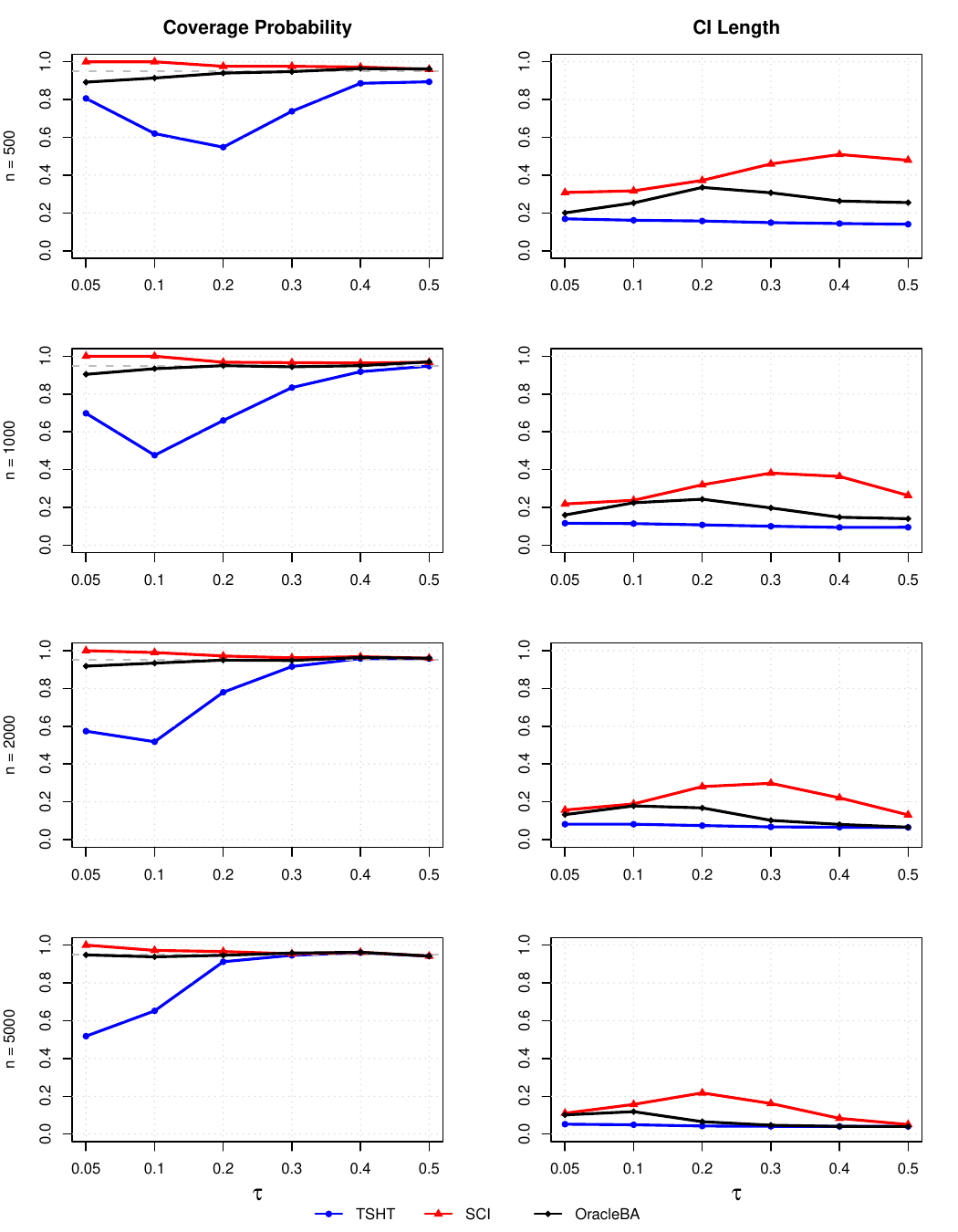}
\caption{\label{fig:simul by tau} Coverage Probabilities and CI Lengths for $\beta_1^*$ Varying with $\tau$ for (S1)}
\end{figure}

\begin{figure}[t]
\centering 
\includegraphics[scale = 0.8]{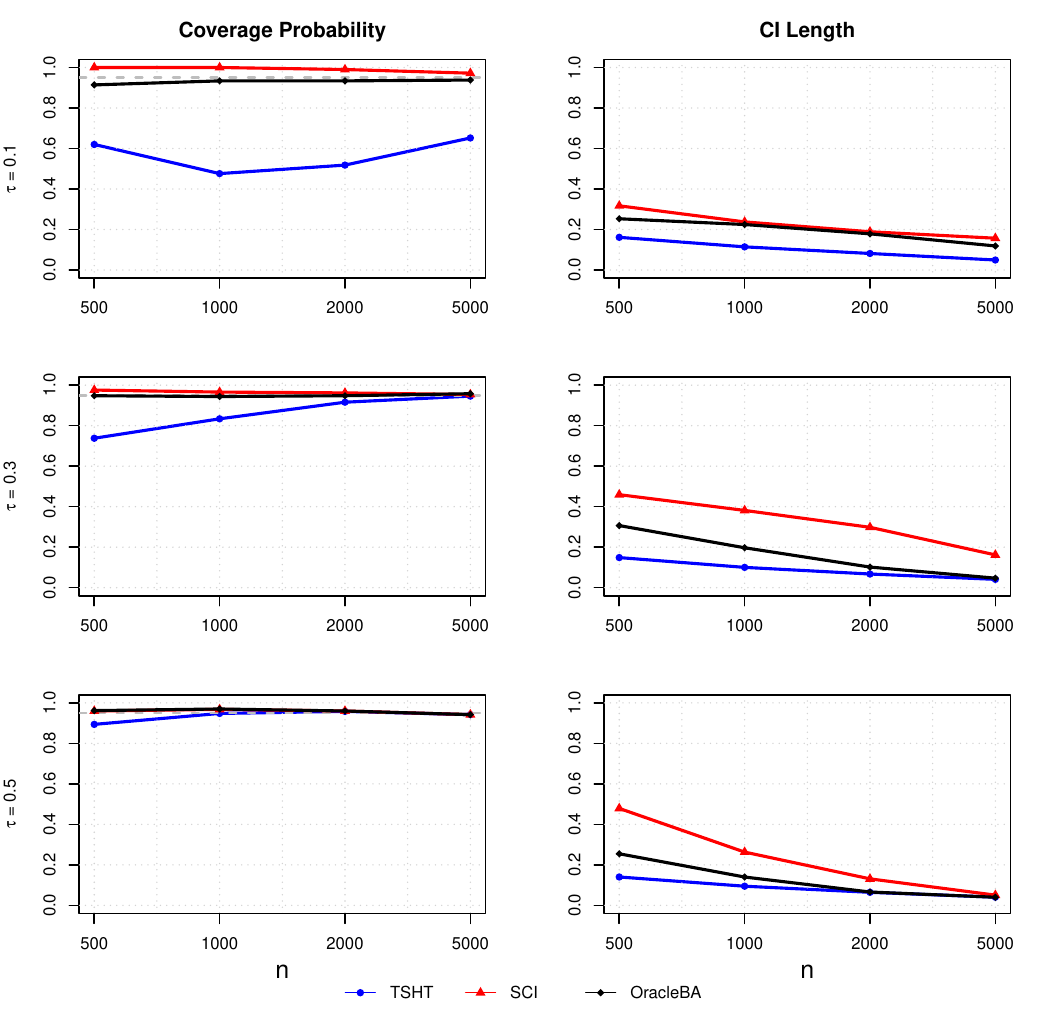}
\caption{\label{fig:simul by n} Coverage Probabilities and CI Lengths for $\beta_1^*$ Varying with $n$ for (S1)}
\end{figure}

To save space, we only show the results under the setup (S1). The results for (S2) are similar and relegated to Section B of the supplement. Figure \ref{fig:simul by tau} shows that the TSHT confidence interval suffers from severe undercoverage when the IV invalidity $\tau$ is close to zero. Even if the sample size reaches 5000, the coverage probability of the TSHT confidence interval is still below 60\% when $\tau = 0.05$. This result indicates the limitation of hard thresholding in the presence of locally invalid IVs, as explained in Section \ref{subsec: local invalid}. Increasing $\tau$ alleviates the undercoverage of the TSHT confidence interval, as a large $\tau$ well separates the invalid IVs from the valid ones, thus improving the accuracy of valid IV selection. 

In contrast, our SCI method remains robust under all IV validity levels $\tau$ and sample sizes $n$ considered in our simulation. The coverage probability of SCI is always close to the nominal level 95\%, which is similar to the OracleBA confidence interval. Regarding CI lengths, our SCI is only slightly wider than the confidence interval by OracleBA, suggesting that our inferential method achieves robustness without substantial power loss. Figure \ref{fig:simul by n} clearly shows that the length of the SCI converges to zero as the sample size $n$ increases, while the coverage probability remains robust. These results echo our theoretical justifications of the asymptotic coverage probability and length of our proposed SCI. 

\textbf{Tuning parameter selection.} As explained in Remark \ref{rem: tuning}, we set the resample size $M=1000$. For the constant $C_0$ in the sequence $\rho_n(M)$ of  (\ref{eq: m star}), we choose the smallest $C_0$ such that more than a pre-specified proportion (denoted as \texttt{prop}) of the $M$ intervals are non-empty, where \texttt{prop}$=5\%$. In Section B of the supplement, we demonstrate that our SCI is robust to different choices of $M$ and \texttt{prop}. Specifically, we execute additional simulations with $M\in\{500,1000,2000\}$ and \texttt{prop}$\in\{5\%,10\%,15\%\}$. The results show that our SCI has nearly
the same empirical coverage probability and length for different choices of $M$ and \texttt{prop}.
  
\section{Real Data Application}\label{sec: real data}
\par Mendelian randomization (MR) uses genetic variants as instrumental variables to estimate the effect of a modifiable exposure on an outcome of interest in the presence of unmeasured confounding. Recent MR studies often use genome-wide association study (GWAS) summary data, with single nucleotide polymorphisms (SNPs) serving as candidate instruments. We consider coronary artery disease (CAD) as the outcome, and low-density lipoprotein (LDL) and high-density lipoprotein (HDL) as exposures. This application implements our SCI for the causal effects of LDL and HDL on heart-disease outcomes while allowing some SNP instruments to be invalid. The IVs can be invalid due to the presence of population
stratification or horizontal pleiotropy, where the latter is the biological phenomenon in which a genetic variant affects the outcome through a pathway other than the exposure of interest. Locally invalid instruments are possible when population
stratification or horizontal pleiotropy are weak. Therefore, the SCI procedure is well suited for inference that is robust to IV selection error.  
\par We examine the causal relationships between lipoprotein traits and heart diseases using GWAS summary statistics. We consider CAD as the outcome $Y_i$ of interest, and a bivariate-treatment model where $D_\ic$ includes the levels of LDL and HDL. For candidate IVs, we use the 9 SNPs that are associated with both LDL and HDL, with detailed information in Section D.1 of the supplement. Specifically, the following summary statistics are available from the database:
\begin{enumerate}
    \item Estimates of the associations between the outcome and IVs analogous to the $\hat\Gamma$ in (\ref{eq: OLS reduce}), together with their standard errors and the $p$-values from the $t$-statistics.
    \item Estimates of the associations between the treatments and IVs analogous to the $\hat\Upsilon$ in (\ref{eq: OLS reduce}), together with their standard errors  and the $p$-values from the $t$-statistics. 
    \item The sample sizes used for estimation. 
\end{enumerate}
Details about the implementation of SCI with summary data are relegated to the supplementary Section D.2. 

\input{tab/real_data}

\par Table \ref{tab: real_data} exhibits the 95\% confidence intervals of the treatment effects of LDL and HDL on CAD using different methods. Overall, the results suggest that a higher level of LDL significantly increases the risk of CAD, while HDL does not exhibit a significant effect. This result echoes the findings of \citet*{global2013discovery} that LDL cholesterol is significantly  positively associated with CAD risk, whereas the inverse association between HDL and CAD in observational studies is insignificant in genetic analyses.

When we compare different methods, the confidence intervals from TSLS using all IVs and TSHT are relatively short. However, TSLS may suffer from invalid IVs, and TSHT is prone to undercoverage caused by IV selection errors when some IVs are locally invalid. In contrast, our SCI is longer but robust to locally invalid IVs.

\section{Conclusion}\label{sec: conclusion}
\par This paper studies identification conditions and robust inference for linear IV models with multiple endogenous treatments and possibly invalid instruments. We establish a condition called the generalized plurality rule under which the vector of multiple treatment effects is uniquely identified. We develop a sufficient condition called the generalized majority rule for identification. For studies with observational data, we propose an inference procedure that is robust to locally invalid instruments and IV selection error. These results shed light on how invalid instruments affect identification in multidimensional treatment settings and provide a practical route to robust inference when IV validity must be learned from the data.  

\bibliography{InvalidIV}
\bibliographystyle{apalike}

\setcounter{section}{0}
\renewcommand\thesection{\Alph{section}}
\clearpage
\setcounter{page}{1}
\setcounter{footnote}{0}
{\Large  
\begin{center}
Supplementary Materials for  ``Identification and Robust Inference for Multiple Treatments with Possibly Invalid Instruments''
\end{center}
}
{
\begin{center}
{Ziwei Mei}, {Qingliang Fan}, {Zijian Guo}\\
\end{center}
}

\vspace{-1em}
\newcounter{counter}[section]
\setcounter{table}{0}
\renewcommand{\thetable}{\thesection\arabic{table}}
\setcounter{equation}{0}
\renewcommand\theequation{\thesection\arabic{equation}}
\setcounter{figure}{0}
\renewcommand\thefigure{\thesection\arabic{figure}}
\setcounter{Theorem}{0}
\setcounter{Assumption}{0}
\setcounter{Lemma}{0}
\setcounter{Remark}{0}
\setcounter{Corollary}{0}
\setcounter{Proposition}{0}
\setcounter{Definition}{0}
\setcounter{Condition}{0}
\renewcommand\theTheorem{\thesection\arabic{Theorem}}
\renewcommand\theLemma{\thesection\arabic{Lemma}}
\renewcommand\theRemark{\thesection\arabic{Remark}}
\renewcommand\theCorollary{\thesection\arabic{Corollary}}
\renewcommand\theProposition{\thesection\arabic{Proposition}}
\renewcommand\theAssumption{\thesection\arabic{Assumption}}
\renewcommand\theDefinition{\thesection\arabic{Definition}}
\newcommand\thealgorithm{\thesection\arabic{algorithm}}
\renewcommand\theCondition{\thesection\arabic{Condition}} 

\setcounter{subsection}{0}
\setcounter{table}{0}
\renewcommand{\thetable}{\thesection\arabic{table}}
\setcounter{equation}{0}
\renewcommand\theequation{\thesection\arabic{equation}}
\setcounter{figure}{0}
\renewcommand\thefigure{\thesection\arabic{figure}}
\setcounter{Theorem}{0}
\setcounter{Assumption}{0}
\setcounter{Lemma}{0}
\setcounter{Remark}{0}
\setcounter{Corollary}{0}
\setcounter{Proposition}{0}

\renewcommand{\thealgocf}{\thesection\arabic{algocf}}
 \setcounter{algocf}{0}

\section{Omitted Formulas}\label{subsec: formula}

\par The standard error of $\hat\beta_1^{(\mathcal{H})}$ is
\begin{equation}\label{eq: asymp var beta}
\hat{\rm SE}(\hat\beta^{(\mathcal{H})}_1) = \sqrt{ \dfrac{(\hat\Omega_{\beta}^{(\mathcal{H})})_{1,1}}{n} },\text{ where } \hat\Omega_\beta^{(\mathcal{H})} = (\hat\sigma_u^{(\mathcal{H})})^2 \cdot \left[(n^{-1}D^{(\mathcal{H})\top} \RP_W D^{(\mathcal{H})})^{-1}\right]_{1:p_d,1:p_d},
\end{equation} 
and $(\hat\sigma_u^{(\mathcal{H})})^2  = n^{-1}\|Y-D^{(\mathcal{H})}\hat\theta^{(\mathcal{H})}\|_2^2$ is the mean squared TSLS residuals. 

The standard error of $\hat\Upsilon_{k,j}$ is
\begin{equation}\label{eq: asymp var Upsilon}
\hat{\rm SE}(\hat\Upsilon_{k,j})=\sqrt{\dfrac{\hat\Omega^{(j)}_{k,k}}{n}},\text{ where }\hat\Omega^{(j)} = (\hat\Omega^{(j)}_{k,t})_{k,t\in[p_z]} = 
(\hat\sigma^\varepsilon_j)^2 \cdot \left(n^{-1}Z^\top \RM_X Z\right)^{-1}, 
\end{equation}
and $(\hat\sigma^\varepsilon_j)^2 = n^{-1}\|\RM_W D_{\cdot j}\|_2^2$ is the mean squared OLS residuals by the regression of $D_{\cdot j}$ on $W = (Z,X)$. 

\par For any $k\in [p_z]\backslash\hat{\mathcal{H}}_\ell$, let $s(k)$ denote the order of $k$ in the sequence that organizes $[p_z]\backslash\hat{\mathcal{H}}_\ell$ in an ascending order. Then algebraically $\hat\pi_k^{(\ell)}$ equals the $(p_d + s(k))$-th entry in $ \hat\theta^{(\hat{\mathcal{H}}_\ell)}$ defined in (\ref{eq: def generic TSLS}) for any $k\in [p_z]\backslash\hat{\mathcal{H}}_\ell$.  Consider an example with $p_z=5$ and $p_d=2$.  $\hat{\mathcal{H}}_\ell=\{1,2\}$, we have $s(3)=1$, $s(4)=2$, and $s(5)=3$. If  $\hat{\mathcal{H}}_\ell=\{2,5\}$, we have $s(1)=1$, $s(3)=2$, and  $s(4)=3$.  The standard error of $\hat\pi^{(\ell)}_k$ is then  estimated by 
\begin{equation}\label{eq: asymp var pi}
   [\hat{\rm SE}(\hat\pi_{k}^{(\ell)})]^2 =  \begin{cases}
    0,\text{ if $k\in\hat{\mathcal{H}}_\ell$,}\\ \dfrac{(\hat\sigma_u^{(\ell)})^2}{n}\left[(n^{-1}D^{(\hat{\mathcal{H}}_\ell)\top}
     \RP_W D^{(\hat{\mathcal{H}}_\ell)})^{-1}\right]_{p_d+s(k),p_d + s(k)},\text{ if $k \in [p_z]\backslash\hat{\mathcal{H}}_\ell$},\\
    \end{cases}
\end{equation}
where $(\hat\sigma_u^{(\ell)})^2 = (\hat\sigma_u^{(\hat{\mathcal{H}}_\ell)})^2$ follows the definition of $(\hat\sigma_u^{(\mathcal{H})})^2$ below (\ref{eq: asymp var beta}). 

\section{Additional Simulation Results}\label{sec: additional simul}
This section includes the simulation results omitted from Section \ref{sec: simul}. Figures \ref{fig:simul by tau2 app} and \ref{fig:simul by n2 app} exhibit the simulation results for the setting (S2) where both invalid IVs are locally invalid. The results are similar to those in Figures \ref{fig:simul by tau} and \ref{fig:simul by n} for (S1) of the main text. We conduct additional simulations to investigate the robustness of SCI under various choices of tuning parameters. Under (S1) and $n=1000$, Figure \ref{fig:simul by M} exhibits the results with $M\in\{500,1000,2000\}$ when \texttt{prop} = 5\%, and Figure \ref{fig:simul by prop} shows the results with \texttt{prop}$\in\{5\%,10\%,15\%\}$. The coverage probabilities and lengths of SCI are very similar under various choices of tuning parameters, showing the robust performance of our SCI. 
\begin{figure}[h]
\centering 
\includegraphics[scale = 0.8]{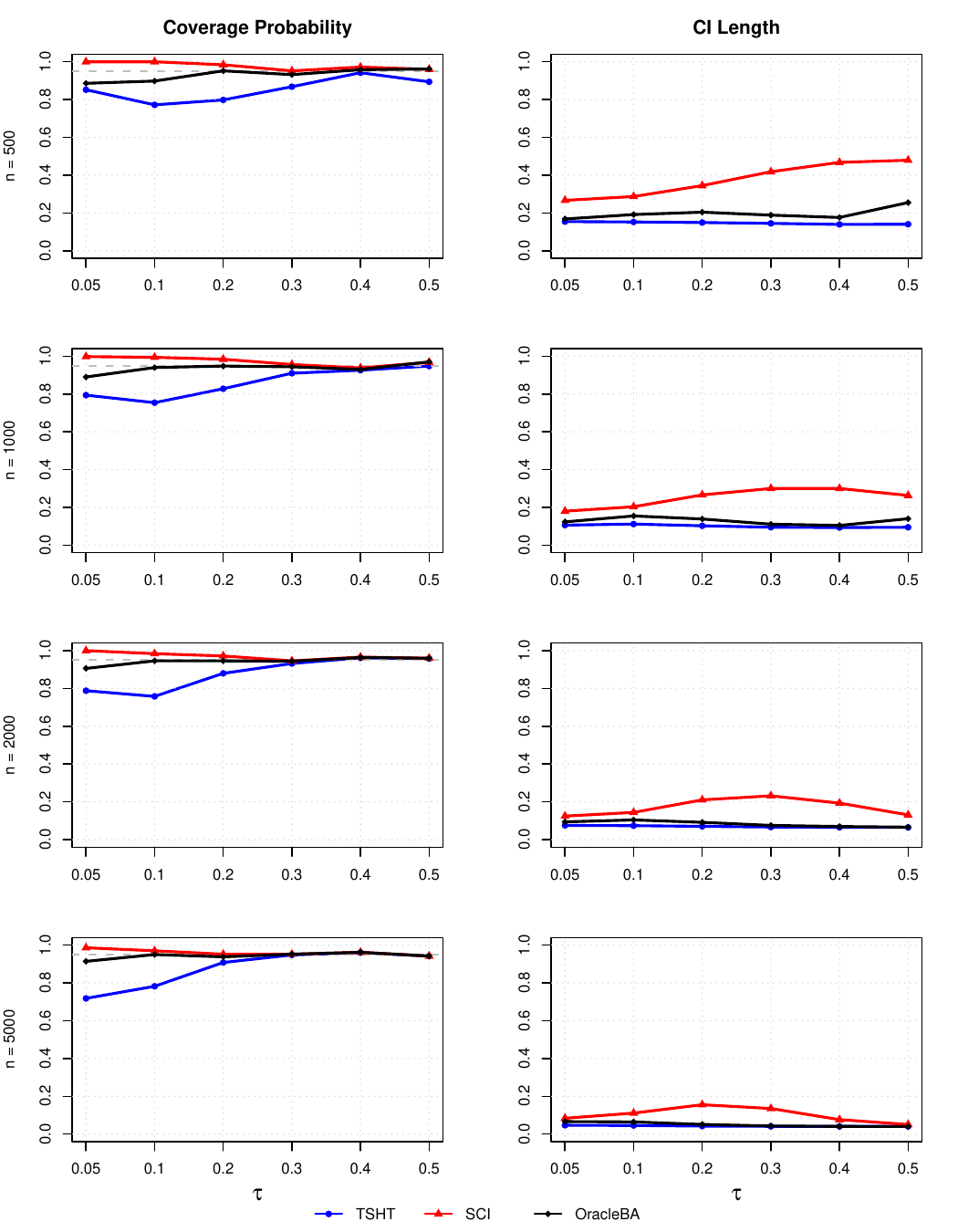}
\caption{\label{fig:simul by tau2 app} Coverage Probabilities and CI Lengths for $\beta_1^*$  Varying with $\tau$ for (S2) }
\end{figure}

\begin{figure}[t]
\centering 
\includegraphics[scale = 0.8]{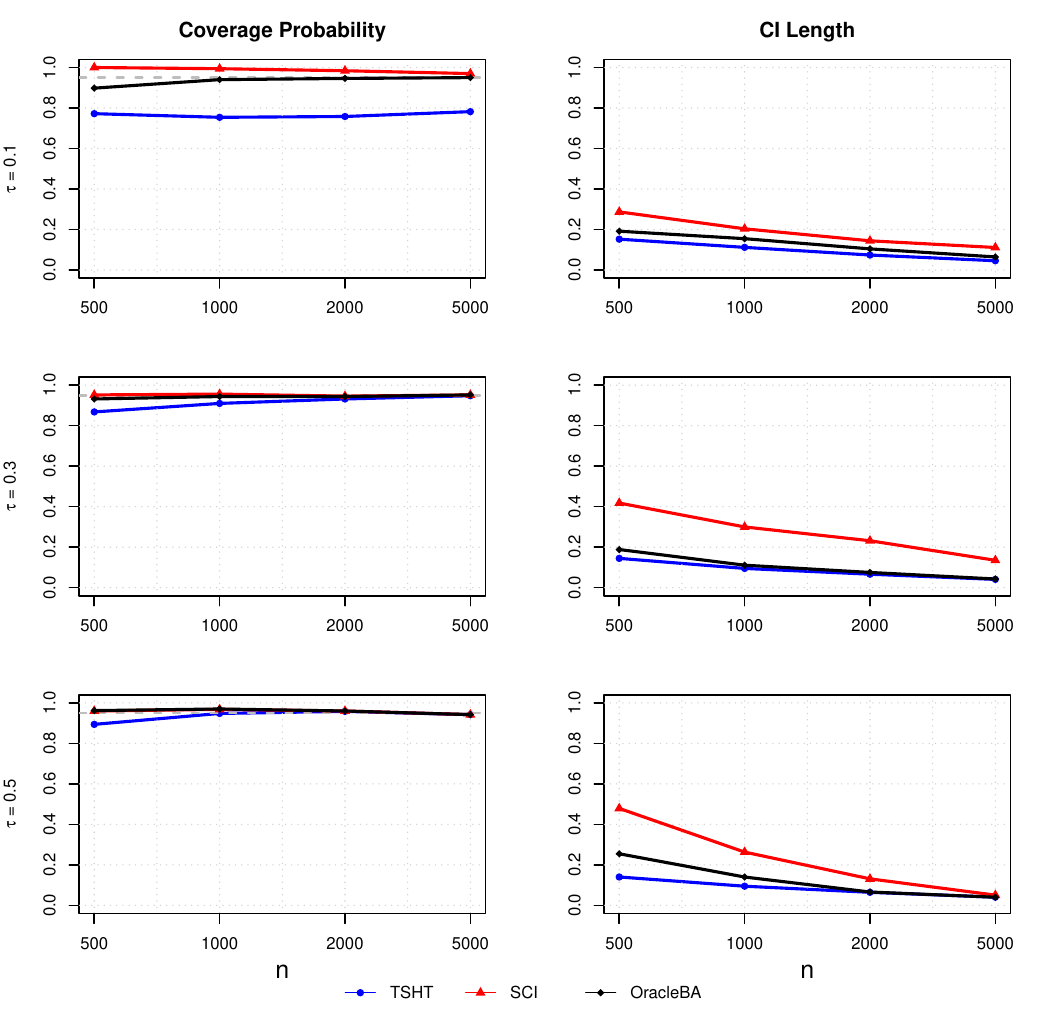}
\caption{\label{fig:simul by n2 app} Coverage Probabilities and CI Lengths for $\beta_1^*$  Varying with $n$ for (S2) }
\end{figure}

\begin{figure}[h]
\centering 
\includegraphics[scale = 0.75]{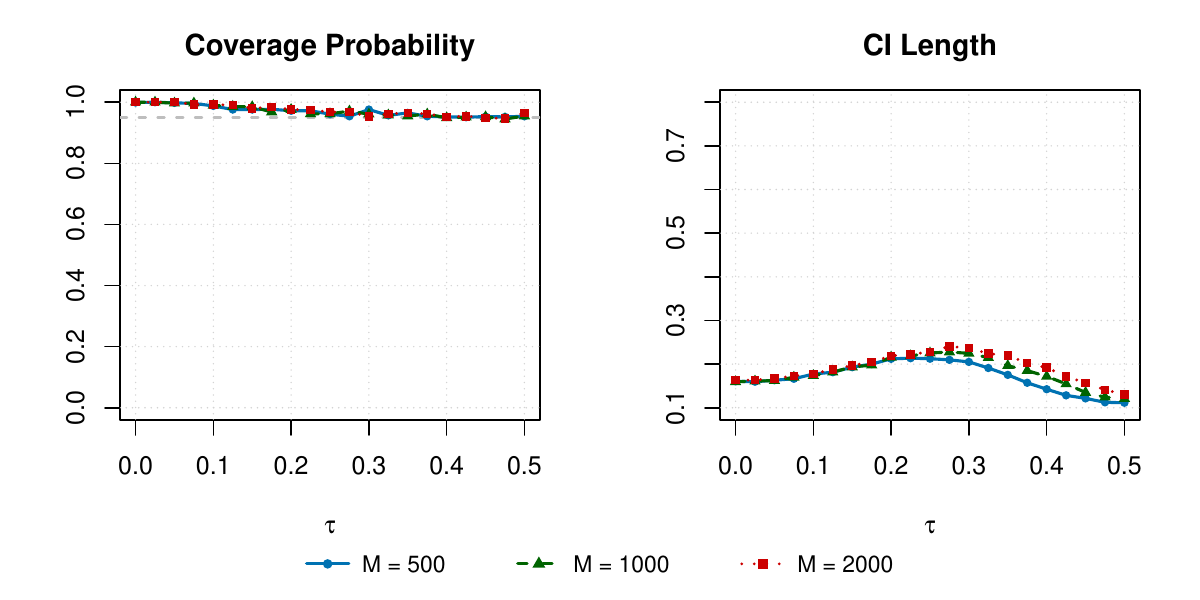}
\caption{\label{fig:simul by M} Simulation Results for $\beta_1^*$  with various $M$.}
\begin{flushleft}
\footnotesize
\textit{Note:} The results are generated under (S1), $n=1000$, and \texttt{prop} = 5\%.
\end{flushleft}
\end{figure}

\vspace{-3em}
\begin{figure}[h]
\centering 
\includegraphics[scale = 0.75]{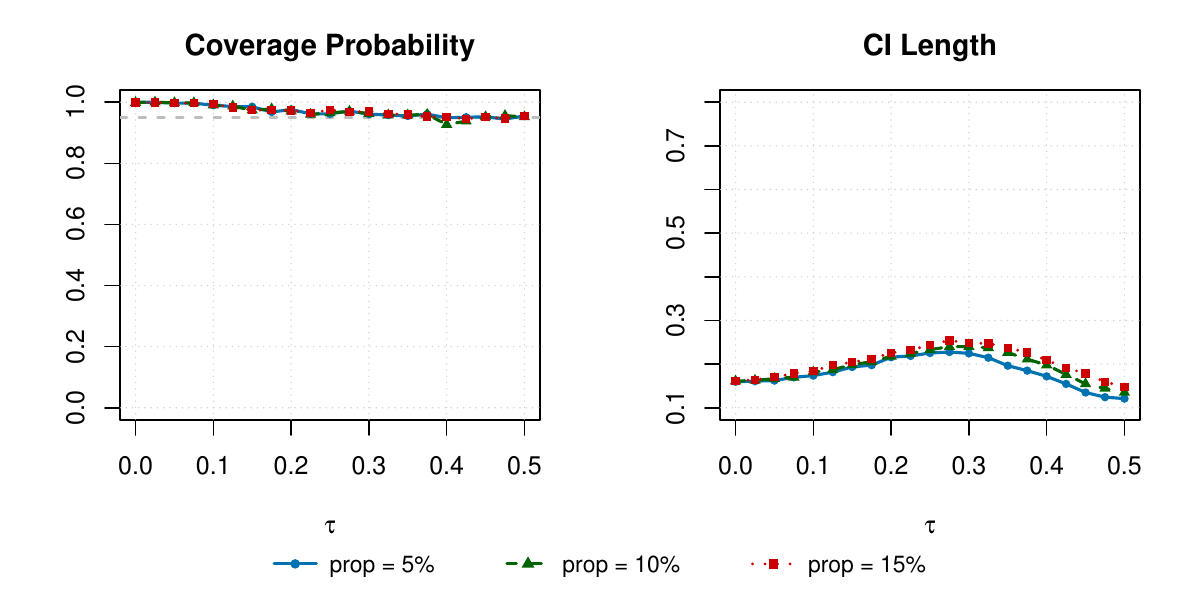}
\caption{\label{fig:simul by prop} Simulation Results for $\beta_1^*$   with various \texttt{prop}}
\begin{flushleft}
\footnotesize
\textit{Note:} The results are generated under (S1), $n=1000$, and $M = 1000$.
\end{flushleft}
\end{figure}

\clearpage

\input{app/proof}
\input{app/proof3}
\input{app/data_app}

\end{document}

%% file: fig/fig_thm.tex
\begin{figure}
\centering
	\subfigure[Successful Identification]{ 
	\begin{tikzpicture}[scale=0.7]
	\begin{axis}[
	    axis lines = center,
	    xmin=-2, xmax=8,
	    ymin=-2, ymax=8,
	    grid=none,
	    xtick=\empty,
	    ytick=\empty,
	]
	\addplot [
	    domain=-2:8, 
	    samples=100, 
	    color=blue,
	]
	{x};
	\addplot [
	    domain=-2:8, 
	    samples=100, 
	    color=blue,
	]
	{-x+3+3};
	\addplot [
	    domain=-2:8, 
	    samples=100, 
	    color=blue,
	]
	{2*x-6+3};
	\addplot [
	    domain=-2:8, 
	    samples=100, 
	    color=red,
	]
	{-2*x+2+6+3};
	\addplot [
	    domain=-2:8, 
	    samples=100, 
	    color=red,
	]
	{x/2+2-1.5+3};
	\node[label={$\boldsymbol{\beta^*}$},circle,fill,inner sep=2pt] at (axis cs:3,3) {};
	\node[label={$\boldsymbol{\beta^{(1)}}$},circle,fill,inner sep=2pt] at (axis cs:3-4/3,3+4/3) {};
	\node[label={$\boldsymbol{\beta^{(2)}}$},circle,fill,inner sep=2pt] at (axis cs:3+0,3+2) {};
 \node[label={$\beta_2$}] at (axis cs:3+4.6,-0.1) {};
\node[label={$\beta_1$}] at (axis cs:-0.4,3+3.6) {};
	\end{axis}
	\end{tikzpicture}
	}\noindent\hspace{0.08em}
	\subfigure[Counterexample]{
	\begin{tikzpicture}[scale=0.7]
\begin{axis}[
    axis lines = center,
    xmin=-2, xmax=8,
    ymin=-2, ymax=8,
    grid=none,
    xtick=\empty,
    ytick=\empty,
]
\addplot [
    domain=-2:8, 
    samples=100, 
    color=blue,
]
{x};
\addplot [
    domain=-2:8, 
    samples=100, 
    color=blue,
]
{-x+3+3};
\addplot [
    domain=-2:8, 
    samples=100, 
    color=blue,
]
{2*x-6+3};
\addplot [
    domain=-2:8, 
    samples=100, 
    color=red,
]
{-2*x+16/3+6+3};
\addplot [
    domain=-2:8, 
    samples=100, 
    color=red,
]
{x/2+2-1.5+3};
\node[label={$\boldsymbol{\beta^*}$},circle,fill,inner sep=2pt] at (axis cs:3,3) {};
\node[label={$\boldsymbol{\beta^{(1)}}$},circle,fill,inner sep=2pt] at (axis cs:3-4/3,3+4/3) {};
\node[label={$\boldsymbol{\beta^{(2)}}$},circle,fill,inner sep=2pt] at (axis cs:3+4/3,3+8/3) {};
\node[label={$\beta_2$}] at (axis cs:3+4.6,-0.1) {};
\node[label={$\beta_1$}] at (axis cs:-0.4,3+3.6) {};
\end{axis}
\end{tikzpicture}
}
\caption{Illustrative graphs for Identification}
{\begin{flushleft}\footnotesize{ Note: Each panel sketches an example with 5 IVs and 2 treatments. Each axis represents the value of one treatment effect. Each straight line represents a hyperplane defined in (\ref{eq: hyperplane}) of an IV. The hyperplanes of valid IVs are in blue color, while those of invalid ones are red. A vector $\beta=(\beta_1,\beta_2)^\top$ receives a vote from an IV when it falls on the corresponding hyperplane. For example, $\beta^*$ in both panels receives votes from three valid IVs, and $\beta^{(1)}$ in both panels receives two votes from a valid IV and an invalid one.  }\end{flushleft}  }
\label{fig: thm}
\end{figure}

%% file: tab/real_data.tex
\begin{table}[t]
\begin{center}
\small
\caption{Confidence Intervals for Effects of  Lipoprotein on CAD }\label{tab: real_data}
\begin{tabular}{|c| c|c|} 
\hline 
Method                   & LDL           & HDL             \\
\hline 
TSLS                 & (0.183,0.500) & (-0.315,0.112) \\
 
 {TSHT}  & (0.254,0.618) & (-0.224,0.328) \\
 {SCI}   & (0.168,0.776) & (-0.375,0.328) \\
                        \hline  
\end{tabular}
\end{center}
{\footnotesize Notes: The table shows the 95\% confidence intervals. ``TSLS'' represents the two-stage least squares estimator using all IVs. ``TSHT'' is the hard thresholding method in Section \ref{subsec: TSHT}. ``SCI'' is our proposed method. }
\end{table}

%% file: app/proof.tex
\section{Proofs}
\subsection{Proofs of Section \ref{sec: identification}}
We first state a lemma about the relationship between the existence of solutions to a linear equation and matrix rank.

\begin{Lemma}\label{lem: matrix rank}(a) \citep[Exercise 6.42]{abadir2005matrix} For any $m\times p_d$ matrix $\Upsilon$ and $m$-dimensional vector $\pi$, the linear equation $\Upsilon x = \pi$ has a solution if and only if ${\rm rank}(\Upsilon)={\rm rank}(\Upsilon,\pi)$. \\
(b) \citep[Exercise 6.33]{abadir2005matrix} For any $m\times p_d$ matrix $\Upsilon$, the linear equation $\Upsilon x = \boldsymbol{0}_m$ has a nonzero solution if and only if ${\rm rank}(\Upsilon)<p_d$. 
\end{Lemma}

\begin{proof}[Proof of Theorem \ref{thm: plur}]By definition, 
\[\begin{aligned}
    \max_{\beta\neq\beta^*}|\{k\in\mathcal{S}^*:\pi_k(\beta)=0\}| &\leq \max_{\beta\neq\beta^*}|\{k \in \mathcal{S}^*\backslash \mathcal{V}^*:\pi_k(\beta)=0\}| + \max_{\beta\neq\beta^*}|\{k\in\mathcal{V}^*:\pi_k(\beta)=0\}| \\
    &\leq |\mathcal{S}^*| - |\mathcal{V}^*| + \max_{\beta\neq\beta^*}|\{k\in\mathcal{V}^*:\pi_k(\beta)=0\}|. 
\end{aligned}\]
The condition $|\mathcal{V}^*| > \frac{|\mathcal{S}^*| + h_0 - 1}{2}$ can be reorganized as  
\[|\mathcal{V}^*| > |\mathcal{S}^*| - |\mathcal{V}^*| + h_0 - 1. \] 
It then suffices to show 
\begin{equation*} 
    \max_{\beta\neq\beta^*}|\{k\in\mathcal{V}^*:\pi_k(\beta)=0\}| \leq h_0 - 1.
\end{equation*} 
Recall that by definition, $\pi(\beta) = \pi^* + \Upsilon^*(\beta^* - \beta)$. As $\pi^*_k = 0$ for all $k\in \mathcal{V}^*$, we have 
\[\begin{aligned}
    \max_{\beta\neq\beta^*}|\{k\in\mathcal{V}^*:\pi_k(\beta)=0\}| &= \max_{\beta\neq\beta^*}|\{k\in\mathcal{V}^*:\Upsilon^*_{k\cdot}(\beta^* - \beta)=0\}| \\
    &= \max_{x\neq \boldsymbol{0}_{p_d}}|\{k\in\mathcal{V}^*:\Upsilon^*_{k\cdot}x=0\}|.
\end{aligned}\]
We then need to show 
\begin{equation}\label{eq: suffices to show T1}
    \max_{x\neq \boldsymbol{0}_{p_d}}|\{k\in\mathcal{V}^*:\Upsilon^*_{k\cdot}x=0\}| \leq h_0 - 1.
\end{equation} 
We prove (\ref{eq: suffices to show T1}) by contradiction. If $\max_{x\neq \boldsymbol{0}_{p_d}}|\{k\in\mathcal{V}^*:\Upsilon^*_{k\cdot}x=0\}| > h_0 - 1$, there exist a nonzero vector $x\in\mathbb{R}^{p_d}\backslash\{\boldsymbol{0}_{p_d}\}$ and a subset of valid IVs indexed by $\mathcal{H}_1\subset \mathcal{V}^*$, such that 
\[|\mathcal{H}_1| > h_0 - 1\text{ and }\Upsilon^*_{\mathcal{H}_1\cdot}x=0.\]
Let $\mathcal{H}_0$ be a subset of $\mathcal{H}_1$ satisfying $|\mathcal{H}_0| = h_0$. Then
\[\Upsilon^*_{\mathcal{H}_0\cdot}x=0.\]
That means the linear equation $\Upsilon^*_{\mathcal{H}_0\cdot}x=0$ has a nonzero solution $x$. Then by Lemma \ref{lem: matrix rank}(b),     
\begin{equation}\label{eq: h0 rank < pd}
   {\rm rank}(\Upsilon^*_{\mathcal{H}_0\cdot}) < p_d.
\end{equation} 
Because $ |\mathcal{H}_0| = h_0$, (\ref{eq: h0 rank < pd}) is contradictory to the definition of $h_0$ in (\ref{def: h}). Thus, (\ref{eq: suffices to show T1}) is verified by contradiction and we complete the proof of Theorem \ref{thm: plur}. 
\end{proof}
\begin{proof}[Proof of Proposition \ref{prop: rare}]
\par Notice that $\Pr\{\text{Assumption \ref{cond: plurality new} holds}\}=1$ is a direct corollary of (\ref{eq: max vote pd}) and $|\mathcal{V}^*| > p_d$. It thus suffices to prove (\ref{eq: max vote pd}). For any $\mathcal{H}\subset \mathcal{S}^*$, define the event
\begin{equation}\label{eq: def R (H)}
    \mathscr{R}(\mathcal{H}) = \{\text{${\rm rank}(\Upsilon^*_{\mathcal{H}\cdot}) < {\rm rank}(\Upsilon^*_{\mathcal{H}\cdot},\pi^*_{\mathcal{H}})$}\}.
\end{equation} 
Further, define 
\[\mathscr{R}^* = \{\text{ ${\rm rank}(\Upsilon^*_{\mathcal{H}\cdot}) = p_d$ for any $\mathcal{H}\subset \mathcal{S}^*$ satisfying $|\mathcal{H}|\geq p_d$}\}\]
and the class of subsets 
\[\mathscr{H}_0 = \{\mathcal{H}\subset \mathcal{S}^*: |\mathcal{H}|>p_d\text{ and }\pi^*_{\mathcal{H}} \neq \boldsymbol{0}_{|\mathcal{H}|}\}. \]
In words, $\mathscr{H}_0$ includes all subsets $\mathcal{H} \subset \mathcal{S}^*$ of cardinality  greater than $p_d$, such that at least one of  the instruments indexed by $\mathcal{H}$ is invalid.  
\par It suffices to show the following results:
\begin{enumerate}[(a)]
    \item (\ref{eq: max vote pd}) holds when both $\mathscr{R}^*$ and $\bigcap_{\mathcal{H} \in \mathscr{H}_0}\mathscr{R}(\mathcal{H})$ hold;   
    \item $\mathscr{R}^*$ and $\bigcap_{\mathcal{H} \in \mathscr{H}_0}\mathscr{R}(\mathcal{H})$ hold with probability one.
\end{enumerate}
\underline{Proof of (a)}. Recall that by definition, $\pi(\beta) = \pi^* + \Upsilon^*(\beta^* - \beta)$. Lemma \ref{lem: matrix rank}(a) suggests that under $\mathscr{R}(\mathcal{H})$, the linear equation $\Upsilon_{\mathcal{H}\cdot}^* x + \pi_{\mathcal{H}}^* = \boldsymbol{0}_{|\mathcal{H}|}$ has no solution. Since $\pi_{\mathcal{H}}(\beta) = \pi^*_{\mathcal{H}} + \Upsilon^*_{\mathcal{H}\cdot}(\beta^* - \beta)$, under $\mathscr{R}(\mathcal{H})$ we cannot find any $\beta$ such that $\pi_{\mathcal{H}}(\beta) = \boldsymbol{0}_{|\mathcal{H}|}$. Therefore, under $\bigcap_{\mathcal{H} \in \mathscr{H}_0}\mathscr{R}(\mathcal{H})$, we have 
\[\pi_{\mathcal{H}}(\beta) \neq \boldsymbol{0}_{|\mathcal{H}|}\text{ for any $\beta \in \mathbb{R}^{p_d}$, if }\mathcal{H} \in \mathscr{H}_0.\]
Therefore, under $\bigcap_{\mathcal{H} \in \mathscr{H}_0}\mathscr{R}(\mathcal{H})$, if there exists a subset $\mathcal{H}\subset \mathcal{S}^*$ and a candidate effect $\beta \in \mathbb{R}^{p_d}$ such that $\pi_{\mathcal{H}}(\beta) = \boldsymbol{0}_{|\mathcal{H}|},$ we have $\mathcal{H}\notin \mathscr{H}_0$. That means one of the following must be satisfied:
\begin{enumerate}[(I)]
    \item $|\mathcal{H}| > p_d$ and  $\pi^*_{\mathcal{H}} = \boldsymbol{0}_{|\mathcal{H}|}$; 
    \item $|\mathcal{H}| \leq p_d$.
\end{enumerate}
\par Under (I), by $\pi^*_{\mathcal{H}} = \boldsymbol{0}_{|\mathcal{H}|}$ and $\pi_{\mathcal{H}}(\beta) = \boldsymbol{0}_{|\mathcal{H}|}$ we have
\begin{equation}\label{eq: Ups 0}
    \Upsilon^*_{\mathcal{H}\cdot}(\beta^* - \beta) = \pi_{\mathcal{H}}(\beta) - \pi^*_{\mathcal{H}} = \boldsymbol{0}_{|\mathcal{H}|}.
\end{equation} 
Note that the matrix $\Upsilon^*_{\mathcal{H}\cdot}$ is full-column-rank under $\mathscr{R}^*$. By Lemma \ref{lem: matrix rank}(b), the only solution to (\ref{eq: Ups 0}) is $\beta = \beta^*$. Therefore, if the additional restriction $\beta \neq \beta^*$ is imposed, the only possibility is (II). Then 
\[\begin{aligned}
    \max_{\beta\neq\beta^*} \left|\{k\in\mathcal{S}^*:\pi_k(\beta) = 0\}\right|&= \max_{ \{\mathcal{H}\subset \mathcal{S}^*:\exists \beta\neq\beta^* \text{ s.t. }\pi_{\mathcal{H}}(\beta) = \boldsymbol{0}_{|\mathcal{H}|} \}}|\mathcal{H}| \\
&\leq p_d. 
\end{aligned}\]
This verifies (\ref{eq: max vote pd}) in Proposition \ref{prop: rare}. 
\par \noindent \underline{Proof of (b)}. We first show 
\begin{equation}\label{eq: R star 1}
\Pr\{\mathscr{R}^*\} = 1. 
\end{equation}
For any $\mathcal{H}\subset \mathcal{S}^*$ with $|\mathcal{H}|\geq p_d$, define 
\[\mathscr{R}^*_{\mathcal{H},1} = \{\Upsilon^*_{\mathcal{H},1} \neq \boldsymbol{0}_{|\mathcal{H}|} \}.\]
By continuity of the distributions, the random vector $\Upsilon^*_{\mathcal{H},1}$ lies on any singleton with probability zero. We thus have 
\[\Pr\{\mathscr{R}^*_{\mathcal{H},1}\} = 1 - \Pr\{\Upsilon^*_{\mathcal{H},1} = \boldsymbol{0}_{|\mathcal{H}|} \} = 1.\]
If $p_d=1$, then $\mathscr{R}^*=\bigcap_{\mathcal{H}\subset\mathcal{S}^*,|\mathcal{H}|\geq p_d}\mathscr{R}^*_{\mathcal{H},1}$, and
\[\Pr\{\mathscr{R}^*\}=\Pr\{\bigcap_{\mathcal{H}\subset\mathcal{S}^*,|\mathcal{H}|\geq p_d}\mathscr{R}^*_{\mathcal{H},1}\}=1,\]
given that the intersection of finitely many events that each have a probability of 1 also has a probability of 1.
Therefore, (\ref{eq: R star 1}) is verified. Otherwise, when $p_d = 2$ define 
\[\mathscr{R}^*_{\mathcal{H},2} = \{{\rm rank}\left(\Upsilon^*_{\mathcal{H},1},\Upsilon^*_{\mathcal{H},2} \right) = 2 \}.\]
Recall $|\mathcal{H}| \geq p_d$. Again by continuity of the distributions, the $|\mathcal{H}|$-dimensional random vector lies on any one-dimensional subspace with zero probability. Let $\mathscr{L}_{\mathcal{H},1}$ denote the linear space spanned by $\Upsilon^*_{\mathcal{H},1}$ whose dimension is 1. Because the nonzero elements are drawn from independent continuous distributions,  the random vectors $\Upsilon^*_{\mathcal{H},1}$ and $\Upsilon^*_{\mathcal{H},2}$ are independent. Therefore,
\begin{equation}\label{eq: Cond prob}
    \Pr\{\mathscr{R}^*_{\mathcal{H},2}|\mathscr{R}^*_{\mathcal{H},1}\} = 1 -  \Pr\{\Upsilon^*_{\mathcal{H},2}\in \mathscr{L}_{\mathcal{H},1}|\mathscr{R}^*_{\mathcal{H},1}\} = 1,
\end{equation}
which implies that 
\begin{equation}
    \Pr\{\mathscr{R}^*_{\mathcal{H},2}\} \geq \Pr\{\mathscr{R}^*_{\mathcal{H},2} \cap  \mathscr{R}^*_{\mathcal{H},1}\} = \Pr\{\mathscr{R}^*_{\mathcal{H},2} | \mathscr{R}^*_{\mathcal{H},1}\}\Pr\{\mathscr{R}^*_{\mathcal{H},1}\} = 1. 
\end{equation}
If $p_d = 2$, $\mathscr{R}^* = \bigcap_{\mathcal{H}\subset\mathcal{S}^*,|\mathcal{H}|\geq p_d}\mathscr{R}^*_{\mathcal{H},2}$ and thus (\ref{eq: R star 1}) is verified. Otherwise, when $p_d \geq 3$, we further define 
\[\mathscr{R}^*_{\mathcal{H},j} = \{{\rm rank}\left(\Upsilon^*_{\mathcal{H},1},\dots,\Upsilon^*_{\mathcal{H},j} \right) = j \}\]
for $j \geq 3$. We can show by induction that for any $j=3,\dots,p_d$, 
\[\Pr\left\{\mathscr{R}^*_{\mathcal{H},j} | \mathscr{R}^*_{\mathcal{H},j-1},\dots,\mathscr{R}^*_{\mathcal{H},1}\right\} = 1\]
following the arguments for the proof of (\ref{eq: Cond prob}). Therefore, 
\[\begin{aligned} 
\Pr\left\{\mathscr{R}^*_{\mathcal{H},p_d}\right\} &\geq \Pr\left\{\bigcap_{j=1}^{p_d}\mathscr{R}^*_{\mathcal{H},j} \right\} \\
    &=\Pr\left\{\mathscr{R}^*_{\mathcal{H},1}\right\}\cdot \prod_{j=2}^{p_d}\Pr\left\{\mathscr{R}^*_{\mathcal{H},j} | \mathscr{R}^*_{\mathcal{H},j-1},\dots,\mathscr{R}^*_{\mathcal{H},1}\right\} \\
    &=1.
\end{aligned}\]
Thus, $\Pr\left\{\mathscr{R}^*\right\} = \Pr\left\{\bigcap_{\mathcal{H}\subset\mathcal{S}^*,|\mathcal{H}|\geq p_d}\mathscr{R}^*_{\mathcal{H},p_d}\right\}=1$ and (\ref{eq: R star 1}) is verified. 
\par We then show $\Pr\{\bigcap_{\mathcal{H}\in\mathscr{H}_0}\mathscr{R}(\mathcal{H}) \} = 1$. Given that $\mathscr{H}_0$ is a finite collection, it suffices to show that $\Pr\{\mathscr{R}(\mathcal{H}) \} = 1 \text{ for any }\mathcal{H}\in\mathscr{H}_0$, where $\mathscr{R}(\mathcal{H})$ is defined in (\ref{eq: def R (H)}). Given that (\ref{eq: R star 1}) has been verified, from now on we always assume that  $\mathscr{R}^*$ is true, and omit the conditioning on $\mathscr{R}^*$ for simplicity. Under $\mathscr{R}^*$, it suffices to show
\begin{equation}\label{eq: R H 1}
    \Pr\{  \rank(\Upsilon_{\mathcal{H}\cdot}^*, \pi^*_{\mathcal{H}}) = p_d + 1\} = 1.
\end{equation}
for any $\mathcal{H}\in\mathscr{H}_0$. Note that by definition any $\mathcal{H}\in\mathscr{H}_0$ includes at least one invalid IV. 
\par Case I. When all instruments in $\mathcal{H}$ are invalid, i.e.\ $\mathcal{H}\subset \mathcal{V}^{*c}$, all entries in $\pi^*_{\mathcal{H}}$ are random.  Then we can prove (\ref{eq: R H 1}) following the arguments in proving (\ref{eq: R star 1}). 
\par Case II. When some of the instruments in $\mathcal{H}$ are valid, i.e.\ $\mathcal{H} \cap \mathcal{V}^*$ is not empty, $\pi^*_{\mathcal{H} \cap \mathcal{V}^*} = \textbf{0}_{|\mathcal{H} \cap \mathcal{V}^*|}$ is deterministic, and $\pi^*_{\mathcal{H} \cap \mathcal{V}^{*c}}$ is random. Recall that under $\mathscr{R}^*$ we have ${\rm rank}(\Upsilon^*_{\mathcal{H}\cdot}) = p_d$, and in this case the event in (\ref{eq: R H 1}) is equivalent to 
\begin{equation}\label{eq: equivalent event}
\begin{aligned}
    \{{\rm rank}(\Upsilon^*_{\mathcal{H}\cdot},\pi^*_{\mathcal{H}}) = p_d + 1\} &=  \{ {\rm rank}(\Upsilon^*_{\mathcal{H}\cdot})< {\rm rank}(\Upsilon^*_{\mathcal{H}\cdot},\pi^*_{\mathcal{H}}) \}\\
    &= \{\Upsilon^*_{\mathcal{H}\cdot}x\neq\pi^*_{\mathcal{H}}\text{ for any }x\in\mathbb{R}^{p_d}\},
\end{aligned}
\end{equation}
where the second row applies Lemma \ref{lem: matrix rank}(a).  Note that 
any $x$ such that $\Upsilon^*_{\mathcal{H}\cdot}x = \pi^*_{\mathcal{H}}$ must satisfy $\Upsilon^*_{(\mathcal{H}\cap\mathcal{V}^*)\cdot}x = \pi^*_{(\mathcal{H}\cap\mathcal{V}^*)} = \boldsymbol{0}_{|(\mathcal{H}\cap\mathcal{V}^*)|}$. Equivalently, any solution to the equation $\Upsilon^*_{\mathcal{H}\cdot}x = \pi^*_{\mathcal{H}}$ is in the set  
\[\mathcal{X}_0 = \{x\in\mathbb{R}^{p_d}: \Upsilon^*_{(\mathcal{H} \cap \mathcal{V}^*)\cdot} x = \textbf{0}_{|\mathcal{H} \cap \mathcal{V}^*|}  \}.\]
Therefore, 
\begin{equation}\label{eq: Ups x neq pi equiv}
    \{\Upsilon^*_{\mathcal{H}\cdot}x\neq\pi^*_{\mathcal{H}}\text{ for any }x\in\mathbb{R}^{p_d}\}=\{\Upsilon^*_{(\mathcal{H}\cap\mathcal{V}^{*c})\cdot}x \neq \pi^*_{(\mathcal{H}\cap\mathcal{V}^{*c})}\text{ for any }x\in\mathcal{X}_0\}. 
\end{equation} 
By (\ref{eq: equivalent event}) and (\ref{eq: Ups x neq pi equiv}), to prove (\ref{eq: R H 1}) it suffices to show 
\begin{equation}\label{eq: Ups x neq pi}
    \Pr\{\Upsilon^*_{(\mathcal{H}\cap\mathcal{V}^{*c})\cdot}x \neq \pi^*_{(\mathcal{H}\cap\mathcal{V}^{*c})}\text{ for any }x\in\mathcal{X}_0\}=1.
\end{equation}

Under $\mathscr{R}^*$, the rank of the $|\mathcal{H} \cap \mathcal{V}^*|\times p_d$ matrix $\Upsilon^*_{(\mathcal{H} \cap \mathcal{V}^*)\cdot}$ equals either $|\mathcal{H} \cap \mathcal{V}^*|$ or $p_d$. By the  Rank-Nullity Theorem, the null space $\mathcal{X}_0$ satisfies 
\begin{equation}\label{eq: dim X0 space}
   {\rm dim}(\mathcal{X}_0) = \max\{p_d - |\mathcal{H} \cap \mathcal{V}^*|,0\}. 
\end{equation}
When ${\rm dim}(\mathcal{X}_0)=0$, the space $\mathcal{X}_0 = \{\boldsymbol{0}_{p_d}\}$. Note that $\mathcal{H}$ includes at least one invalid IV, and thus $\pi^*_{(\mathcal{H}\cap\mathcal{V}^{*c})} \neq \boldsymbol{0}_{|\mathcal{H}\cap\mathcal{V}^{*c}|} = \Upsilon^*_{(\mathcal{H}\cap\mathcal{V}^{*c})\cdot}\boldsymbol{0}_{p_d}$, and thus the event in (\ref{eq: Ups x neq pi}) holds. 

Otherwise, ${\rm dim}(\mathcal{X}_0) = p_d - |\mathcal{H}\cap\mathcal{V}^*|>0$. Define the null space and range space of the linear transformation $\Upsilon^*_{(\mathcal{H}\cap\mathcal{V}^{*c})\cdot}$ as
\[\mathcal{N}_0 = \{x\in\mathcal{X}_0:\Upsilon^*_{(\mathcal{H}\cap\mathcal{V}^{*c})\cdot}x = \boldsymbol{0}_{|\mathcal{H}\cap\mathcal{V}^{*c}|} \},\]
and 
\[\mathcal{W}_0 = \{\Upsilon^*_{(\mathcal{H}\cap\mathcal{V}^{*c})\cdot}x: x\in\mathcal{X}_0\}.\]
By the Rank-Nullity Theorem,
\begin{equation}\label{eq: dim W0 bound}
    \begin{aligned}
    {\rm dim}( \mathcal{W}_0 ) &= {\rm dim}({\mathcal{X}_0}) - {\rm dim}( \mathcal{N}_0 ) \\
    &\leq  {\rm dim}({\mathcal{X}_0})  =  p_d - |\mathcal{H} \cap \mathcal{V}^*|.
\end{aligned}
\end{equation}   Recall that $|\mathcal{H}|>p_d$ for any $\mathcal{H}\in\mathscr{H}_0$, and  $\pi^*_{\mathcal{H}\cap\mathcal{V}^{*c}}$ is a random vector of length $$|\mathcal{H}| - |\mathcal{H}\cap\mathcal{V}^{*}| > p_d - |\mathcal{H}\cap\mathcal{V}^{*}| \geq {\rm dim}(\mathcal{W}_0),$$
where the last step applies (\ref{eq: dim W0 bound}). That means the length of the continuously distributed random vector $\pi^*_{\mathcal{H}\cap\mathcal{V}^{*c}}$ is larger than the dimension of $\mathcal{W}_0$.  Therefore, $\pi^*_{\mathcal{H}\cap\mathcal{V}^{*c}}$ falls onto the subspace $\mathcal{W}_0$ with probability zero. This indicates that (\ref{eq: Ups x neq pi}) holds with probability one, and therefore (\ref{eq: R H 1}) holds. This completes the proof of Part (b) and concludes the proof of Proposition \ref{prop: rare}.
\end{proof}

%% file: app/proof3.tex
\subsection{Proofs of Section \ref{subsec: theory}}
In the following proofs, we abuse the notation $C$ or $C_j$ for $j\in\mathbb{N}$ to denote absolute constants, which may vary from place to place.  

We first establish model-selection consistency in the first-stage hard thresholding for relevant IVs and the just-identifying IV subsets, respectively. 
\begin{Proposition}\label{prop: first stage}Suppose that Assumptions \ref{assu: distribution} and \ref{assu: finite sample rank} hold. Then $\lim_{n\to\infty}\Pr\{\hat{\mathcal{S}} = \mathcal{S}^*\} = 1$, and  $\lim_{n\to\infty}\Pr\{\hat{\mathscr{H}} = \mathscr{H}^*\} = 1$. 
\end{Proposition}
\begin{proof}[Proof of Proposition \ref{prop: first stage}]
According to the classical central limit theorems for ordinary least squares, the estimators of reduced-form coefficients $\hat\Gamma$ and $\hat\Upsilon$ satisfy the following asymptotic normality 
\begin{equation}
    \dfrac{\hat\Upsilon_{k,j}-\Upsilon^*_{k,j}}{\hat{\rm SE}(\hat\Upsilon_{k,j})} \convd N(0,1),
\end{equation}
and the standard error $\hat{\rm SE}(\hat\Upsilon_{k,j}) = O_p(n^{-1/2})$.
Therefore, for any $k\in[p_z]$, if $k\notin \mathcal{S}^*$,  $\Upsilon^*_{kj} = 0$, and thus 
\[|\hat{\Upsilon}_{kj}/\hat{\rm SE}(\hat\Upsilon_{k,j})|  = o_p\left(\sqrt{ \log n}\right), \]
indicating that for any   absolute constant $c$,
\begin{equation}\label{eq: prob first stage 1}
\Pr(|\hat{\Upsilon}_{kj}/\hat{\rm SE}(\hat\Upsilon_{k,j})| > \sqrt{\log n} ) \to 0\text{ for any }k\notin\mathcal{S}^*. 
\end{equation}
Besides, recall that $\|\Upsilon^*_{k\cdot}\|_\infty \gg \sqrt{\log n/n}$ for any $k \in \mathcal{S}^*$  by Assumption \ref{assu: finite sample rank}. Given that $\hat{\rm SE}(\hat\Upsilon_{k,j}) = O_p(n^{-1/2})$, we have $\max_{j\in[p_d]}|\Upsilon^*_{k,j}|/\hat{\rm SE}(\hat\Upsilon_{k,j}) > 2\sqrt{\log n}$ with probability approaching one. Therefore, 
\[\begin{aligned} \max_{j\in[p_d]}|\hat{\Upsilon}_{kj}|/ \hat{\rm SE}(\hat\Upsilon_{k,j})&\geq \max_{j\in[p_d]}|\Upsilon_{kj}^*|/\hat{\rm SE}(\hat\Upsilon_{k,j}) - \max_{j\in[p_d]}|\hat{\Upsilon}_{kj} - \Upsilon_{kj}^* | /\hat{\rm SE}(\hat\Upsilon_{k,j}) \\
&\geq 2\sqrt{\log n} - O_p(1) \geq  \sqrt{\log n}
\end{aligned}\]
with probability approaching one. Therefore, 
\begin{equation}\label{eq: prob first stage 2} 
\Pr\left(\max_{j\in[p_d]}|\hat{\Upsilon}_{kj}/\hat{\rm SE}(\hat\Upsilon_{k,j})| > \sqrt{\log n} \right) \to 1\text{ for any }k\in\mathcal{S}^*. 
\end{equation}
The limits (\ref{eq: prob first stage 1}) and (\ref{eq: prob first stage 2}) indicate that $\Pr(\hat{\mathcal{S}} = \mathcal{S}^*) \to 1$ as $n\to\infty$. 

We then prove $\Pr\{\hat{\mathscr{H}} = \mathscr{H}^*\} \to 1$. When $\mathcal{H}\in\mathscr{H}^*$, $\lambda_{\min}( \Upsilon_{\mathcal{H}\cdot}^{*\top}  \Upsilon_{\mathcal{H}\cdot}^*)$ is bounded away from zero. In addition, standard limit theory yields that 
\[\lambda_{\min}(\hat\Upsilon_{\mathcal{H}\cdot}^\top \hat\Upsilon_{\mathcal{H}\cdot}) - \lambda_{\min}( \Upsilon_{\mathcal{H}\cdot}^{*\top}  \Upsilon_{\mathcal{H}\cdot}^*)= O_p(n^{-1/2}),\]
and the eigenvalues of $\hat\Sigma_{\varepsilon}^{-1/2}$ and $\left( Z_{\cdot \mathcal{H}}^\top \RM_{W^{(\mathcal{H})}} Z_{\cdot \mathcal{H}} \right)/n$ are bounded away from zero and infinity with probability approaching one. Therefore, for any $\mathcal{H}\in\mathscr{H}^*$, we have 
\[{\rm CD}(\mathcal{H}) \geq c\cdot n\lambda_{\min}(\hat\Upsilon_{\mathcal{H}\cdot}^\top \hat\Upsilon_{\mathcal{H}\cdot}) \geq c^\prime \cdot n \gg \log n\]
for some absolute constant $c$ with probability approaching  one. That means $\Pr\{\mathcal{H}\in\hat{\mathscr{H}}\} \to 1$. Also, for  any $\mathcal{H}\notin\mathscr{H}^*$, we have ${\rm rank}(\Upsilon^*_{\mathcal{H}\cdot}) \leq p_d - 1$. Therefore, there exists a nonzero vector $v\in\mathbb{R}^{p_d}$ such that $ \|v\|_2= 1$ and $\Upsilon^*_{\mathcal{H}\cdot}v = 0$. By the law of large numbers and central limit theorem for OLS, we have 
\[\hat \Upsilon_{\mathcal{H}\cdot}v = \hat \Upsilon_{\mathcal{H}\cdot}v  -\Upsilon^*_{\mathcal{H}\cdot}v = O_p(n^{-1/2}). \]
Therefore, 
\[{\rm CD}(\mathcal{H}) \leq C\cdot n\lambda_{\min}(\hat \Upsilon_{\mathcal{H}\cdot}^\top \hat \Upsilon_{\mathcal{H}\cdot}) \leq C\cdot nv^\top\hat \Upsilon_{\mathcal{H}\cdot}^\top \hat \Upsilon_{\mathcal{H}\cdot}v = O_p(1)\ll \log n\]
where $C$ is an absolute constant. Therefore, $\Pr\{\mathcal{H}\notin\hat{\mathscr{H}}\} \to 1$. We complete the proof of $\Pr\{\hat{\mathscr{H}} = \mathscr{H}^*\} \to 1$.
\end{proof}
\begin{proof}[Proof of Proposition \ref{prop: m star}]
Define the event 
\begin{equation}\label{eq: def E1}
    \mathcal{E}_1 = \left\{\hat{\mathcal{S}} ={\mathcal{S}}^*,\hat{\mathscr{H}} = \mathscr{H}^* \right\}.
\end{equation} 
We have $\Pr\{\mathcal{E}_1\}\to 1$ by Proposition \ref{prop: first stage}.  
Let $\mathcal{O}$ denote the observed data, and define 
    \begin{equation}\label{eq: U hat U m def}
        \hat U =\left(\textbf{1}\{k\notin\hat{\mathcal{H}}_{\ell^*}\}\cdot \frac{\hat\pi_k^{({\ell^*}) } - \pi_k^* } {\hat{\rm SE}(\hat\pi_k^{({\ell^*}) })}\right)_{k\in\hat{\mathcal{S}}}\text{ and }U^{[m]} = -\xi^{[m]}.
    \end{equation}
    Furthermore, define the following event for the data $\mathcal{O}$:
\begin{equation}
\begin{aligned}
    \mathcal{E}_2 &= \left\{\|\hat U\|_\infty \leq   z_{1-\alpha_0/(2|\mathcal{S}^*|)}\right\},\\ 
    \mathcal{E} &= \mathcal{E}_1\cap\mathcal{E}_2.\\ 
\end{aligned}
\end{equation}
By central limit theorem we can show that $\mathop{\lim\inf}\limits_{n\to\infty}\Pr(\mathcal{E}_2) \geq 1- \alpha_0$. Therefore, 
\[\mathop{\lim\inf}\limits_{n\to\infty}\Pr(\mathcal{E}) \geq 1 - \alpha_0.\]
Recall that $\widehat{U}$ is a function of the observed data $\mathcal{O}$. Also, $U^{[m]}_{k} = -\xi^{[m]}_k \sim \mathcal{N}(0,1)$. Let $f(\cdot \mid \mathcal{O})$ denote the conditional density function of $U^{[m]}$ given the data $\mathcal{O}$, that is,

$$
f\left(U^{[m]}=U \mid \mathcal{O}\right)=\prod_{k\in \hat{\mathcal{S}}} \frac{1}{\sqrt{2 \pi}} \exp \left(-\frac{U_{k}^{2}}{2}\right)
$$
On the event $\mathcal{E}$, we have
$$
\sum_{k\in \hat{\mathcal{S}}} \frac{\widehat{U}_{k}^{2}}{2} \leq \frac{|\mathcal{S}^*|}{2}\left[z_{1-\alpha_0/(2|\mathcal{S}^*|)}\right]^{2},
$$
and further establish
\begin{equation}\label{eq: lower f}
f\left(U^{[m]}=\widehat{U} \mid \mathcal{O}\right) \cdot \mathbf{1}_{\mathcal{O} \in \mathcal{E}} \geq c(\alpha_0):=\left(\frac{1}{\sqrt{2 \pi}}\right)^{|\mathcal{S}^*|} \exp \left(-\frac{ |\mathcal{S}^*|}{2}\left[z_{1-\alpha_0/(2|\mathcal{S}^*|)}\right]^{2}\right) 
\end{equation}
We use $\Pr(\cdot \mid \mathcal{O})$ to denote the conditional probability with respect to the observed data $\mathcal{O}$. Note that
$$
\begin{aligned}
& \Pr\left(\min _{1 \leq m \leq M}\left\|U^{[m]}-\widehat{U}\right\|_{\infty} \leq \rho_n(M) \mid \mathcal{O}\right) \\
& =1-\Pr\left(\min _{1 \leq m \leq M}\left\|U^{[m]}-\widehat{U}\right\|_{\infty} \geq \rho_n(M) \mid \mathcal{O}\right) \\
& =1-\prod_{m=1}^{M}\left[1-\Pr\left(\left\|U^{[m]}-\widehat{U}\right\|_{\infty} \leq \rho_n(M) \mid \mathcal{O}\right)\right] \\
& \geq 1-\exp \left[-\sum_{m=1}^{M} \Pr\left(\left\|U^{[m]}-\widehat{U}\right\|_{\infty} \leq \rho_n(M) \mid \mathcal{O}\right)\right]
\end{aligned}
$$
where the second equality follows from the conditional independence of $\left\{U^{[m]}\right\}_{1 \leq m \leq M}$ given the observed data $\mathcal{O}$,  and the last inequality follows from $1-x \leq e^{-x}$. Applying the inequality above, we can further establish
\begin{equation}
\begin{aligned}\label{eq: p hat U Um}
& \Pr\left(\min _{1 \leq m \leq M}\left\|U^{[m]}-\widehat{U}\right\|_{\infty} \leq \rho_n(M) \mid \mathcal{O}\right) \cdot \mathbf{1}_{\mathcal{O} \in \mathcal{E}} \\
& \geq\left(1-\exp \left[-\sum_{m=1}^{M} \Pr\left(\left\|U^{[m]}-\widehat{U}\right\|_{\infty} \leq \rho_n(M) \mid \mathcal{O}\right)\right]\right) \cdot \mathbf{1}_{\mathcal{O} \in \mathcal{E}}  \\
& =1-\exp \left[-\sum_{m=1}^{M} \Pr\left(\left\|U^{[m]}-\widehat{U}\right\|_{\infty} \leq \rho_n(M) \mid \mathcal{O}\right) \cdot \mathbf{1}_{\mathcal{O} \in \mathcal{E}}\right]
\end{aligned}
\end{equation}
For the remainder of the proof, we establish a lower bound for
\begin{equation}\label{eq: p hat u inicator}
\Pr\left(\left\|U^{[m]}-\widehat{U}\right\|_{\infty} \leq \rho_n(M) \mid \mathcal{O}\right) \cdot \mathbf{1}_{\mathcal{O} \in \mathcal{E}} 
\end{equation}
and apply (\ref{eq: p hat U Um}) to establish a lower bound for
$$
\Pr\left(\min _{1 \leq m \leq M}\left\|U^{[m]}-\widehat{U}\right\|_{\infty} \leq \rho_n(M) \mid \mathcal{O}\right)
$$
We further decompose the targeted probability in (\ref{eq: p hat u inicator}) as
\begin{equation}\label{eq: decom}
\begin{aligned}
& \Pr\left(\left\|U^{[m]}-\widehat{U}\right\|_{\infty} \leq \rho_n(M) \mid \mathcal{O}\right) \cdot \mathbf{1}_{\mathcal{O} \in \mathcal{E}} \\
= & \int f\left(U^{[m]}=U \mid \mathcal{O}\right) \cdot \mathbf{1}_{\left\{\|U-\widehat{U}\|_{\infty} \leq \rho_n(M)\right\}} d U \cdot \mathbf{1}_{\mathcal{O} \in \mathcal{E}} \\
= & \int f\left(U^{[m]}=\widehat{U} \mid \mathcal{O}\right) \cdot \mathbf{1}_{\left\{\|U-\hat{U}\|_{\infty} \leq \rho_n(M)\right\}} d U \cdot \mathbf{1}_{\mathcal{O} \in \mathcal{E}} \\
& +\int\left[f\left(U^{[m]}=U \mid \mathcal{O}\right)-f\left(U^{[m]}=\widehat{U} \mid \mathcal{O}\right)\right] \cdot \mathbf{1}_{\left\{\|U-\widehat{U}\|_{\infty} \leq \rho_n(M)\right\}} d U \cdot \mathbf{1}_{\mathcal{O} \in \mathcal{E}} .
\end{aligned}
\end{equation}
By (\ref{eq: lower f}), we establish
\begin{equation}\label{eq: lower f 2}
\begin{aligned}
& \int f\left(U^{[m]}=\widehat{U} \mid \mathcal{O}\right) \cdot \mathbf{1}_{\left\{\|U-\widehat{U}\|_{\infty} \leq \rho_n(M)\right\}} d U \cdot \mathbf{1}_{\mathcal{O} \in \mathcal{E}} \\
& \geq c(\alpha_0) \cdot \int \mathbf{1}_{\left\{\|U-\widehat{U}\|_{\infty} \leq \rho_n(M)\right\}} d U \cdot \mathbf{1}_{\mathcal{O} \in \mathcal{E}} \\
& = c(\alpha_0) \cdot\left[2 \rho_n(M)\right]^{|\mathcal{S}^*|} \cdot \mathbf{1}_{\mathcal{O} \in \mathcal{E}} .
\end{aligned}
\end{equation}
By Taylor expansion, there exists a $t \in(0,1)$ such that
$$
f\left(U^{[m]}=U \mid \mathcal{O}\right)-f\left(U^{[m]}=\widehat{U} \mid \mathcal{O}\right)=[\nabla f(\widehat{U}+t(U-\widehat{U}))]^{\top}(U-\widehat{U})
$$
with the gradient function 
$$
\nabla f(u)=-uf(u) = -\frac{1}{(2\pi)^{|\mathcal{S}^*|/2}}u\cdot \exp(-\|u\|_2^2/2).
$$
Since $\|\nabla f\|_{2}$ is  bounded from above, there exists a positive constant $C>0$ such that
$$
\left|f\left(U^{[m]}=U \mid \mathcal{O}\right)-f\left(U^{[m]}=\widehat{U} \mid \mathcal{O}\right)\right| \leq C \sqrt{|\mathcal{S}^*|}\|U-\widehat{U}\|_{\infty}
$$
Then we establish
\begin{equation}
\begin{aligned}\label{eq: f decom err}
& \left|\int\left[f\left(U^{[m]}=U \mid \mathcal{O}\right)-f\left(U^{[m]}=\widehat{U} \mid \mathcal{O}\right)\right] \cdot \mathbf{1}_{\left\{\|U-\widehat{U}\|_{\infty} \leq \rho_n(M)\right\}} d U \cdot \mathbf{1}_{\mathcal{O} \in \mathcal{E}}\right| \\
& \leq C \sqrt{|\mathcal{S}^*|} \cdot \rho_n(M) \cdot \int \mathbf{1}_{\left\{\|U-\widehat{U}\|_{\infty} \leq \rho_n(M)\right\}} d U \cdot \mathbf{1}_{\mathcal{O} \in \mathcal{E}}  \\
& =C \sqrt{|\mathcal{S}^*|} \cdot \rho_n(M) \cdot\left[2 \rho_n(M)\right]^{|\mathcal{S}^*|} \cdot \mathbf{1}_{\mathcal{O} \in \mathcal{E}}
\end{aligned}
\end{equation}
Since $\rho_n(M) \rightarrow 0$ and $c(\alpha_0)$ is a positive constant, then there exists a positive integer $M_{0}$ such that
$$
C \sqrt{|\mathcal{S}^*|} \cdot \rho_n(M) \leq \frac{1}{2} c(\alpha_0) \quad \text { for } \quad M \geq M_{0}
$$
We combine the above inequality, (\ref{eq: decom}), (\ref{eq: lower f 2}) and (\ref{eq: f decom err}) and obtain that for $M \geq M_{0}$,
$$
\begin{aligned}
\Pr\left(\|U^{[m]}-\widehat{U}\|_{\infty} \leq \rho_n(M) \mid \mathcal{O}\right) \cdot \mathbf{1}_{\mathcal{O} \in \mathcal{E}} 
&\geq \frac{1}{2} c(\alpha_0) \cdot\left[2 \rho_n(M)\right]^{|\mathcal{S}^*|} \cdot \mathbf{1}_{\mathcal{O} \in \mathcal{E}} .
\end{aligned}
$$
Together with (\ref{eq: p hat U Um}), we establish that for $M \geq M_{0}$,
\begin{align*}
& \Pr\left(\min _{1 \leq m \leq M}\left\|U^{[m]}-\widehat{U}\right\|_{\infty} \leq \rho_n(M) \mid \mathcal{O}\right) \cdot \mathbf{1}_{\mathcal{O} \in \mathcal{E}} \\
& \geq 1-\exp \left[-M \cdot \frac{1}{2} c(\alpha_0) \cdot\left[2 \rho_n(M)\right]^{|\mathcal{S}^*|} \cdot \mathbf{1}_{\mathcal{O} \in \mathcal{E}}\right] \\
& =\left(1-\exp \left[-M \cdot \frac{1}{2} c(\alpha_0) \cdot\left[2 \rho_n(M)\right]^{|\mathcal{S}^*|}\right]\right) \cdot \mathbf{1}_{\mathcal{O} \in \mathcal{E}}.
\end{align*}
In addition, we show that on the event $\mathcal{E}$, 
\begin{equation}\label{eq: m not in M}
\min_{m\notin\mathcal{M}}\|U^{[m]}-\hat U\|_\infty > \rho_n(M).
\end{equation}
This can be verified by 
\[\begin{aligned}
\min_{m\notin\mathcal{M}}\|U^{[m]}-\hat U\|_\infty &\geq \min_{m\notin\mathcal{M}}\|\xi^{[m]}\|_\infty - \|\hat U\|_\infty 
> 1.1z_{1-\alpha_0/(2|\mathcal{S}^*|)} - z_{1-\alpha_0/(2|\mathcal{S}^*|)} \\
&=0.1z_{1-\alpha_0/(2|\mathcal{S}^*|)} > \rho_n(M)
\end{aligned}\]
when $M$ is sufficiently large, where the second inequality applies (\ref{eq: screen sample}) and $\mathcal{E}_2$. (\ref{eq: m not in M}) implies that
\begin{equation}\label{eq: equiv min}
    \{\min_{m\in\mathcal{M}}\|U^{[m]}-\hat U\|_\infty\leq \rho_n(M)\} = \{\min_{1\leq m\leq M}\|U^{[m]}-\hat U\|_\infty\leq \rho_n(M)\}.
\end{equation} 
We use $\E_{\mathcal{O}}$ to denote the expectation taken with respect to the observed data $\mathcal{O}$. We further integrate with respect to $\mathcal{O}$ and establish that for $M \geq M_{0}$,
$$
\begin{aligned}
\ \ \ \ &\Pr\left(\min _{m\in\mathcal{M}}\left\|U^{[m]}-\widehat{U}\right\|_{\infty} \leq \rho_n(M)\right) \\
& =\E_{\mathcal{O}}\left[\Pr\left(\min _{m\in\mathcal{M}}\left\|U^{[m]}-\widehat{U}\right\|_{\infty} \leq \rho_n(M) \mid \mathcal{O}\right)\right] \\
& \geq \E_{\mathcal{O}}\left[\Pr\left(\min _{m\in\mathcal{M}}\left\|U^{[m]}-\widehat{U}\right\|_{\infty} \leq \rho_n(M) \mid \mathcal{O}\right) \cdot \mathbf{1}_{\mathcal{O} \in \mathcal{E}}\right] \\
& =\E_{\mathcal{O}}\left[\Pr\left(\min _{1 \leq m \leq M}\left\|U^{[m]}-\widehat{U}\right\|_{\infty} \leq \rho_n(M) \mid \mathcal{O}\right) \cdot \mathbf{1}_{\mathcal{O} \in \mathcal{E}}\right] \\
& \geq \E_{\mathcal{O}}\left[\left(1-\exp \left[-M \cdot \frac{1}{2} c(\alpha_0) \cdot\left[2 \rho_n(M)\right]^{|\mathcal{S}^*|}\right]\right) \cdot \mathbf{1}_{\mathcal{O} \in \mathcal{E}}\right],
\end{aligned}
$$
where the fourth row applies (\ref{eq: equiv min}). 

Taking $\rho_n(M)=C_0 \left[\frac{  \log n}{ M}\right]^{\frac{1}{2|\hat{\mathcal{S}}|}}$ with $C_0=\frac{1}{2}(\frac{2}{c(\alpha_0)})^{1/(| {\mathcal{S}}^*|)}$, we establish that for $M \geq M_{0}$,

$$
\Pr\left(\min _{m\in\mathcal{M}}\left\|U^{[m]}-\widehat{U}\right\|_{\infty} \leq \rho_n(M)\right) \geq\left(1-n^{-1}\right) \cdot \Pr(\mathcal{E}).
$$
Therefore, 
\[\begin{aligned}
\mathop{\lim\inf}_{n\to\infty}\lim_{M\to\infty}\Pr\left(\min_{m\in\mathcal{M}}\|\hat U - U^{[m]}\|_\infty \leq \rho_n(M) \right)  \geq \lim\inf_{n\to\infty}\Pr(\mathcal{E}) \geq 1 - \alpha_0.
\end{aligned}\]
This concludes the proof of Proposition \ref{prop: m star}.
\end{proof}

\begin{proof}[Proof of Theorem \ref{thm: cover}]
Define 
\[\hat V = (\hat V_k)_{k\in\hat{\mathcal{S}}},\ \hat V_k = \hat U_k{\bf 1}\{k\in\hat{\mathcal{S}}\backslash\hat{\mathcal{H}}_{\ell^*}\},\]
and 
\[V^{[m]} = (V^{[m]}_k)_{k\in\hat{\mathcal{S}}},\ V^{[m]}_k = U^{[m]}_k{\bf 1}\{k\in\hat{\mathcal{S}}\backslash\hat{\mathcal{H}}_{\ell^*}\},\]
where $\hat U$ and $U^{[m]}$ are defined in (\ref{eq: U hat U m def}). Furthermore, define the event 
\begin{equation}\label{eq: event M}
    \mathcal{E}_3 = \mathcal{E}_1\cap\{\min_{m\in\mathcal{M}}\|\hat V - V^{[m]}\|_\infty \leq \rho_n(M)\},
\end{equation}
where $\mathcal{E}_1$ is defined in (\ref{eq: def E1}). We have $\Pr\{\mathcal{E}_3\}\geq 1-\alpha_0 + o(1)$ by Propositions \ref{prop: first stage} and \ref{prop: m star}. Under $\mathcal{E}_3$, let $m^*$ denote the index of sampling that satisfies  
\[\|\hat V - V^{[m^*]}\|_\infty \leq \rho_n(M).\]
Define $0/0=0$ and $1/0=\infty$ to accommodate zero standard errors. Under $\mathcal{E}_3$, for any $k\in\mathcal{V}^*$, we have $\pi_k^*=0$ and thus 
\begin{equation}\label{eq: argue valid 1}
\dfrac{|\tilde\pi_k^{(\ell^*,m^*)}|}{\hat{\rm SE}(\hat\pi_k^{(\ell^*)})} = \dfrac{|\tilde\pi_k^{(\ell^*,m^*)}-\pi_k^*|}{\hat{\rm SE}(\hat\pi_k^{(\ell^*)})} \leq \|\hat V - V^{[m^*]}\|_\infty \leq \rho_n(M),
\end{equation}
and thus $k\in\tilde{\mathcal{V}}^{(\ell^*,m^*)}$. This implies $\mathcal{V}^*\subset\tilde{\mathcal{V}}^{(\ell^*,m^*)}$ and hence $|\tilde{\mathcal{V}}^{(\ell^*,m^*)}|\geq |\mathcal{V}^*| > (|\mathcal{S}^*|+p_d-1)/2$. Under $\mathcal{E}_1$, we have $\ell^*\in\hat{\mathcal{T}}^{\rm GenMaj}_{m^*}$, which means $\tilde{\mathcal{V}}^{(\ell^*,m^*)}$ is used to construct SCI. 
Then 
\begin{equation}\label{eq: prob cover bound 1}
\begin{aligned}
    \Pr\{ \beta^*_1 \in {\rm SCI}(\alpha) \} &\geq \Pr\{ (\beta^*_1 \in {\rm SCI}(\alpha))\cap \mathcal{E}_3 \} \\&\geq \Pr\{ (\beta^*_1 \in {\rm CI}(\tilde{\mathcal{V}}^{(\ell^*,m^*)},\alpha-\alpha_0))\cap \mathcal{E}_3 \}.
\end{aligned}
\end{equation}
If 
\begin{equation}\label{eq: valid proof}
\Pr\{\tilde{\mathcal{V}}^{(\ell^*,m^*)} = \mathcal{V}^*|\mathcal{E}_3\}\to 1,
\end{equation} 
we have $\Pr\{(\tilde{\mathcal{V}}^{(\ell^*,m^*)} \neq \mathcal{V}^*)\cap \mathcal{E}_3\}\to0$ and thus
\begin{equation}\label{eq: prob lb limit}
\begin{aligned}
    &\ \ \ \ \Pr\{ (\beta^*_1 \in {\rm CI}(\tilde{\mathcal{V}}^{(\ell^*,m^*)},\alpha-\alpha_0))\cap \mathcal{E}_3 \} \\
    &=\Pr\{ (\beta^*_1 \in {\rm CI}(\tilde{\mathcal{V}}^{(\ell^*,m^*)},\alpha-\alpha_0))\cap \mathcal{E}_3 \cap (\tilde{\mathcal{V}}^{(\ell^*,m^*)} = \mathcal{V}^*)\}\\
    &\ \ \ \ +\Pr\{ (\beta^*_1 \in {\rm CI}(\tilde{\mathcal{V}}^{(\ell^*,m^*)},\alpha-\alpha_0))\cap \mathcal{E}_3 \cap(\tilde{\mathcal{V}}^{(\ell^*,m^*)} \neq \mathcal{V}^*)\} \\
    &= \Pr\{ (\beta^*_1 \in {\rm CI}(\mathcal{V}^*,\alpha-\alpha_0))\cap \mathcal{E}_3 \} +o(1)\\
    &\geq \Pr\{ (\beta^*_1 \in {\rm CI}(\mathcal{V}^*,\alpha-\alpha_0)) \} - \left(1-\Pr(  \mathcal{E}_3)\right)+o(1). 
\end{aligned}
\end{equation} 
By the asymptotic normality of TSLS estimators, we have 
\begin{equation}\label{eq: asymp norm tsls}
    \lim_{n\to\infty}\Pr\{ (\beta^*_1 \in {\rm CI}(\mathcal{V}^*,\alpha-\alpha_0)) \} = 1- \alpha + \alpha_0.
\end{equation}
 By (\ref{eq: prob cover bound 1}), (\ref{eq: prob lb limit}), and (\ref{eq: asymp norm tsls}), we have 
\[\begin{aligned}
    &\ \ \ \  \mathop{\lim\inf}_{n\to\infty}\lim_{M\to\infty}\Pr\{ \beta^*_1 \in {\rm SCI}(\alpha) \} \\
    &\geq \mathop{\lim\inf}_{n\to\infty}\lim_{M\to\infty}\left[\Pr\{ (\beta^*_1 \in {\rm CI}(\mathcal{V}^*,\alpha-\alpha_0)) \} - \left(1-\Pr(  \mathcal{E}_3)\right)\right]\\
    &= (1 - \alpha + \alpha_0) - 1 +\mathop{\lim\inf}_{n\to\infty}\lim_{M\to\infty}\Pr(\mathcal{E}_3)   \\
    &\geq 1 - \alpha, 
\end{aligned} \]
where the last step applies Proposition \ref{prop: m star}. 
\par It then suffices to show (\ref{eq: valid proof}). Define $0/0=0$ and $1/0=\infty$ to accommodate zero standard errors. Under $\mathcal{E}_3$, for any $k\in\mathcal{V}^*$, by (\ref{eq: argue valid 1}) we have $k \in \tilde{\mathcal{V}}^{(\ell^*,m^*)}$. For $k\notin\mathcal{V}^*$, we have $\pi_k^*\neq 0$. 
Define $t_{\min} = \min_{k\in \hat{\mathcal{S}}\backslash \mathcal{V}^*}(|\pi_k^*|/\hat{\rm SE}(\hat{\pi}_k^{(\ell^*)}))$. With a sufficiently large $M$, we have $t_{\min} > 2\cdot\rho_n(M)$. 
Therefore, 
\begin{equation}\label{eq: argue valid 2}
\dfrac{|\tilde\pi_k^{(\ell^*,m^*)}|}{\hat{\rm SE}(\hat\pi_k^{(\ell^*)})} \geq t_{\min}  - \|\hat V - V^{[m^*]}\|_\infty > 2\rho_n(M)-\rho_n(M) = \rho_n(M),
\end{equation}
which implies $k\notin\tilde{\mathcal{V}}^{(\ell^*,m^*)}$. Therefore, (\ref{eq: argue valid 1}) and (\ref{eq: argue valid 2}) imply (\ref{eq: valid proof}), which ends the proof of Theorem \ref{thm: cover}.  
\end{proof}

\begin{proof}[Proof of Theorem \ref{thm: length}] Note that under $\mathcal{E}_1$, for all $k\in\mathcal{S}^*$, $\ell\in[ |\mathscr{H}^*|]$ and any $m\in\mathcal{M}$, 
\begin{equation}\label{eq: tilde pi hat pi}
    \begin{aligned}
    |\tilde\pi_k^{(\ell,m)} - \hat\pi_k^{(\ell)}|  = |\xi^{[m]}_k|\cdot\hat{\rm SE}(\hat\pi_k^{(\ell)}) &\leq \hat{\rm SE}(\hat\pi_k^{(\ell)})\cdot 1.1z_{1-\alpha_0/(2|\hat{\mathcal{S}}|)}\\
    &\leq \hat{\rm SE}(\hat\pi_k^{(\ell)})\cdot C\log(2|\hat{\mathcal{S}}|/\alpha_0) =  \hat{\rm SE}(\hat\pi_k^{(\ell)})\cdot C\log(2|\mathcal{S}^*|/\alpha_0)
\end{aligned}
\end{equation} 
w.p.a.1 for some absolute constant $C>0$. In addition, define 
\begin{equation}\label{eq: E4 event}
    \mathcal{E}_4 = \bigcap_{k\in\mathcal{S}^*,\ell\in[|\mathscr{H}^*|]}\mathcal{E}_{4,k \ell},\text{ where }\mathcal{E}_{4,k \ell} = \left\{| \hat\pi_k^{(\ell)} - \pi^{(\ell)}_k|  \leq  \hat{\rm SE}(\hat\pi_k^{(\ell)}) z_{1-\alpha_0/(8|\mathscr{H}^*||\mathcal{S}^*|)} \right\}.
\end{equation} 
By asymptotic normality we have 
\begin{equation}\label{eq: prob E4}
    \begin{aligned}
    \Pr\left\{\mathcal{E}_4\right\} &= 1 - \Pr\{ \bigcup_{k\in\mathcal{S}^*,\ell\in[|\mathscr{H}^*|]}\mathcal{E}_{4,k\ell}^c\} \\
    &\geq 1-|\mathscr{H}^*||\mathcal{S}^*|\cdot \alpha_0 /(4|\mathscr{H}^*||\mathcal{S}^*|) + o(1) \geq 1-  \alpha_0 / 3
\end{aligned}
\end{equation} 
with a sufficiently large $n$. Under $\mathcal{E}_1$, $\mathcal{E}_4$, and (\ref{eq: tilde pi hat pi}), we have for any $m\in\mathcal{M}$
\begin{equation}\label{eq: tilde pi - pi}
\begin{aligned}
    |\tilde\pi_k^{(\ell,m)} - \pi_k^{(\ell)}| &\leq  \hat{\rm SE}(\hat\pi_k^{(\ell)})\cdot \left[C\log(2|\mathcal{S}^*|/\alpha_0) + z_{1-\alpha_0/(8|\mathscr{H}^*||\mathcal{S}^*|)} \right]\\ &\leq \hat{\rm SE}(\hat\pi_k^{(\ell)})\cdot C\log(8|\mathscr{H}^*||\mathcal{S}^*|/\alpha_0).
\end{aligned}
\end{equation} 
Define 
\begin{equation}\label{eq: E5 event}
    \mathcal{E}_5 = \bigcap_{k\in\mathcal{S}^*,\ell\in[|\mathscr{H}^*|]}\mathcal{E}_{5,k\ell},\text{ where }\mathcal{E}_{5,k\ell} = \left\{   \max_{m\in\mathcal{M}}\hat{\rm SE}(\hat\beta_1^{(\tilde{\mathcal{V}}^{(\ell,m)})}) + \hat{\rm SE}(\hat\pi_k^{(\ell)})  \leq \dfrac{C}{\sqrt{n}}  \right\}.
\end{equation} 
By standard limit theory we can deduce that $\hat{\rm SE}(\hat\pi_k^{(\ell)}) = O_p(n^{-1/2})$. Combining with
 (\ref{eq: uniform bound SE}) that will be proved later, we have $\lim_{n\to\infty}\Pr\left\{\mathcal{E}_5\right\} \to 1$. Recall that Proposition \ref{prop: first stage} implies $\lim_{n\to\infty}\Pr\left\{\mathcal{E}_1\right\} \to 1$. Therefore, 
\begin{equation}\label{eq: prob E5}
    \begin{aligned}
    \Pr\left\{\mathcal{E}_1\cap\mathcal{E}_5\right\} \to 1.
\end{aligned}
\end{equation} 
 Suppose that  $n$ and $M$ are sufficiently large. When $k\in\tilde{\mathcal{V}}^{(\ell,m)}$, by the definition in (\ref{eq: sampling V hat}) we have 
 \begin{equation}\label{eq: tilde pi bound}
     |\tilde\pi_k^{(\ell,m)}| \leq \hat{\rm SE}(\hat\pi_k^{(\ell)})\cdot \rho_n(M)
 \end{equation} Thus, under $\mathcal{E}_1$, $\mathcal{E}_4$, and $\mathcal{E}_5$, for any $k\in\tilde{\mathcal{V}}^{(\ell,m)}$ we have 
\[\begin{aligned}
| \pi_k^{(\ell)}| &\leq |\tilde\pi_k^{(\ell,m)}| + |\tilde\pi_k^{(\ell,m)} -  \pi_k^{(\ell)}| 
\\ 
&\leq \hat{\rm SE}(\hat\pi_k^{(\ell)})\cdot\left[\rho_n(M)+C\log(8|\mathcal{S}^*|/\alpha_0)\right] \leq   \dfrac{C_1\log(C_2/\alpha_0)}{\sqrt{n}} 
\end{aligned}\]
for some absolute constants $C_1$ and $C_2$, where the second inequality applies  (\ref{eq: tilde pi - pi}), (\ref{eq: tilde pi bound}) and $\mathcal{E}_5$.   Therefore, if $\ell\in\hat{\mathcal{T}}^{\rm GenMaj}_m$, we have  $|\tilde{\mathcal{V}}^{(\ell,m)}|>\frac{|\mathcal{S}^*|+p_d-1}{2}$, which implies  
\[\left|\left\{k\in\mathcal{S}^*:|\pi_k^{(\ell)}|\leq \dfrac{C_1\log(C_2/\alpha_0)}{\sqrt{n}}\right\}\right|>\frac{|\mathcal{S}^*|+p_d-1}{2}.\]
By the condition (\ref{eq: finite plurality}) in Assumption \ref{assume: finite multiple plurality},  we have $\ell \notin \mathcal{SI}$. Therefore, under $\mathcal{E}_1$, $\mathcal{E}_4$, and $\mathcal{E}_5$, we have $\hat{\mathcal{T}}^{\rm GenMaj}_m\cap \mathcal{SI}=\varnothing$, which means if $\ell\in\hat{\mathcal{T}}^{\rm GenMaj}_m$ for some $m\in\mathcal{M}$, we have  
\begin{equation}\label{eq: pi k * bound}
   |\pi_k^*| \leq \dfrac{C_1\log(C_2/\alpha_0)}{\sqrt{n}} \text{ for all }k \in \hat{\mathcal{H}}_\ell. 
\end{equation}
Therefore, define 
\begin{equation*}
    \mathcal{E}_6 = \bigcap_{m\in\mathcal{M},\ell\in\hat{\mathcal{T}}^{\rm GenMaj}_m}\left\{\|\pi^*_{\hat{\mathcal{H}}_\ell}\|_\infty\leq \frac{C_1\log(C_2/\alpha_0)}{\sqrt{n}}\right\},
\end{equation*}
and the discussions from (\ref{eq: tilde pi bound}) to (\ref{eq: pi k * bound}) result in
\begin{equation}\label{eq: prob E6}
    \Pr\{\mathcal{E}_6\} \geq \Pr\{\mathcal{E}_1\cap\mathcal{E}_4\cap \mathcal{E}_5\} \geq 1 - \frac{\alpha_0}{3} + o(1). 
\end{equation}

We next state the following Lemma \ref{lem: local to true}, claiming that locally invalid IVs satisfying (\ref{eq: pi k * bound}) will finally produce TSLS estimators that are local to the true $\beta^*$. Recall that $\hat\beta^{(\mathcal{H})}$ in (\ref{eq: def generic TSLS}) denotes the TSLS estimator using the IVs in $\mathcal{H}$. Define 
\begin{equation}\label{eq: def E7}
    \mathcal{E}_7 = \bigcap_{m\in\mathcal{M},\ell\in\hat{\mathcal{T}}^{\rm GenMaj}_m}\{\|\hat\beta^{(\tilde{\mathcal{V}}^{(\ell,m)})}-\beta^*\|_\infty\leq C\log(C/\alpha_0)/\sqrt{n}\}
\end{equation}
where $C$ is an absolute constant.
\begin{Lemma}\label{lem: local to true}Under the conditions of Theorem \ref{thm: length},  we have 
\begin{equation}\label{eq: Prob E6 E7}
\mathop{\lim\inf}_{n\to\infty}\lim_{M\to\infty} \Pr\{\mathcal{E}_7 \} \geq \mathop{\lim\inf}_{n\to\infty}\lim_{M\to\infty} \Pr\{\mathcal{E}_6 \} - \alpha_0/3.
\end{equation}
\end{Lemma}
Note that the SCI in (\ref{eq: def SCI CI}) only adopts the selected valid IV set $\tilde{\mathcal{V}}^{(\ell,m)}$ indexed by $m\in\mathcal{M}$ and  $\ell\in\hat{\mathcal{T}}^{\rm GenMaj}_m$, which satisfies $|\tilde{\mathcal{V}}^{(\ell,m)}|>\frac{|\hat{\mathcal{S}}|+p_d-1}{2}$ according to (\ref{eq: T t}). This implies (\ref{eq: pi k * bound}) under $\mathcal{E}_1$, $\mathcal{E}_4$, and $\mathcal{E}_5$. By Lemma \ref{lem: local to true}, all TSLS estimators using the selected valid IV sets $\tilde{\mathcal{V}}^{(\ell,m)}$ for SCI in (\ref{eq: def SCI CI}) satisfy (\ref{eq: def E7}) with probability at least $1-2\alpha_0/3$. 
By definition of the TSLS confidence intervals in (\ref{eq: def generic CI}), the lower bounds and upper bounds of the TSLS confidence intervals are \[{\rm LB}^{(\ell,m)} = \hat\beta_1^{(\tilde{V}^{(\ell,m)})}-z_{1-(\alpha-\alpha_0)/2}\hat{\rm SE}(\hat\beta_1^{(\tilde{V}^{(\ell,m)})}),\text{ and }{\rm UB}^{(\ell,m)} = \hat\beta_1^{(\tilde{V}^{(\ell,m)})}+z_{1-(\alpha-\alpha_0)/2}\hat{\rm SE}(\hat\beta_1^{(\tilde{V}^{(\ell,m)})}).\]
We next argue that the event 
\begin{equation}\label{eq: uniform bound SE}
    \mathcal{E}_8=\{\text{$\hat{\rm SE}(\hat\beta_1^{(\tilde{V}^{(\ell,m)})})\leq C n^{-1/2}$ uniformly for any $m\in\mathcal{M}$ and $\ell\in\hat{\mathcal{T}}^{\rm GenMaj}_m$}\}
\end{equation}
holds with probability approaching one. 
First, note that by construction $\tilde\pi^{(\ell,m)}_k=0$ for all $k \in \hat{\mathcal{H}}_\ell$, and thus $\hat{\mathcal{H}}_\ell\subset \tilde{\mathcal{V}}^{(\ell,m)}$. That means $\tilde{\mathcal{V}}^{(\ell,m)}$ always includes a just-identifying IV subset with high probability, and thus
$\hat{\rm SE}(\hat\beta_1^{(\tilde{V}^{(\ell,m)})}) \leq C^{(\tilde{V}^{(\ell,m)})} n^{-1/2}$ for each fixed $(\ell,m)$ with probability approaching one by standard law of large numbers, where $C^{(\tilde{V}^{(\ell,m)})}$ is a positive constant relevant to $\tilde{V}^{(\ell,m)}$. Second, the set $\tilde{V}^{(\ell,m)}$ is a subset of $[p_z]$, meaning the number of its possible configurations is finite. Therefore, $C^{(\tilde{V}^{(\ell,m)})}$ can be replaced by an absolute constant $C$, and the bound  in (\ref{eq: uniform bound SE}) holds uniformly. Therefore, $\Pr\{\mathcal{E}_8\}\to 1$.

Then under $\mathcal{E}_7$ and $\mathcal{E}_8$, we have   
\[\begin{aligned}
\max_{m\in\mathcal{M},\ell\in\hat{\mathcal{T}}^{\rm GenMaj}_m}  |{\rm LB}^{(\ell,m)} - \beta_1^*| &\leq \max_{m\in\mathcal{M},\ell\in\hat{\mathcal{T}}^{\rm GenMaj}_m} \left[ \|\hat\beta^{(\tilde{V}^{(\ell,m)})} -\beta^*\|_\infty+z_{1-(\alpha-\alpha_0)/2}\hat{\rm SE}(\hat\beta_1^{(\tilde{V}^{(\ell,m)})})\right]\\
    &\leq \dfrac{C_3\log(C_4/\alpha_0)}{2\sqrt{n}}
\end{aligned}\]
for some absolute constants $C_3>0$ and $C_4>0$. Similarly, 
\[\max_{m\in\mathcal{M},\ell\in\hat{\mathcal{T}}^{\rm GenMaj}_m}  |{\rm UB}^{(\ell,m)} - \beta_1^*|\leq \dfrac{C_3\log(C_4/\alpha_0)}{2\sqrt{n}}.\]
Therefore, under  $\mathcal{E}_7$ and $\mathcal{E}_8$,
\begin{equation}\label{eq: proof length bound}
    \begin{aligned}
{\rm Len}(\alpha) &\leq \max_{m\in\mathcal{M},\ell\in\hat{\mathcal{T}}^{\rm GenMaj}_m}{\rm UB}^{(\ell,m)} - \min_{m\in\mathcal{M},\ell\in\hat{\mathcal{T}}^{\rm GenMaj}_m}{\rm LB}^{(\ell,m)} \\
&\leq \max_{m\in\mathcal{M},\ell\in\hat{\mathcal{T}}^{\rm GenMaj}_m}|{\rm UB}^{(\ell,m)}-\beta_1^*| + \max_{m\in\mathcal{M},\ell\in\hat{\mathcal{T}}^{\rm GenMaj}_m}|{\rm LB}^{(\ell,m)}-\beta_1^*| \\
&\leq \dfrac{C_3\log(C_4/\alpha_0)}{\sqrt{n}}.
\end{aligned}
\end{equation} 
Because  $\Pr\{\mathcal{E}_7\}\geq 1-2\alpha_0/3+o(1)$ by (\ref{eq: prob E6}) and (\ref{eq: Prob E6 E7}), and $\Pr\{\mathcal{E}_8\}\to 1$,  we have $\Pr\{\mathcal{E}_7\cap\mathcal{E}_8\}\geq 1- \alpha_0+o(1)$. Therefore, (\ref{eq: proof length bound}) holds with probability no smaller than $1-\alpha_0+o(1)$, which verifies (\ref{eq: len}) and concludes the proof of Theorem \ref{thm: length}. 
\end{proof}

\begin{proof}[Proof of Lemma \ref{lem: local to true}]We first show that 
\begin{equation}\label{eq: local to true just}
    \Pr\left\{\left(\max_{m\in\mathcal{M},\ell\in\hat{\mathcal{T}}^{\rm GenMaj}_m}\|\hat\beta^{(\hat{\mathcal{H}}_\ell)} - \beta^*\|_2 \leq \frac{C\log(C/\alpha_0)
    }{\sqrt{n}}\right)\cap\mathcal{E}_6\right\} \geq  \Pr(\mathcal{E}_6) - \frac{\alpha_0}{9} + o(1).
\end{equation}
Define $\theta^{*(\mathcal{H})}=(\beta^{*\top},\pi^{*\top}_{[p_z]\backslash\mathcal{H}},\varphi^{*\top})^\top$ for a generic subset $\mathcal{H}$. By definition in (\ref{eq: def generic TSLS}), we have  
\begin{equation*} 
\begin{aligned}\hat\theta^{(\hat{\mathcal{H}}_\ell)} &= (D^{(\hat{\mathcal{H}}_\ell)\top} \RP_W D^{(\hat{\mathcal{H}}_\ell)})^{-1}D^{(\hat{\mathcal{H}}_\ell)\top} \RP_W (D^{(\hat{\mathcal{H}}_\ell)}\theta^{*(\hat{\mathcal{H}}_\ell)}+Z_{\cdot\hat{\mathcal{H}}_\ell}\pi^*_{\hat{\mathcal{H}}_\ell}+u) \\
&= \theta^{*(\hat{\mathcal{H}}_\ell)} + (D^{(\hat{\mathcal{H}}_\ell)\top} \RP_W D^{(\hat{\mathcal{H}}_\ell)})^{-1}D^{(\hat{\mathcal{H}}_\ell)\top} \RP_W  Z_{\cdot\hat{\mathcal{H}}_\ell}\pi^*_{\hat{\mathcal{H}}_\ell} + (D^{(\hat{\mathcal{H}}_\ell)\top} \RP_W D^{(\hat{\mathcal{H}}_\ell)})^{-1}D^{(\hat{\mathcal{H}}_\ell)\top} \RP_W u.
\end{aligned}
\end{equation*} 
Under $\mathcal{E}_6$,  (\ref{eq: pi k * bound}) is valid, and thus the $L_2$-norm of the second term is bounded by $\frac{C\log(C/\alpha_0)}{2\sqrt{n}}$. By central limit theorems for i.i.d.~ variables with finite fourth moments, we can show that the third term is  asymptotically normal under Assumptions \ref{assu: distribution} and \ref{assu: finite sample rank}. Therefore, it is bounded by $\frac{C\log(C/\alpha_0)}{2\sqrt{n}}$ with probability at least $1-\alpha_0/9+o(1)$ when $C$ is sufficiently large. As $p_z$ is fixed, there are finitely many possible just-identifying subsets and thus the convergence holds uniformly for all $\ell$.  Therefore,
\begin{equation*}
\|\hat\beta^{(\hat{\mathcal{H}}_\ell)}-\beta^*\|_2 \leq \|\hat\theta^{(\hat{\mathcal{H}}_\ell)}-\theta^{*(\hat{\mathcal{H}}_\ell)}\|_2 \leq \frac{C\log(C/\alpha_0)}{\sqrt{n}}
\end{equation*}
uniformly for all $\ell$ with probability at least $ \Pr(\mathcal{E}_6)-\alpha_0/9+o(1) $, which verifies (\ref{eq: local to true just}). 
 
 Also, note that 
 the OLS estimators $\hat\Gamma $ and $\hat\Upsilon$ are also asymptotically normal centered at the corresponding true coefficients. We thus have  
 \begin{equation}\label{eq: OLS prob bound}
    \Pr\left\{ \|\hat\Gamma -\Gamma^*\|_\infty + \|\hat\Upsilon - \Upsilon^*\|_2 \leq \dfrac{C\log (C/\alpha_0)}{\sqrt{n}} \right\} \geq 1 - \frac{\alpha_0}{9} + o(1).
 \end{equation}
 By (\ref{eq: local to true just}) and (\ref{eq: OLS prob bound}), we have with probability no smaller than $\Pr(\mathcal{E}_6)-2\alpha_0/9+o(1) $, 
 \begin{equation}\label{eq: hat pi ell pi *}
     \begin{aligned}
     \|\hat\pi^{(\ell)}-\pi^*\|_\infty &= \|(\hat\Gamma-\hat\Upsilon\hat\beta^{(\hat{\mathcal{H}}_\ell)})-(\Gamma^*-\Upsilon^*\beta^*)\|_\infty\\
     &\leq  \|\hat\Gamma-\Gamma^*\|_\infty + \|\hat\Upsilon - \Upsilon^*\|_2\|\beta^*\|_2+ \|\Upsilon^*\|_2\|\hat\beta^{(\hat{\mathcal{H}}_\ell)} - \beta^*\|_2 + \|\hat\Upsilon - \Upsilon^*\|_2\|\hat\beta^{(\hat{\mathcal{H}}_\ell)} - \beta^*\|_2\\
     &\leq \dfrac{C\log (C/\alpha_0)}{\sqrt{n}}. 
 \end{aligned}
 \end{equation} 
 Therefore, with probability no smaller than $\Pr(\mathcal{E}_6)-2\alpha_0/9+o(1) $, for any $k\in\tilde{\mathcal{V}}^{(\ell,m)}$ 
 \[\begin{aligned}
     |\pi^*_k| \leq |\pi_k^*-\hat\pi_k^{(\ell)}| + |\hat\pi_k^{(\ell)}-\tilde\pi_k^{(\ell,m)}| + |\tilde\pi_k^{(\ell,m)}|\leq \dfrac{C\log (C/\alpha_0)}{\sqrt{n}},
 \end{aligned}\]
 where the first term is bounded by (\ref{eq: hat pi ell pi *}), the second term is bounded by (\ref{eq: tilde pi hat pi}), and the third term is bounded using the definition of $\tilde{\mathcal{V}}^{(\ell,m)}$. Therefore, define the event 
 \begin{equation}\label{eq: E9}
     \mathcal{E}_9 = \bigcap_{m\in\mathcal{M},\ell\in\hat{\mathcal{T}}^{\rm GenMaj}_m}\left\{|\pi^*_k|\leq \dfrac{C\log (C/\alpha_0)}{\sqrt{n}}\text{ for all }k\in\tilde{\mathcal{V}}^{(\ell,m)}\right\},
 \end{equation}
 and we have $\Pr(\mathcal{E}_9) \geq \Pr(\mathcal{E}_6)-2\alpha_0/9 +o(1)$. Following the idea of showing (\ref{eq: local to true just}), we can prove that $\|\hat{\beta}^{(\tilde{V}^{(\ell,m)})}-\beta^*\|_2 \leq C\log(C/\alpha_0)/\sqrt{n}$ with probability no smaller than 
 \[\Pr\{\mathcal{E}_9\} - \alpha_0/9+o(1) \geq \Pr\{\mathcal{E}_6\} -\alpha_0/3 + o(1).\]
 This verifies (\ref{eq: Prob E6 E7}) and concludes the proof of Lemma \ref{lem: local to true}.
\end{proof}

%% file: app/data_app.tex
\newcommand{\SE}{{\rm SE}}
\newcommand{\bC}{{\textbf{C}}}
\section{SCI for GWAS Summary Data}\label{app: data}
\subsection{Detailed Information of the SNPs IVs}

This section includes two tables of detailed information on the associations of the selected SNPs with LDL and HDL, respectively. 
\begin{table}[h]
\caption{ Detailed information on the associations of the selected SNPs with low-density lipoprotein (LDL) cholesterol.}
\centering 
\resizebox{\textwidth}{!}{%
\begin{tabular}{llllllcc}
\toprule
SNP & CHR & POS & BETA & SE & P-value & A1 & A2 \\
\midrule
rs12748152 & $1$  & $27138393$  & $0.050$  & $0.007$ & $3.209\times10^{-12}$  & T & C \\
rs17630235 & $12$ & $112591686$ & $-0.026$ & $0.004$ & $5.614\times10^{-11}$  & A & G \\
rs1800961  & $20$ & $43042364$  & $-0.069$ & $0.011$ & $6.034\times10^{-10}$  & T & C \\
rs1883025  & $9$  & $107664301$ & $-0.030$ & $0.004$ & $6.140\times10^{-11}$  & T & C \\
rs2075650  & $19$ & $45395619$  & $0.177$  & $0.006$ & $1.000\times10^{-200}$ & G & A \\
rs2642438  & $1$  & $220970028$ & $0.035$  & $0.004$ & $7.318\times10^{-16}$  & G & A \\
rs3800406  & $6$  & $35133074$  & $-0.032$ & $0.006$ & $9.392\times10^{-8}$   & G & A \\
rs7117842  & $11$ & $122534504$ & $0.019$  & $0.004$ & $7.556\times10^{-7}$   & C & T \\
rs964184   & $11$ & $116648917$ & $-0.086$ & $0.008$ & $2.008\times10^{-26}$  & C & G \\
\bottomrule
\end{tabular}%
}
\begin{flushleft}
\footnotesize
\textit{Note:} SNP, single-nucleotide polymorphism; CHR, chromosome; POS, genomic position in GRCh37; BETA, estimated change in LDL cholesterol per copy of the effect allele (A1); SE, standard error of BETA; P-value, association P-value; A1, effect allele; A2, other allele.
\end{flushleft}
\end{table}

\begin{table}[h]
\caption{Detailed information on the associations of the selected SNPs with high-density lipoprotein (HDL) cholesterol. }
\centering 
\resizebox{\textwidth}{!}{%
\begin{tabular}{llllllcc}
\toprule
SNP & CHR & POS & BETA & SE & P-value & A1 & A2 \\
\midrule
rs12748152 & $1$  & $27138393$  & $-0.051$ & $0.006$ & $9.736\times10^{-16}$ & T & C \\
rs17630235 & $12$ & $112591686$ & $-0.021$ & $0.004$ & $6.079\times10^{-9}$  & A & G \\
rs1800961  & $20$ & $43042364$  & $-0.127$ & $0.010$ & $1.639\times10^{-34}$ & T & C \\
rs1883025  & $9$  & $107664301$ & $-0.070$ & $0.004$ & $1.496\times10^{-65}$ & T & C \\
rs2075650  & $19$ & $45395619$  & $-0.055$ & $0.005$ & $9.716\times10^{-26}$ & G & A \\
rs2642438  & $1$  & $220970028$ & $0.030$  & $0.004$ & $7.780\times10^{-14}$ & G & A \\
rs3800406  & $6$  & $35133074$  & $-0.035$ & $0.006$ & $2.515\times10^{-9}$  & G & A \\
rs7117842  & $11$ & $122534504$ & $0.027$  & $0.004$ & $1.057\times10^{-14}$ & C & T \\
rs964184   & $11$ & $116648917$ & $0.107$  & $0.007$ & $6.086\times10^{-48}$ & C & G \\
\bottomrule
\end{tabular}%
}
\begin{flushleft}
\footnotesize
\textit{Note:} SNP, single-nucleotide polymorphism; CHR, chromosome; POS, genomic position in GRCh37; BETA, estimated change in HDL cholesterol per copy of the effect allele (A1); SE, standard error of BETA; P-value, association P-value; A1, effect allele; A2, other allele.
\end{flushleft}
\end{table}

\subsection{Procedures}
This section illustrates the procedures of our SCI method for GWAS summary data. Recall that the summary data includes the estimators of $\hat\Gamma$ (association between SNPs and the outcome), $\hat \Upsilon_{\cdot 1}$ (association between SNPs and LDL),  $\hat \Upsilon_{\cdot 2}$ (association between SNPs and HDL), and their corresponding standard errors.  
Throughout this section, we assume that the SNPs are independent of each other, and the estimators $\hat\Gamma$, and $\hat \Upsilon_{\cdot 1}$, and $\hat \Upsilon_{\cdot 2}$ are independent of each other. Similar simplifications are adopted in the two-sample MR literature \citep{pierce2013efficient,bowden2015mendelian,yao2024deciphering}.
Even if the data of LDL and HDL are collected from the same cohort, we need the independence assumption as a pragmatic approximation due to the lack of available estimates of the covariance between HDL and LDL.

We need to reformulate the following estimators and test statistics:
\begin{enumerate}[(1)]
    \item The CD statistic in (\ref{eq: CD test}) to select just-identifying IV subsets. 
    \item The TSLS estimator $\hat\beta^{(\mathcal{H})}$ as in (\ref{eq: def generic TSLS}) using the IVs in $\mathcal{H}$, and its standard error $\SE(\hat\beta^{(\mathcal{H})})$.
    \item The estimator $\hat\pi^{(\ell)}_k$ as in (\ref{eq: hat pi H}) and its standard error $\SE(\hat\pi_k^{(\ell)})$.
\end{enumerate}
With these estimators formulated, the procedure of our SCI method for GWAS summary data follows Section \ref{sec: uniform inference} for observational data. 

{\bf (1) The CD statistic. } Under the assumptions that the SNPs are independent of each other and there are no covariates, we simplify the CD statistic in (\ref{eq: CD test}) as
\begin{equation}\label{eq: CD test simplified}
    \mathrm{CD}(\mathcal{H})
    = \lambda_{\min}(\tilde\Upsilon_{\mathcal{H}\cdot}^{\top}\tilde\Upsilon_{\mathcal{H}\cdot}),
\end{equation} 
where $\tilde\Upsilon = \hat\Upsilon / \max_{{k\in[p_z],j\in[p_d]}}\hat{\rm SE}(\hat\Upsilon_{k,j}) $ denotes the estimator of IV relevance standardized by the maximum standard error.

{\bf (2) The TSLS estimator $\hat\beta^{(\mathcal{H})}$ and its standard error.}
When all IVs in $\mathcal{H}$ are valid, by $\Gamma^*_{\mathcal{H}} = \Upsilon^*_{\mathcal{H}\cdot}\beta^* + \pi^*_{\mathcal{H}}$ and $\pi^*_{\mathcal{H}} = 0$, we have
\begin{equation}\label{eq: beta N definition}
    \beta^*
    =
    \left(
        \Upsilon_{\mathcal{H}\cdot}^{*\top}
        \Upsilon^*_{\mathcal{H}\cdot}
    \right)^{-1}
    \Upsilon_{\mathcal{H}\cdot}^{*\top}
    \Gamma^*_{\mathcal{H}}.
\end{equation}
We therefore use the estimator 
\begin{equation}\label{eq: beta N estimator}
    \widehat{\beta}^{(\mathcal{H})}
    =
    \left(
        \widehat \Upsilon_{\mathcal{H}\cdot}^{\top}
        \widehat \Upsilon_{\mathcal{H}\cdot}
    \right)^{-1}
    \widehat \Upsilon_{\mathcal{H}\cdot}^{\top}
    \widehat\Gamma_{\mathcal{H}}.
\end{equation}

We impose the following simplifying covariance structure. Suppose that
\begin{equation}\label{eq: Gamma asymptotic distribution}
     \widehat\Gamma-\Gamma^* 
    \overset{a}{\sim}
    N(0,\hat\Omega_{\Gamma}),
\end{equation}
where $\hat\Omega_{\Gamma} = {\rm diag}\left[\{(\hat{\rm SE}(\widehat\Gamma_k))^2\}_{k\in[p_z]}\right]$, and $\overset{a}{\sim}$ denotes asymptotic distributional equivalence. In addition,  for each $k\in[p_z] $ and $j\in[p_d]$ ($p_d$=2 in our application), we assume that
\begin{equation}\label{eq: A element asymptotic distribution}
   \left(\widehat \Upsilon_{k,j}-\Upsilon^*_{k,j}\right)
    \overset{a}{\sim}
    N(0,\hat\omega_{k,j}^2),
\end{equation}
where $\hat\omega_{k,j} = \hat{\rm SE}(\widehat \Upsilon_{k,j})$. Recall that we assume the elements of $\widehat \Upsilon$ are asymptotically mutually
independent and that $\widehat \Upsilon$ and $\widehat\Gamma$ are asymptotically
independent. 

Define
\begin{equation*}\label{eq: Q B residual N}
    \hat Q_{\mathcal{H}}
    =
    \hat\Upsilon_{\mathcal{H}\cdot}^{\top}\hat\Upsilon_{\mathcal{H}\cdot},
    \qquad
    \hat B_{\mathcal{H}}
    =
    \hat Q_{\mathcal{H}}^{-1}\hat\Upsilon_{\mathcal{H}\cdot}^{\top},
    \qquad
    \hat r_{\mathcal{H}}
    =
    \hat\Gamma_{\mathcal{H}}
    -
    \hat\Upsilon_{\mathcal{H}\cdot}\hat\beta^{(\mathcal{H})}.
\end{equation*} 
In the following we deduce the asymptotic variance of $\widehat\beta^{(\mathcal{H})}$  using the delta method.

\emph{Asymptotic variance of $\widehat\beta^{(\mathcal{H})}$.}
The first-order expansion of \eqref{eq: beta N estimator} is
\begin{align}
    d\hat\beta^{(\mathcal{H})}
    ={}&
    \hat B_{\mathcal{H}}\,d\hat\Gamma_{\mathcal{H}}
    +
    \hat Q_{\mathcal{H}}^{-1}
    (d\hat\Upsilon_{\mathcal{H}\cdot})^{\top}\hat r_{\mathcal{H}}
    -
    \hat B_{\mathcal{H}}
    (d\hat\Upsilon_{\mathcal{H}\cdot})\hat \beta^{(\mathcal{H})}.
    \label{eq: beta N differential}
\end{align}
For $k\in\mathcal{H}$, let $b(k)$ denote the position of $k$ within
$\mathcal{H}$. The derivative of $\beta^{(\mathcal{H})}$ with respect to
$\hat\Upsilon_{k,j}$ is
\begin{equation}\label{eq: beta gradient A}
    \hat g_{k,j}
    :=
    \frac{\partial\hat\beta^{(\mathcal{H})}}{\partial \hat\Upsilon_{k,j}}
    =
    \hat r_{\mathcal{H},b(k)}
    \left(\hat Q_{\mathcal{H}}^{-1}\right)_{\cdot j}
    -
    \hat\beta^{(\mathcal{H})}_{j}
    \left(\hat B_{\mathcal{H}}\right)_{\cdot b(k)},
    \qquad k\in\mathcal{H}.
\end{equation}
It follows from the delta method that
\begin{equation*}\label{eq: beta N asymptotic distribution}
    \hat{\rm SE}(\hat\beta_1^{(\mathcal{H})}) = \sqrt{(\hat V_{\mathcal{H}})_{1,1}},
\end{equation*}
where
\begin{equation*}\label{eq: beta N asymptotic variance}
    \hat V_{\mathcal{H}}
    =
    \hat B_{\mathcal{H}}
    \hat \Omega_{\Gamma,\mathcal{H}\mathcal{H}}
    \hat B_{\mathcal{H}}^{\top}
    +
    \sum_{k\in\mathcal{H}}
    \sum_{j=1}^{p_d}
    \hat\omega_{k,j}^{2}\hat g_{k,j}\hat g_{k,j}^{\top}.
\end{equation*}

{\bf (3) The estimator $\hat\pi^{(\ell)}$
and its standard error.} 
With $\hat\Gamma$, $\hat\Upsilon$, and the estimator $\hat\beta^{(\mathcal{H})}$, the estimator $\hat\pi^{(\ell)}$ follows (\ref{eq: hat pi H}), given as $\hat\pi^{(\ell)} = \hat\Gamma - \hat\Upsilon\hat\beta^{(\hat{\mathcal{H}}_\ell)}$. We then derive the asymptotic variance of $\hat\pi^{(\ell)}$ using the delta method.
\emph{Asymptotic variance of $\widehat\pi^{(\ell)}$.} Recall that in the sampled estimator (\ref{eq: sampling pi}) we only need the standard error of $\hat\pi^{(\ell)}_k$ for $k\in{\hat{\mathcal{H}}_\ell^c}$. We therefore focus on the asymptotic variance of $\hat\pi^{(\ell)}_{{\hat{\mathcal{H}}_\ell^c}}$. Note that 
\begin{equation}\label{eq: pi S estimator}
    \hat\pi^{(\ell)}_{{\hat{\mathcal{H}}_\ell^c}} = \hat\Gamma_{{\hat{\mathcal{H}}_\ell^c}} - \hat\Upsilon_{{\hat{\mathcal{H}}_\ell^c}\cdot}\hat\beta^{(\hat{\mathcal{H}}_\ell)}.
\end{equation}
The first-order expansion of \eqref{eq: pi S estimator} satisfies
\begin{equation}\label{eq: pi S differential initial}
    d\hat\pi^{(\ell)}_{{\hat{\mathcal{H}}_\ell^c}}
    =
    d\hat \Gamma_{{\hat{\mathcal{H}}_\ell^c}}
    -
    (d\hat \Upsilon_{{\hat{\mathcal{H}}_\ell^c}\cdot})\hat \beta^{(\hat{\mathcal{H}}_\ell)}
    -
    \hat \Upsilon_{{\hat{\mathcal{H}}_\ell^c}\cdot}\,d\hat \beta^{(\hat{\mathcal{H}}_\ell)}.
\end{equation}
Substituting \eqref{eq: beta N differential} gives
\begin{align}
d\hat\pi^{(\ell)}_{{\hat{\mathcal{H}}_\ell^c}}
    ={}&
    d\hat \Gamma_{{\hat{\mathcal{H}}_\ell^c}}
    -
    \hat \Upsilon_{{\hat{\mathcal{H}}_\ell^c}\cdot}
    \hat B_{\hat{\mathcal{H}}_\ell}\,d\hat \Gamma_{\hat{\mathcal{H}}_\ell}
    -
    (d\hat \Upsilon_{{\hat{\mathcal{H}}_\ell^c}\cdot})\hat \beta^{(\hat{\mathcal{H}}_\ell)}
    \nonumber\\
    &-
    \hat \Upsilon_{{\hat{\mathcal{H}}_\ell^c}\cdot}
    \hat Q_{\hat{\mathcal{H}}_\ell}^{-1}
    (d\hat \Upsilon_{\hat{\mathcal{H}}_\ell\cdot})^{\top}\hat r_{\hat{\mathcal{H}}_\ell}
    +
    \hat \Upsilon_{{\hat{\mathcal{H}}_\ell^c}\cdot}
    \hat B_{\hat{\mathcal{H}}_\ell}
    (d\hat \Upsilon_{\hat{\mathcal{H}}_\ell\cdot})\hat \beta^{(\hat{\mathcal{H}}_\ell)}.
    \label{eq: pi S differential}
\end{align}

Define
\begin{equation}\label{eq: pi Gamma Jacobian}
    \hat J_{\Gamma,{\hat{\mathcal{H}}_\ell^c}}
    =
    \begin{bmatrix}
        {\rm I}_{|{\hat{\mathcal{H}}_\ell^c}|}
        &
        -\hat \Upsilon_{{\hat{\mathcal{H}}_\ell^c}\cdot}\hat B_{\hat{\mathcal{H}}_\ell}
    \end{bmatrix}.
\end{equation}
The contribution of $\widehat\Gamma$ to the asymptotic variance of
$\hat\pi^{(\ell)}_{{\hat{\mathcal{H}}_\ell^c}}$ is therefore
\begin{equation}\label{eq: pi variance Gamma}
    \hat V_{\pi,\Gamma}
    =
    \hat J_{\Gamma,{\hat{\mathcal{H}}_\ell^c}}
   \hat \Omega_{\Gamma,({\hat{\mathcal{H}}_\ell^c},\hat{\mathcal{H}}_\ell)}
    \hat J_{\Gamma,{\hat{\mathcal{H}}_\ell^c}}^{\top},
\end{equation}
where $\hat \Omega_{\Gamma,({\hat{\mathcal{H}}_\ell^c},\hat{\mathcal{H}}_\ell)}$ reorders $\hat \Omega_{\Gamma}$ defined below (\ref{eq: Gamma asymptotic distribution}).

For $k\in{\hat{\mathcal{H}}_\ell^c}$, let $a(k)$ denote the position of $k$ within
${\hat{\mathcal{H}}_\ell^c}$, and let $\iota_{a(k)}$ be the corresponding unit vector in
$\mathbb{R}^{|{\hat{\mathcal{H}}_\ell^c}|}$. The derivative of $\hat\pi^{(\ell)}_{{\hat{\mathcal{H}}_\ell^c}}$
with respect to $\hat \Upsilon_{k,j}$ is
\begin{equation}\label{eq: pi gradient AS}
    \hat h_{k,j}^{{\hat{\mathcal{H}}_\ell^c}}
    :=
    \frac{\partial\hat\pi^{(\ell)}_{{\hat{\mathcal{H}}_\ell^c}}}{\partial \hat \Upsilon_{k,j}}
    =
    -\hat \beta^{(\hat{\mathcal{H}}_\ell)}_{j}\iota_{a(k)},
    \qquad k\in{\hat{\mathcal{H}}_\ell^c}.
\end{equation}
For $k\in\hat{\mathcal{H}}_\ell$, the corresponding derivative is
\begin{equation}\label{eq: pi gradient AN}
    \hat h_{k,j}^{\hat{\mathcal{H}}_\ell}
    :=
    \frac{\partial\hat\pi^{(\ell)}_{{\hat{\mathcal{H}}_\ell^c}}}{\partial \hat \Upsilon_{k,j}}
    =
    -\hat \Upsilon_{{\hat{\mathcal{H}}_\ell^c}\cdot}\hat g_{k,j},
    \qquad k\in\hat{\mathcal{H}}_\ell,
\end{equation}
where $\hat g_{k,j}$ is defined in \eqref{eq: beta gradient A}.

We assume there exists a vector $\pi^{(\ell)}$ such that 
\begin{equation}\label{eq: pi S asymptotic distribution} 
        \hat\pi^{(\ell)}_{{\hat{\mathcal{H}}_\ell^c}}
        -
        \pi^{(\ell)}_{{\hat{\mathcal{H}}_\ell^c}}
    \overset{a}{\sim}
    N\left(0,\hat V_{\pi,{\hat{\mathcal{H}}_\ell^c}}\right),
\end{equation}
where
\begin{align}
    \hat V_{\pi,{\hat{\mathcal{H}}_\ell^c}}
    ={}&
    \hat J_{\Gamma,{\hat{\mathcal{H}}_\ell^c}}
    \hat \Omega_{\Gamma,({\hat{\mathcal{H}}_\ell^c},\hat{\mathcal{H}}_\ell)}
    \hat J_{\Gamma,{\hat{\mathcal{H}}_\ell^c}}^{\top}
    \nonumber\\
    &+
    \sum_{k\in{\hat{\mathcal{H}}_\ell^c}}
    \sum_{j=1}^{p_d}
    \hat \omega_{k,j}^{2}
    \hat h_{k,j}^{{\hat{\mathcal{H}}_\ell^c}}
    \left(\hat h_{k,j}^{{\hat{\mathcal{H}}_\ell^c}}\right)^{\top}
    \nonumber\\
    &+
    \sum_{k\in\hat{\mathcal{H}}_\ell}
    \sum_{j=1}^{p_d}
    \hat \omega_{k,j}^{2}
    \hat h_{k,j}^{\hat{\mathcal{H}}_\ell}
    \left(\hat h_{k,j}^{\hat{\mathcal{H}}_\ell}\right)^{\top}.
    \label{eq: pi S asymptotic variance}
\end{align}
Accordingly, the estimated standard error of
$\hat\pi^{(\ell)}_k$, for $k\in{\hat{\mathcal{H}}_\ell^c}$, is
\begin{equation}\label{eq: pi component standard error}
    \widehat{\mathrm{SE}}(\hat\pi^{(\ell)}_k)
    =
    \sqrt{ 
            \left(\widehat V_{\pi,{\hat{\mathcal{H}}_\ell^c}}\right)_{a(k),a(k)}
     }.
\end{equation}